\documentclass[]{imsart}

\RequirePackage{amsmath,amssymb,amsfonts,amsthm,enumitem,aliascnt,graphicx,xcolor,multirow,array,url}
\RequirePackage[numbers]{natbib}
\RequirePackage[colorlinks,citecolor=blue,urlcolor=blue]{hyperref}

\startlocaldefs

\numberwithin{equation}{section}

\theoremstyle{plain}

\newtheorem{theorem}{Theorem}[section]

\newaliascnt{corollary}{theorem}
\newtheorem{corollary}[corollary]{Corollary}
\aliascntresetthe{corollary}

\newaliascnt{lemma}{theorem}
\newtheorem{lemma}[lemma]{Lemma}
\aliascntresetthe{lemma}

\newaliascnt{proposition}{theorem}

\aliascntresetthe{proposition}

\definecolor{black}{RGB}{0,0,0}
\definecolor{blue}{RGB}{0,0,255}
\definecolor{darkgreen}{RGB}{0,100,0}
\definecolor{red}{RGB}{255,0,0}
\definecolor{brown}{RGB}{165,42,42}
\definecolor{darkorchid4}{RGB}{104,34,139}
\definecolor{darkslategray}{RGB}{47,79,79}
\definecolor{darkolivegreen}{RGB}{85,107,47}
\definecolor{indigo}{RGB}{75,0,130}

\renewcommand{\mid}{\,\middle\arrowvert\,}
\newcommand{\midil}{\,|\,}

\endlocaldefs

\begin{document}

\begin{frontmatter}

% "Title of the Paper"
\title{Multi-sample Comparison Using Spatial Signs for Infinite Dimensional Data}
\runtitle{Multi-sample Comparison Using Spatial Signs}

\author{\fnms{Joydeep} \snm{Chowdhury}\thanksref{t1}\corref{}\ead[label=e1]{joydeepchowdhury01@gmail.com}}

\address{King Abdullah University of Science and Technology, Thuwal, Saudi Arabia, \printead{e1}}

\author{\fnms{Probal} \snm{Chaudhuri}\ead[label=e2]{probal@isical.ac.in}}

\address{Indian Statistical Institute, Kolkata, India, \printead{e2}}

\thankstext{t1}{Corresponding author}

% indicate corresponding author with \corref{}
% \author{\fnms{John} \snm{Smith}\thanksref{t1}\corref{}\ead[label=e1]{smith@foo.com}\ead[label=e2,url]{www.foo.com}}
% \thankstext{t1}{Thanks to somebody} 
% \address{line 1\\ line 2\\ \printead{e1}\\ \printead{e2}}

%\author{\fnms{???} \snm{???}\ead[label=e1]{???}}
%\address{\printead{e1}}
%\and
%\author{\fnms{???} \snm{???}\ead[label=e2]{???}}
%\address{\printead{e2}}

\runauthor{J. Chowdhury and P. Chaudhuri}

\begin{abstract}
We consider an analysis of variance type problem, where the sample observations are random elements in an infinite dimensional space. This scenario covers the case, where the observations are random functions. For such a problem, we propose a test based on spatial signs.
We develop an asymptotic implementation as well as a bootstrap implementation and a permutation implementation of this test and investigate their size and power properties.
We compare the performance of our test with that of several mean based tests of analysis of variance for functional data studied in the literature.
Interestingly, our test not only outperforms the mean based tests in several non-Gaussian models with heavy tails or skewed distributions, but in some Gaussian models also.
Further, we also compare the performance of our test with the mean based tests in several models involving contaminated probability distributions.
Finally, we demonstrate the performance of these tests in three real datasets: a Canadian weather dataset, a spectrometric dataset on chemical analysis of meat samples and a dataset on orthotic measurements on volunteers.
\end{abstract}



% history:
% \received{\smonth{1} \syear{0000}}

\tableofcontents

\end{frontmatter}

% Main text entry area

\section{Introduction} \label{sec:1}
Analysis of variance based on ranks of real valued observations was studied by \cite{kruskal1952use}.
Suppose that we have $ n = \sum_{k=1}^{K} n_k $ independent real valued observations $ \{ X_{k i} : i = 1, \ldots, n_k;\, k = 1, \ldots, K \} $ from $ K $ groups.
We are interested in testing whether the distributions of the groups are same or not.
We rank all the observations together, and let $ r_{k i} $ be the rank of the $ i^\text{th} $ observation in the $ k^\text{th} $ group. Let $ \bar{r}_{k \cdot} = n_k^{-1} \sum_{i=1}^{n_k} r_{k i} $ and $ \bar{r}_{\cdot \cdot} = n^{-1} \sum_{k=1}^{K} n_k \bar{r}_{k \cdot} = 0.5 ( n + 1 ) $.
The test statistic for the Kruskal-Wallis test is
\begin{align*}
W_n = ( n - 1 ) \frac{\sum_{k=1}^{K} n_k ( \bar{r}_{k \cdot} - \bar{r}_{\cdot \cdot} )^2}{\sum_{k=1}^{K} \sum_{i=1}^{n_k} ( r_{k i} - \bar{r}_{\cdot \cdot} )^2} ,
\end{align*}
and a large value of $ W_n $ indicates that the distributions of the groups are not same.
When the groups have the same continuous distribution, the distribution of $ W_n $ is independent of the underlying distribution of the groups.
The usual mean based ANOVA test is developed assuming that the distributions of the groups are Gaussian. It is well-known that the Kruskal-Wallis test exhibits better performance than the mean based ANOVA test when the groups have non-Gaussian distributions with heavy tails.

In \cite{mottonen1995multivariate}, \cite{mottonen1997efficiency} and \cite{oja2010multivariate}, several nonparametric tests based on multivariate spatial signs and ranks were investigated.
In \cite{choi1997approach,choi2002multivariate}, spatial signs were employed to generalize the Kruskal-Wallis test in a multivariate analysis of variance setup.
Those spatial signs and ranks are intrinsically related to spatial quantiles investigated in \cite{chaudhuri1996geometric} and \cite{koltchinskii1997m}.
Consider independent multivariate observations $ \{ \mathbf{X}_{k i} : i = 1, \ldots, n_k;\, k = 1, \ldots, K \} $ from $ K $ groups, and $ n = \sum_{k=1}^{K} n_k $.
Define $ \mathbf{r}( \mathbf{x} ) = n^{-1} \sum_{k=1}^{K} \sum_{i=1}^{n_k} \mathbf{s}\left( \mathbf{x} - \mathbf{X}_{k i} \right) $ and $ \bar{\mathbf{r}}_{k \cdot} = n_k^{-1} \sum_{i=1}^{n_k} \mathbf{r}( \mathbf{X}_{k i} ) $, where $ \mathbf{s}\left( \mathbf{x} - \mathbf{y} \right) = \allowbreak \left( \mathbf{x} - \mathbf{y} \right) / \left\| \mathbf{x} - \mathbf{y} \right\| $ for $ \mathbf{x} \neq \mathbf{y} $ and $ \mathbf{s}\left( \mathbf{0} \right) = \mathbf{0} $.
Note that when the observations are univariate, $ \mathbf{r}\left( \mathbf{X}_{k i} \right) = n^{-1} ( 2 \text{rank}( \mathbf{X}_{k i} ) - ( n + 1 ) ) $, where $ \text{rank}( \mathbf{X}_{k i} ) $ is the rank of the univariate observation $ \mathbf{X}_{k i} $ within the pooled sample.
Let $ \mathbf{r}_k\left( \mathbf{x} \right) = n_k^{-1} \sum_{i=1}^{n_k} \mathbf{s}\left( \mathbf{x} - \mathbf{X}_{k i} \right) $ and $ \hat{\boldsymbol{\Sigma}}_n = ( n - K )^{-1} \sum_{k=1}^{K} \sum_{i=1}^{n_k} \mathbf{r}_k\left( \mathbf{X}_{k i} \right) \mathbf{r}_k\left( \mathbf{X}_{k i} \right)^t $.
The test statistic for the Choi-Marden one-way ANOVA test is
\begin{align*}
M_n = \sum_{k=1}^{K} n_k ( \bar{\mathbf{r}}_{k \cdot} )^t \hat{\boldsymbol{\Sigma}}_n^{-1} \bar{\mathbf{r}}_{k \cdot} ,
\end{align*}
which is used to test the equality of the distributions of the groups.
Unlike the Kruskal-Wallis test for univariate data, the Choi-Marden test is not distribution-free under the null hypothesis, and the authors implemented it using the asymptotic distribution of the test statistic.
But, like the Kruskal-Wallis test, the Choi-Marden test usually exhibits better powers and asymptotic relative efficiencies when the underlying distributions of the groups are non-Gaussian with heavy tails.
When the observations $ \mathbf{X}_{k i} $ are random functions, considered as random elements in an appropriate $ L_2 $ space, instead of finite dimensional random vectors, the Choi-Marden one-way ANOVA test cannot be directly applied for them. This is because in this case, the underlying space is an infinite dimensional Hilbert space, and the estimated covariance operator $ \hat{\boldsymbol{\Sigma}}_n $ in an infinite dimensional space is not invertible and cannot be used to standardize the test statistic. Methodology for functional data, which are infinite dimensional in nature, and also high dimensional data requires non-involvement of the inverse of the sample covariance (see, e.g., \cite{kong2020high,harrar2022recent} for some methods in high dimensional data developed without using the inverse of the sample covariance).

Analysis of variance for functional data has been investigated by several authors.
An ANOVA test for functional data based on the $ L_2 $ distance between the group means in the sample was proposed in \cite{cuevas2004anova}.
In \cite{shen2004f}, \cite{zhang2007statistical} and \cite{zhang2011statistical}, hypothesis testing in a linear model involving functional responses was investigated.
In \cite{causeur2019functional}, a similar testing procedure was described in functional linear models and ANOVA, and demonstrated its utility in investigating brain electrical activity.
In \cite{zhang2014one}, a test of ANOVA was described, where the test statistic is obtained by integrating the pointwise F-statistic for the functional observations.
In \cite{zhang2018new}, another test of ANOVA was proposed, where the test statistic is the supremum of the pointwise F-statistic.
In \cite{horvath2015introduction}, a test of ANOVA was introduced, where the functional observations are projected on a number of principal components of the sample covariance operator, rendering the test as a test for finite dimensional observations.
In \cite{gorecki2015comparison}, a permutation test of ANOVA was described for functional data. In \cite{cuesta2010simple}, a test of ANOVA for functional data was described based on random projections.
In \cite{aristizabal2019analysis}, the procedures by \cite{shen2004f} and \cite{cuesta2010simple} was adapted in an ANOVA problem for spatially correlated functional data.
In \cite{guo20192}, an $ L_2 $-norm based ANOVA procedure was used, which is similar to those proposed by \cite{shen2004f} and \cite{zhang2007statistical}, for weakly dependent functional time series.
All these tests in the literature involving functional data are based on the mean of the response, as in the classical univariate and multivariate ANOVA problems.
Recently, an ANOVA procedure was proposed in \cite{shinohara2019distance} based on the pairwise distances of the observations.

In this paper, we construct a test for one-way ANOVA for infinite dimensional data based on spatial signs, and describe implementations of this test based on the asymptotic distribution as well as a permutation procedure and a bootstrap procedure. We found that this test exhibits superior performance than the mean based tests in non-Gaussian heavy-tail processes, contaminated Gaussian and non-Gaussian processes and skewed processes. In the case of some Gaussian models also, our test outperforms the mean based tests.
We demonstrate the usefulness of our testing procedure in simulated and real datasets.

\section{A test based on spatial signs} \label{sec:2}
In the context of functional data analysis, sample observations are considered as real valued random functions defined over some domain, which may be an interval, a rectangle, or some other set. The space of such functions can be considered as an $ L_2 $ space by equipping it with an appropriate $ L_2 $-norm. It is well known that such an $ L_2 $ space will be a separable Hilbert space when it is associated with a Borel $ \sigma $-field and a $ \sigma $-finite measure. In other words, one can model those random functions as random elements in a separable Hilbert space. The advantage of this approach is that the methodology is readily applicable to diverse types of data that can be viewed as elements in a separable Hilbert space, and we develop an ANOVA procedure for such data.

We consider a sample of $ n $ independent observations in a separable Hilbert space $ \mathcal{H} $, which is divided into $ K $ groups with the $ k^\text{th} $ group having $ n_k $ observations for $ k = 1, \ldots, K $. The members of the $ k^\text{th} $ group are denoted as $ \mathbf{X}_{k i} $ for $ i = 1, \ldots, n_k $.
We assume that all the observations are independent, and for each of the groups, the observations within that group have identical distribution.
Let the probability distribution of the $ k^\text{th} $ group be denoted as $ P_k $. We are interested in testing whether all the distributions $ P_k $ are the same or not. So, the null hypothesis is
\begin{align} \label{h_0}
\text{H}_0 : \; P_1 = P_2 = \cdots = P_K \;.
%\quad \text{against} \quad
%\text{H}_\text{A} : \; \text{$ P_k $'s are not all same}.
\end{align}
%We do not make any restrictions on the distributions $ P_k $, except that they are non-atomic probability distributions on the separable Hilbert space $\mathcal{H}$ equipped with the Borel $ \sigma $-field. In particular, ties between observations are allowed in our setup. Also, the treatment effects are assumed to alter the underlying distributions of the groups. The literature on ANOVA for functional data described in \autoref{sec:1} is usually concerned with treatment effects which specifically alter the means of the groups. Our setup, which is inspired from the setups considered in \cite{kruskal1952use} for univariate observations and in \cite{choi1997approach} for multivariate observations, covers the case where the treatment effects alter the means of the underlying distributions of the groups, as well as cases where the treatment effects alter the underlying the distributions of the groups but not necessarily the means of the groups.
We assume that the distributions $ P_k $ are non-atomic probability distributions on the separable Hilbert space $\mathcal{H}$ equipped with the Borel $ \sigma $-field. Also, the treatment effects are assumed to alter the underlying distributions of the groups. The literature on ANOVA for functional data described in \autoref{sec:1} is usually concerned with treatment effects which specifically alter the means of the groups. Our setup, which is inspired from the setups considered in \cite{kruskal1952use} for univariate observations and in \cite{choi1997approach} for multivariate observations, covers the case where the treatment effects alter the means of the underlying distributions of the groups, as well as cases where the treatment effects alter the underlying the distributions of the groups but not necessarily the means of the groups.

Define $ \mathbf{s}( \mathbf{x} ) = \| \mathbf{x} \|^{-1} \mathbf{x} $ for $ \mathbf{x} \neq \mathbf{0} $ and $ \mathbf{s}( \mathbf{0} ) = \mathbf{0} $.
Consider the quantity $ E[ \mathbf{s}( \mathbf{X}_{k 1} - \mathbf{X}_{l 1} ) ] $, which is the expected direction vector from the random element $ \mathbf{X}_{l 1} $ to the random element $ \mathbf{X}_{k 1} $. If the distributions $ P_k $ and $ P_l $ are same, then we have $ E[ \mathbf{s}( \mathbf{X}_{k 1} - \mathbf{X}_{l 1} ) ] = E[ \mathbf{s}( \mathbf{X}_{l 1} - \mathbf{X}_{k 1} ) ] = - E[ \mathbf{s}( \mathbf{X}_{k 1} - \mathbf{X}_{l 1} ) ] = \mathbf{0} $.
Define
$ \mathbf{R}( \mathbf{x} ) = n^{-1} \sum_{k=1}^{K} \sum_{i=1}^{n_k} \mathbf{s}( \mathbf{x} - \mathbf{X}_{k i} ) $.
Note that $ \mathbf{R}( \mathbf{x} ) $ is the average spatial sign of the point $ \mathbf{x} $ with respect to the combined sample, which is the average of all direction vectors from the data points to $ \mathbf{x} $.
$ \mathbf{R}( \mathbf{x} ) $ can also be viewed as the spatial rank of $ \mathbf{x} $ in the combined sample of all the observations (cf. \cite{choi1997approach}, \cite{mottonen1997efficiency}).
Also define $ \bar{\mathbf{R}}_k = n_k^{-1} \sum_{i = 1}^{n_k} \mathbf{R}\left( \mathbf{X}_{k i} \right) $.
Note that $ E[ \bar{\mathbf{R}}_k ] = n^{-1} \sum_{l=1}^{K} n_l E[ \mathbf{s}( \mathbf{X}_{k 1} - \mathbf{X}_{l 1} ) ] = \mathbf{0} $ for all $ k $ under $ \text{H}_0 $ in \eqref{h_0}. So, if one is interested to test $ \text{H}_0 $, the test may be carried out based on the magnitudes of $ \bar{\mathbf{R}}_k $, $ k = 1, \ldots, K $. A high magnitude will indicate that $ \text{H}_0 $ in \eqref{h_0} is not true.

Define $ \mathbf{U}_n = \left( \sqrt{n_1} \bar{\mathbf{R}}_1, \ldots, \sqrt{n_K} \bar{\mathbf{R}}_K \right) $.
Note that $ \mathbf{U}_n $ is a random element in the product Hilbert space $ \mathcal{H} \times \mathcal{H} \times \cdots \times \mathcal{H} = \mathcal{H}^K $.
Clearly, $ E[ \mathbf{U}_n ] = \mathbf{0} $ under the null hypothesis $ \text{H}_0 $ in \eqref{h_0}.
Based on this idea discussed above, we define the test statistic
\begin{align*}
SS_n = \| \mathbf{U}_n \|^2 = \sum_{k=1}^{K} n_k \left\| \bar{\mathbf{R}}_k \right\|^2 .
\end{align*}
We shall reject $ \text{H}_0 $ in \eqref{h_0} when $ SS_n $ is significantly large. We call the resulting test the SS test.
The normalization using the $ \sqrt{n_k} $ terms in the definition of $ \mathbf{U}_n $ is done so that its asymptotic distribution is non-degenerate, as described in \autoref{thm:1}.
The asymptotic distribution of $ SS_n $ under $ \text{H}_0 $ in \eqref{h_0} is obtained from this theorem in \autoref{coro:1}, which is used to derive the asymptotic implementation of the test.

$ \bar{\mathbf{R}}_k $ can be expressed as a weighted sum of sample means of spatial signs.
Define $ \boldsymbol{\nu}_{kl} = E\left[ \mathbf{s}( \mathbf{X}_{k1} - \mathbf{X}_{l1} ) \right] $ and $ \widehat{\boldsymbol{\nu}}_{kl} = (n_k n_l)^{-1} \sum_{i_k = 1}^{n_k} \sum_{i_l = 1}^{n_l} \left( \mathbf{X}_{k i_k} - \mathbf{X}_{l i_l} \right) $.
Let $ \widehat{\boldsymbol{\nu}} = \left( \widehat{\boldsymbol{\nu}}_{11}, \ldots, \widehat{\boldsymbol{\nu}}_{KK} \right)' $ be the vector containing all the sample means of spatial signs. Let $ \mathbf{0}_K $ denote the vector of $ K $ null elements, and let $ \mathbf{a}_k = ( \mathbf{z}_1', \ldots, \mathbf{z}_{k}', \ldots, \mathbf{z}_{K}' )' $, where $ \mathbf{z}_i = \mathbf{0}_K $ for all $ i \neq k $, and $ \mathbf{z}_{k} = n^{-1} ( n_1, \ldots, n_K )' $. So, for each $ k $, $ \mathbf{a}_k $ is a $ K^2 $-dimensional vector of weights whose sum is 1. Then, $ \sqrt{n_k} \bar{\mathbf{R}}_k = ( \sqrt{n_k / n} ) \mathbf{a}_k' \sqrt{n} \widehat{\boldsymbol{\nu}} $ for all $ k $.
Define the matrix of weights $ \mathbf{A}_n = \left( ( \sqrt{n_1 / n} ) \mathbf{a}_1, \ldots, ( \sqrt{n_K / n} ) \mathbf{a}_K \right)' $. We have $ \mathbf{U}_n = \mathbf{A}_n \sqrt{n} \widehat{\boldsymbol{\nu}} $. It can be verified that when $ n^{-1} n_k \to \lambda_k \in ( 0, 1 ) $ for all $ k $ as $ n \to \infty $, $ \mathbf{A}_n $ converges to a non-null matrix, say, $ \mathbf{A} $, and the covariance operator of $ \sqrt{n} \widehat{\boldsymbol{\nu}} $ also converges, say, to $ \boldsymbol{\Gamma} $. Then, the covariance operator of $ \mathbf{U}_n $ converges to $ \mathbf{A} \boldsymbol{\Gamma} \mathbf{A}' $ as $ n \to \infty $. Note that $ \mathbf{A} \boldsymbol{\Gamma} \mathbf{A}' $ is a $ K \times K $ matrix of covariance operators. The $ (k_1, k_2) $th element of this matrix is $ \boldsymbol{\sigma}_{k_1 k_2} $, which is described in \autoref{thm:1}. A similar approach for deriving the asymptotic distribution of a rank-based statistic for univariate data was developed in \cite{brunner2017rank}.

\begin{theorem} \label{thm:1}
Assume that $ n^{-1} n_k \to \lambda_k \in ( 0, 1 ) $ for all $ k $ as $ n \to \infty $.
Let $ \mathbf{X}_i $, $ \mathbf{Y}_j $ and $ \mathbf{Z}_k $ be independent random elements having the distributions of the $ i^\text{th} $, the $ j^\text{th} $ and the $ k^\text{th} $ group of the sample, respectively.
Define
$ \mathbf{C}( i, j, k ) = \text{Cov}\left( E\left[ \mathbf{s}\left( \mathbf{X}_i - \mathbf{Z}_k \right) \mid \mathbf{Z}_k \right] ,\, E\left[ \mathbf{s}\left( \mathbf{Y}_j - \mathbf{Z}_k \right) \mid \mathbf{Z}_k \right] \right) $ and $
\boldsymbol{\Sigma} = \left( \boldsymbol{\sigma}_{k_1 k_2} \right)_{K \times K} $, where
\begin{align*}
\boldsymbol{\sigma}_{k_1 k_2}
& = \sqrt{\lambda_{k_1} \lambda_{k_2}} \sum_{l=1}^{K} \lambda_l
\left[ \mathbf{C}( k_1, k_2, l ) - \mathbf{C}( l, k_2, k_1 ) - \mathbf{C}( k_1, l, k_2 ) \right] \\
& \quad + \sum_{l_1 = 1}^{K} \sum_{l_2 = 1}^{K} \lambda_{l_1} \lambda_{l_2} \mathbf{C}( l_1, l_2, k_1 ) \mathbb{I}( k_1 = k_2 ) .
\end{align*}
Then, $ \left[ \mathbf{U}_n - E\left[ \mathbf{U}_n \right] \right] \stackrel{w}{\longrightarrow} G\left( \mathbf{0}, \boldsymbol{\Sigma} \right) $ as $ n \to \infty $, where $ G\left( \mathbf{0}, \boldsymbol{\Sigma} \right) $ denotes a Gaussian random element in $ \mathcal{H}^K $ with mean $ \mathbf{0} $ and covariance operator $ \boldsymbol{\Sigma} $.
\end{theorem}
\begin{corollary} \label{coro:1}
Assume that $ n^{-1} n_k \to \lambda_k \in ( 0, 1 ) $ for all $ k $ as $ n \to \infty $.
Under $ \text{H}_0 $ in \eqref{h_0}, we have $ E\left[ \mathbf{U}_n \right] = \mathbf{0} $.
Consequently, $ \mathbf{U}_n \stackrel{w}{\longrightarrow} G\left( \mathbf{0}, \boldsymbol{\Sigma} \right) $ and $ SS_n = \| \mathbf{U}_n \|^2 \stackrel{w}{\longrightarrow} \| \mathbf{W} \|^2 $ as $ n \to \infty $, where $ \mathbf{W} $ is a random element having distribution $ G\left( \mathbf{0}, \boldsymbol{\Sigma} \right) $.
\end{corollary}

The p-value of the SS test in the asymptotic implementation is given by $ P[ \| \mathbf{W} \|^2 \ge \text{the observed value of } SS_n ] $.
The p-value can be numerically computed using a Monte Carlo procedure. However, since $ \boldsymbol{\Sigma} $ is unknown, to implement the test, we need to estimate the covariance operator $ \boldsymbol{\Sigma} $ from the sample. Let $ \otimes $ be the outer product on $ \mathcal{H} $ such that for every $ \mathbf{x}, \mathbf{y} \in \mathcal{H} $, $ \mathbf{x} \otimes \mathbf{y} : \mathcal{H} \to \mathcal{H} $ is a linear operator defined by $ ( \mathbf{x} \otimes \mathbf{y} ) ( \mathbf{w} ) = \langle \mathbf{w}, \mathbf{x} \rangle \mathbf{y} $.
Note that for $ i, j, k \in \{ 1, \ldots, K \} $,
\begin{align*}
\mathbf{C}( i, j, k )
& = E\left[ E\left[ \mathbf{s}( \mathbf{X}_i - \mathbf{Z}_k ) \midil \mathbf{Z}_k \right] \otimes E\left[ \mathbf{s}( \mathbf{Y}_j - \mathbf{Z}_k ) \midil \mathbf{Z}_k \right] \right] \\
& \quad
- E\left[ \mathbf{s}( \mathbf{X}_i - \mathbf{Z}_k ) \right] \otimes E\left[ \mathbf{s}( \mathbf{Y}_j - \mathbf{Z}_k ) \right] ,
\end{align*}
where $ \mathbf{X}_i $, $ \mathbf{Y}_j $ and $ \mathbf{Z}_k $ are independent random elements having the distributions of the $ i^\text{th} $, the $ j^\text{th} $ and the $ k^\text{th} $ group, respectively.
Let
\begin{align*}
& \mathbf{C}_n( i, j, k ) \\
& =
\frac{1}{n_k - 1} \sum_{l_k = 1}^{n_k} \left[ \left( \frac{1}{n_{i}} \sum_{l_{i} = 1}^{n_{i}} \mathbf{s}\left( \mathbf{X}_{i l_{i}} - \mathbf{X}_{k l_k} \right) \right) \otimes \left( \frac{1}{n_{j}} \sum_{l_{j} = 1}^{n_{j}} \mathbf{s}\left( \mathbf{X}_{j l_{j}} - \mathbf{X}_{k l_k} \right) \right) \right] \nonumber\\
& \quad
- \left( \frac{1}{n_{i} n_k} \sum_{l_{i} = 1}^{n_{i}} \sum_{l_k = 1}^{n_k} \mathbf{s}\left( \mathbf{X}_{i l_{i}} - \mathbf{X}_{k l_k} \right) \right) \otimes \left( \frac{1}{n_{j} n_k} \sum_{l_{j} = 1}^{n_{j}} \sum_{l_k = 1}^{n_k} \mathbf{s}\left( \mathbf{X}_{j l_{j}} - \mathbf{X}_{k l_k} \right) \right) ,
\end{align*}
where $ i, j, k \in \{ 1, \ldots, K \} $.
The estimate of $ \boldsymbol{\Sigma} $ is defined as
\begin{align*}
\widehat{\boldsymbol{\Sigma}}_n & = \left( \boldsymbol{\sigma}^{(n)}_{k_1 k_2} \right)_{K \times K} ,
\text{ where} \\
\boldsymbol{\sigma}^{(n)}_{k_1 k_2}
& = \frac{\sqrt{n_{k_1} n_{k_2}}}{n} \sum_{l=1}^{K} \frac{n_l}{n}
\left[ \mathbf{C}_n( k_1, k_2, l ) - \mathbf{C}_n( l, k_2, k_1 ) - \mathbf{C}_n( k_1, l, k_2 ) \right] \\
& \quad
+ \sum_{l_1 = 1}^{K} \sum_{l_2 = 1}^{K} \frac{n_{l_1} n_{l_2}}{n^2} \mathbf{C}_n( l_1, l_2, k_1 ) \mathbb{I}( k_1 = k_2 ) .
\end{align*}

Using the estimated covariance operator $ \widehat{\boldsymbol{\Sigma}}_n $, we generate independent observations from the distribution $ G\left( \mathbf{0}, \widehat{\boldsymbol{\Sigma}}_n \right) $. Let us denote the generated observations as $ \mathbf{W}_1, \ldots, \mathbf{W}_N $, where $ N $ is a suitably large number. Then the approximate p-value of the SS test is given by the proportion of $ \| \mathbf{W}_i \|^2 $ values greater than or equal to the observed value of $ SS_n $. A similar estimation procedure for the p-value is described in \cite{cuevas2004anova} for their test of ANOVA for functional data.

For the asymptotic validity of this procedure, the estimated covariance operator needs to asymptotically consistent, which is established in the following theorem.
\begin{theorem} \label{thm:cov}
Assume that $ n^{-1} n_k \to \lambda_k \in ( 0, 1 ) $ for all $ k $ as $ n \to \infty $. Then, $ \widehat{\boldsymbol{\Sigma}}_n \longrightarrow \boldsymbol{\Sigma} $ almost surely as $ n \to \infty $ in the operator norm.
\end{theorem}

%We implement the SS test based on \autoref{coro:1}, when each of the groups have reasonably large number of observations. Note that in the definition of $ \mathbf{C}_n( i, j, k ) $, we have normalized the sum in the first term by $ (n_k - 1) $ rather than $ n_k $. Though both of these would yield asymptotically consistent estimates of the covariance operator $ \boldsymbol{\Sigma} $, we prefer the normalization by $ (n_k - 1) $ as we have observed in our simulations that the resulting estimated size of the test from other normalization is higher than the nominal level of the test, when either one of the subgroups is relatively small in size (less than 20), or the total sample size is less than 100. Normalization by $ (n_k - 1) $ is found to mitigate the problem, and the asymptotic testing procedure yields acceptable sizes for when all the subgroup sizes are at least 20. This normalization is inspired by the unbiased estimate of the variance of a real valued random variable. Note that none of the normalizations discussed above yield an unbiased estimate of $ \boldsymbol{\Sigma} $, and computing the unbiased estimate of $ \boldsymbol{\Sigma} $ is found to be computationally very intensive even for moderate sample sizes.

When the number of observations in any of the groups is small, the asymptotic null distribution of the test statistic described in \autoref{coro:1} may not be a good approximation of its finite sample distribution. 
%For such cases, we can implement the test based on either a bootstrap procedure or a permutation procedure.
For such cases, we can implement the test based on a bootstrap procedure (see, e.g., \cite{efron1994introduction}) or a permutation procedure.

%We first describe the bootstrap implementation (see, e.g., \cite{efron1994introduction}).
%Bootstrap procedures in linear regression analysis and mean based ANOVA problems involving functional data were described in \cite{faraway1997regression}, \cite{zhang2013analysis} and \cite{zhang2018new}.
When the null hypothesis is true, the distributions of all the groups are same, and the natural estimator of that common underlying distribution is the empirical distribution of the pooled sample of all the groups.
So, in our bootstrap procedure, we draw a random sample of size $ n $ with replacement from the pooled sample and assign them to the groups so that given the original sample $ \{ \mathbf{X}_{k i} \midil k = 1, \ldots, K; i = 1, \ldots, n_k \} $, the observations in the bootstrap sample $ \{ \mathbf{X}_{k i}^* \midil k = 1, \ldots, K; i = 1, \ldots, n_k \} $ are independent and identically distributed, and for every $ k $ and $ i $, we have $ P[ \mathbf{X}_{k i}^* = \mathbf{X}_{l j} ] = n^{-1} $ for all $ l $ and $ j $. From the bootstrap sample, we compute the bootstrap value of the statistic $ SS_n^* $. By repeating this procedure $ M_b $ times independently, we get $ M_b $ values of $ SS_n^* $. The p-value of this testing procedure is given by the proportion of the $ M_b $ values of $ SS_n^* $ higher than the actual value of $ SS_n $.
When an observation is repeated in the bootstrap sample, we use the convention of $ \mathbf{s}( \mathbf{0} ) = \mathbf{0} $ while computing the statistic $ SS_n^* $.

Next, we describe the permutation implementation.
Since the underlying distributions of all the groups in the sample are identical when the null hypothesis is true, we generate a random permutation of the pooled sample, then assign the first $ n_1 $ elements in that permuted sample to the first group, the next $ n_2 $ elements to the next group, and so on. From this permuted sample, we compute the value of the test statistic, which is denoted as $ SS_n^\# $. We get $ M_p $ values of $ SS_n^\# $ by repeating this procedure $ M_p $ times independently. The p-value of the permutation implementation of the test is given by the proportion of the $ M_p $ values of $ SS_n^\# $ higher than the actual value of $ SS_n $.

The advantage of the permutation implementation is that no observation is repeated in the generated sample unlike the bootstrap implementation.
However, in our simulation analysis, we did not find any discernible difference between the performances of the two implementations.
In principle, the permutation implementation can also be developed based on all possible permutations of the sample instead of a fixed number of random permutations. However, the number of all permutations of $ ( 1, \ldots, n ) $ increases very rapidly with $ n $, which makes the permutation implementation computationally very expensive even for a moderate $ n $. The test based on random permutations is prescribed to cover for such situations.
We use the permutation implementation of the test for small sample sizes.

The validity of the implementations of our test is established in the following theorem.
\begin{theorem} \label{thm:bootstrapperm}
Let $ n^{-1} n_k \to \lambda_k \in ( 0, 1 ) $ for all $ k $ as $ n \to \infty $.
Under $ \text{H}_0 $ in \eqref{h_0}, for every $ 0 < \alpha < 1 $, the size of a level $ \alpha $ test based on the asymptotic procedure or the bootstrap procedure or the permutation procedure described above converges to $ \alpha $ as $ n \to \infty $ and $ M_b, M_p \to \infty $.
\end{theorem}

Note that the testing procedure is based on the magnitudes of the quantities $ \bar{\mathbf{R}}_k $'s, and the distance between $ \bar{\mathbf{R}}_k $ and $ E[ \bar{\mathbf{R}}_k ] $ converges to $ 0 $ in probability for all $ k $ as $ n \to \infty $ irrespective of $ \text{H}_0 $ in \eqref{h_0}. Also, for all $ k $, $ E[ \bar{\mathbf{R}}_k ] = n^{-1} \sum_{l=1}^{K} n_l E[ \mathbf{s}( \mathbf{X}_{k 1} - \mathbf{X}_{l 1} ) ] \longrightarrow \sum_{l=1}^{K} \lambda_l E[ \mathbf{s}( \mathbf{X}_{k 1} - \mathbf{X}_{l 1} ) ] $ as $ n \to \infty $.
Recall that when the underlying distributions of the $ K $ groups are identical, i.e., when $ \text{H}_0 $ in \eqref{h_0} holds, $ E[ \bar{\mathbf{R}}_k ] = \mathbf{0} $ for all $ k $.
When the distributions of the groups are not identical such that the asymptotic limit of $ E[ \bar{\mathbf{R}}_k ] $ is nonzero for at least one $ k $, the power of all the implementations of the test converges to 1 as the sample size increases, which we state in the next theorem.
\begin{theorem} \label{thm:2}
Assume that $ n^{-1} n_k \to \lambda_k \in ( 0, 1 ) $ for all $ k $ as $ n \to \infty $.
When $ \text{H}_0 $ in \eqref{h_0} is not true such that $ \sum_{l=1}^{K} \lambda_l E[ \mathbf{s}( \mathbf{X}_{k 1} - \mathbf{X}_{l 1} ) ] \neq \mathbf{0} $ for any $ k $, we have $ \left\| E\left[ \mathbf{U}_n \right] \right\| \to \infty $ as $ n \to \infty $, and $ SS_n \stackrel{P}{\longrightarrow} \infty $ as $ n \to \infty $.
As a consequence, for every $ 0 < \alpha < 1 $, the power of a level $ \alpha $ test based on each of the asymptotic procedure, the bootstrap procedure and the permutation procedure converges to $ 1 $ as $ n \to \infty $ and $ M_b, M_p \to \infty $.
\end{theorem}

The standard setup of ANOVA with location shift is a special case of our general setup.
In the literature on ANOVA for functional data, most of the authors concentrated on this setup with certain additional assumptions required for their tests.
In this setup, it is assumed that
\begin{align} \label{setup}
\mathbf{X}_{k i} = \boldsymbol{\mu}_k + \boldsymbol{\epsilon}_{k i} ,
\end{align}
where $ \boldsymbol{\epsilon}_{k i} $'s are independent random elements in $ \mathcal{H} $ having identical distribution $ P_0 $, and $ \boldsymbol{\mu}_k $'s are fixed elements in $ \mathcal{H} $. We are interested in testing whether the $ \boldsymbol{\mu}_k $'s are identical for all $ k $ or not, i.e.,
\begin{align} \label{h_0l}
\text{H}_0' : \; \boldsymbol{\mu}_1 = \boldsymbol{\mu}_2 = \cdots = \boldsymbol{\mu}_K \;.
\end{align}
When $ \text{H}_0' $ in \eqref{h_0l} is true, all the underlying distributions of the $ K $ groups are same in the setup \eqref{setup}, and hence $ \text{H}_0 $ in \eqref{h_0} holds.
Note that the existence of the mean or covariance or any other moment of the error distribution $ P_0 $ is not assumed in the setup \eqref{setup}. This relaxation enables the methodology developed here to be applicable in heavy-tailed processes which have no moments. In the ANOVA with location shift setup for functional data considered in
\cite{cuevas2004anova}, \cite{shen2004f}, \cite{zhang2007statistical}, \cite{zhang2011statistical}, \cite{zhang2014one}, \cite{zhang2018new}, \cite{horvath2015introduction} and \cite{gorecki2015comparison},
it is additionally assumed that the error distribution $ P_0 $ has mean $ \mathbf{0} $ and a covariance operator. Thus our setup \eqref{setup} covers the ANOVA with location shift setup considered by the other authors, but does not assume the existence of the mean or the covariance operator of $ \boldsymbol{\epsilon}_{k i} $'s.

The validity and the asymptotic consistency of the different implementations of the SS test in the ANOVA with location shift setup follow from the result below.
\begin{theorem} \label{thm:2linear}
Let $ n^{-1} n_k \to \lambda_k \in ( 0, 1 ) $ for all $ k $ as $ n \to \infty $.
Under $ \text{H}_0' $ in \eqref{h_0l}, for every $ 0 < \alpha < 1 $, the sizes of a level $ \alpha $ test based on the asymptotic procedure, the bootstrap procedure and the permutation procedure described above converge to $ \alpha $ as $ n \to \infty $ and $ M_b, M_p \to \infty $.
Further, when $ \text{H}_0' $ in \eqref{h_0l} is not true and when the support of $ P_0 $ is not contained in a straight line in $ \mathcal{H} $, we have $ \sum_{l=1}^{K} \lambda_l E[ \mathbf{s}( \mathbf{X}_{k 1} - \mathbf{X}_{l 1} ) ] \neq \mathbf{0} $ for every $ k $, and hence for every $ 0 < \alpha < 1 $, the power of a level $ \alpha $ test based on each of the procedures converges to $ 1 $ as $ n \to \infty $ and $ M_b, M_p \to \infty $.
\end{theorem}

The condition that the support of $ P_0 $ is not contained in a straight line in $ \mathcal{H} $ essentially means that $ P_0 $ is not degenerate in any direction in $ \mathcal{H} $, and this condition holds for all the common non-degenerate distributions like the Gaussian processes, the $ t $ processes, etc.

We have framed the null hypothesis as the equality of the distributions underlying the groups. However, the test statistic $ SS_n $ is based on spatial signs, and can be used in a heteroscedastic situation also. The methodology developed here covers heteroscedastic situations, where the null hypothesis is
\begin{align*}
\text{H}_0^* : \; \sum_{l=1}^{K} \lambda_l E[ \mathbf{s}( \mathbf{X}_{k 1} - \mathbf{X}_{l 1} ) ] = \mathbf{0} \;\text{for all}\; k ,
\end{align*}
where $ \lambda_k = \lim_{n \to \infty} n^{-1} n_k \in ( 0, 1 ) $ for all $ k $. If $ \sqrt{n} ( n^{-1} n_k - \lambda_k ) \to 0 $ as $ n \to \infty $ for all $ k $, then using arguments similar to those in the proof of \autoref{coro:1}, it can be verified that $ SS_n \stackrel{w}{\longrightarrow} \| \mathbf{W} \|^2 $ as $ n \to \infty $ under $ \text{H}_0^* $, where $ \mathbf{W} $ is as described in \autoref{coro:1}. This indicates that results analogous to \autoref{thm:bootstrapperm} and \autoref{thm:2} can be established under $ \text{H}_0^* $, which guarantee the validity of the test procedures for testing $ \text{H}_0^* $.

The null hypothesis $ \text{H}_0^* $ holds in many heteroscedastic setups. For example, consider any distribution $ P $, which is symmetric with respect to the origin $ \mathbf{0} $, i.e., for $ P $ being the distribution of the random element $ \mathbf{X} $, it is also the distribution of $ - \mathbf{X} $. Let $ c_1, \ldots, c_K $ be positive constants, $ \boldsymbol{\mu} $ be a fixed element, and $ P_k $ be the distribution of the random element $ \boldsymbol{\mu} + c_k \mathbf{X} $ for $ k = 1, \ldots, K $. Let $ \mathbf{X}' $ be an independent copy of $ \mathbf{X} $. Then,
\begin{align*}
\mathbf{s}( \mathbf{X}_{k 1} - \mathbf{X}_{l 1} ) \stackrel{d}{=} \mathbf{s}( c_{k} \mathbf{X} - c_{l} \mathbf{X}' ) \stackrel{d}{=} \mathbf{s}( - c_{k} \mathbf{X} + c_{l} \mathbf{X}' ) \stackrel{d}{=} - \mathbf{s}( c_{k} \mathbf{X} - c_{l} \mathbf{X}' )
\end{align*}
for all $ k $ and $ l $,
which implies that $ E[ \mathbf{s}( \mathbf{X}_{k 1} - \mathbf{X}_{l 1} ) ] = \mathbf{0} $ for all $ k $ and $ l $, and $ \text{H}_0^* $ is satisfied. Examples of such symmetric distributions $ P $ include all centered Gaussian processes and centered $ t $ processes.

\section{Comparison of different tests} \label{sec:3}
In this section, we shall compare the asymptotic and finite sample performance of the SS test with those of several other tests in the literature. \autoref{thm:2linear} in the preceding section implies that our SS test is consistent under any fixed alternative. Similar consistency under fixed alternative holds for several of the mean based tests. Hence, in order to compare the asymptotic performance of these tests, we consider a class of shrinking alternatives:
\begin{align}
\boldsymbol{\mu}_k = ( 1 / \sqrt{n} ) \boldsymbol{\delta}_k \text{ for } k = 1, \ldots, K, \text{ where } \boldsymbol{\delta}_1, \ldots, \boldsymbol{\delta}_K \in \mathcal{H} \text{ are fixed}.
\label{asympower1}
\end{align}
Similar shrinking alternatives were considered in \cite{zhang2014one} and \cite{zhang2018new}.
Note that the Frech\'{e}t derivative of $ \mathbf{s}( \mathbf{x} ) = \| \mathbf{x} \|^{-1} \mathbf{x} $ exists at all $ \mathbf{x} \neq \mathbf{0} $, and we denote it as $ \mathbf{s}^{(1)}( \mathbf{x} )( \cdot ) $. The expression for $ \mathbf{s}^{(1)}( \mathbf{x} )( \cdot ) $ is given in the proof of \autoref{lemma:lemma1}. From the following theorem, we derive the asymptotic power of the SS test under the class of shrinking alternatives in \eqref{asympower1}.
\begin{theorem} \label{thm:3}
Let $ \mathbf{X} $ and $ \mathbf{X}' $ be independent random elements having identical distribution $ P_0 $.
Let $ n^{-1} n_k \to \lambda_k \in ( 0, 1 ) $ for all $ k $ as $ n \to \infty $,
and $ \bar{\boldsymbol{\delta}} = \sum_{l = 1}^{K} \lambda_l \boldsymbol{\delta}_l $.
Assume that $ P_0 $ is non-atomic and not contained in any straight line in $ \mathcal{H} $,
and $ E\left[ \left\| \mathbf{X} - \mathbf{X}' \right\|^{-1} \right] < \infty $.
Then, under \eqref{asympower1}, $ E\left[ \mathbf{U}_n \right] \to \mathbf{U}_0 = \left( \mathbf{u}_1, \ldots, \mathbf{u}_K \right) $ as $ n \to \infty $, where
$ \mathbf{u}_k = E\left[ \mathbf{s}^{(1)}\left( \mathbf{X} - \mathbf{X}' \right) \right] \left( \sqrt{\lambda_k} \left( \boldsymbol{\delta}_k - \bar{\boldsymbol{\delta}} \right) \right) $
for all $ k = 1, \ldots, K $. Consequently,
$ \mathbf{U}_n \stackrel{w}{\longrightarrow} G\left( \mathbf{U}_0, \boldsymbol{\Sigma} \right) $ as $ n \to \infty $, and
$ SS_n = \| \mathbf{U}_n \|^2 \stackrel{w}{\longrightarrow} \left\| \mathbf{U}_0 - \sum_{i=1}^{\infty} \langle \mathbf{U}_0, \boldsymbol{\beta}_i \rangle \right\|^2 + \sum_{i=1}^{\infty} \left( \langle \mathbf{U}_0, \boldsymbol{\beta}_i \rangle + \sqrt{\alpha_i} \mathbf{Z}_i \right)^2 $,
where
$ \alpha_1, \alpha_2, \ldots $ is the decreasing sequence of eigenvalues of $ \boldsymbol{\Sigma} $, $ \boldsymbol{\beta}_i $ is the corresponding eigenvector of $ \alpha_i $ for all $ i $, and $ \mathbf{Z}_1, \mathbf{Z}_2, \ldots $ are independent standard normal random variables.
Further, $ \mathbf{U}_0 \neq \mathbf{0} $ if $ \boldsymbol{\delta}_k \neq \boldsymbol{\delta}_l $ for at least one pair of $ k $ and $ l $.
\end{theorem}

\subsection{Description of other tests}\label{descriptionoftests}
We describe below the tests in the literature with which we compare the performance of the SS test in this paper. We compare the finite sample powers against fixed alternatives and the asymptotic powers against shrinking alternatives of these tests.
Though our testing procedures are also applicable for heteroscedastic data, several of the testing procedures with which we compare the performance of our test were presented for the homoscedastic case (see, e.g., \cite{shen2004f,zhang2007statistical,zhang2011statistical,zhang2014one}). For this reason, we restrict the comparison of performances in the simulation study to the homoscedastic setup. The real datasets we consider for the comparison of performances also do not exhibit heteroscedasticity.

The tests in the literature, which we consider here, are developed for the case where the observations are real-valued square-integrable random functions on an interval $ [ a, b ] $, i.e., $ \mathbf{X}_{k i}( \cdot ) : [ a, b ] \to \mathbb{R} $ is a random function for all $ k $ and $ i $. So, $ \mathbf{X}_{k i}( \cdot ) $ is an element of the $ L_2[ a, b ] $ space, which is a separable Hilbert space, and $ \| \mathbf{x} \| = \left( \int \| \mathbf{x}( t ) \|^2 \mathrm{d}t \right)^{1/2} $ for $ \mathbf{x} \in L_2[ a, b ] $.
Define $ \bar{\mathbf{X}}_{k \cdot}( t ) = n_k^{-1} \sum_{i = 1}^{n_k} \mathbf{X}_{k i}( t ) $ and $ \bar{\mathbf{X}}_{\cdot \cdot}( t ) = n^{-1} \sum_{k = 1}^{K} \sum_{i = 1}^{n_k} \mathbf{X}_{k i}( t ) $ for all $ t \in [ a, b ] $.

We call the test of ANOVA for functional data proposed by \cite{cuevas2004anova} as the CFF test.
From the methodology developed in \cite{zhang2007statistical}, an $ L_2 $ norm based test of ANOVA can be derived, which we call the ZC test.
In \cite{shen2004f} and \cite{zhang2011statistical}, another test for functional linear models was proposed, and we call the corresponding test of ANOVA as the F-type test.
The test of ANOVA described in \cite{zhang2014one} is denoted here as the GPF test,
and the test proposed in \cite{zhang2018new} is called the F-max test following the authors.
We call the permutation test given by \cite{gorecki2015comparison} as the GS test.
The random projection based test proposed in \cite{cuesta2010simple} is denoted here as the CAFB test.
The test by \cite{horvath2015introduction} is denoted as the HR test.
The detailed descriptions of the test statistics and the implementations of the tests are presented in \autoref{asymptoticpowerexpression}.

We shall compare the finite sample performance of the SS test with all the above tests. For the asymptotic power comparison against shrinking alternatives, we consider all the tests except the F-type test, the CAFB test and the GS test. The F-type test is excluded from the asymptotic power comparison because the F-type test and the ZC test have identical asymptotic performance. The CAFB test and the GS test are also excluded from asymptotic power comparison because the authors of these tests did not study their asymptotic distributions, and we found it difficult to derive their asymptotic distributions.

For the comparison of the asymptotic power against shrinking alternatives, we need the mathematical expressions of the asymptotic powers of these tests in this setup.
The asymptotic power of the F-max test under \eqref{asympower1} is derived from Proposition 3 in \cite{zhang2018new}, and that of the GPF test under \eqref{asympower1} is obtained from Proposition 3 in \cite{zhang2014one}.
From Theorem 4.17 and Remark 4.11 in \cite{zhang2013analysis}, we get that the ZC test and the F-type test are asymptotically equivalent, and hence have the same asymptotic power.
The expressions of the asymptotic powers of the CFF test, the ZC test and the HR test under the class of shrinking alternatives described in \eqref{asympower1}, where $ \mathcal{H} = L_2[ a, b ] $, are derived in \autoref{asymptoticpowerexpression}.

\subsection{Asymptotic power study of different tests}
\label{asympower}
We compare the asymptotic powers of the tests under the class of shrinking alternatives described in \eqref{asympower1} in probability models with 3 groups.
Note that for a random element $ \mathbf{V} $ following Gaussian process $ G $, its corresponding $ t $ process with $ k $ degrees of freedom can be obtained by dividing $ \mathbf{V} $ by $ \sqrt{\boldsymbol{\chi} / k} $, where $ \boldsymbol{\chi} $ is a chi-square random variable with $ k $ degrees of freedom independent of $ \mathbf{V} $. We denote this $ t $ process as $ t_{k, G} $.
The derivation of the asymptotic powers of the tests except the SS test requires that $ E\left[ \left\| \mathbf{X} \right\|^2 \right] < \infty $, and for $ \mathbf{X} $ following a $ t_{k, G} $ distribution with $ k \ge 3 $, $ E\left[ \left\| \mathbf{X} \right\|^2 \right] < \infty $.
We take $ \mathcal{H} = L_2[ a, b ] $, and consider the standard Brownian motion (SBM) over $ [ a, b ] $ along with its corresponding $ t_3 $ and $ t_4 $ processes as particular cases of $ P_0 $, where $ P_0 $ is as described below \eqref{setup}.
The interval $ [ a, b ] $ is taken as $ [ 0.25, 0.75 ] $, which ensures that $ \mathbf{X}_{k i}\left( t \right) $ is non-degenerate at every $ t \in [ a, b ] $ for $ P_0 $ being any of the processes mentioned above.

\begin{figure}[h]
	\centering
	\includegraphics[width=1\linewidth]{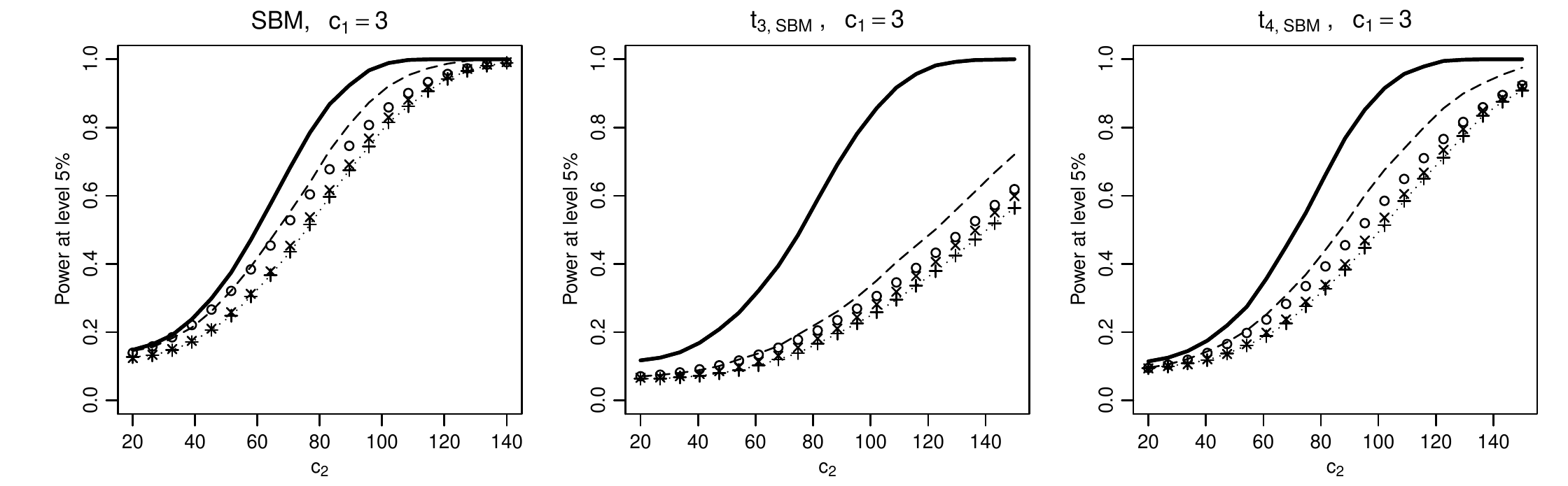}
	\caption{Asymptotic power curves of the SS test ($ \boldsymbol{-} $), CFF test (++), ZC test ($ \cdots $), F-max test (--{ }--), GPF test ($ \times \times $) and HR test ($ \circ \circ $) under nominal level 5\% in SBM and associated $ t $ processes.}
	\label{fig:anovaasym1}
\end{figure}
We take the shifts in the groups under \eqref{asympower1} as $ \boldsymbol{\delta}_1 = \mathbf{0} $, $ \boldsymbol{\delta}_2 = c_1 \boldsymbol{\eta}_1 $ and $ \boldsymbol{\delta}_3 = c_2 \boldsymbol{\eta}_2 $.
Here, $ c_1 , c_2 $ are real valued constants and $ \boldsymbol{\eta}_1, \boldsymbol{\eta}_2 \in \mathcal{H} $ are defined as $ \boldsymbol{\eta}_1( t ) = t $ and $ \boldsymbol{\eta}_2( t ) = (t - 0.25) (0.75 - t) $ for $ t \in [ 0.25, 0.75 ] $.
Also, we take $ \lambda_1 = \lambda_2 = \lambda_3 = 1/3 $.
The nominal level is fixed at 5\%.
To generate asymptotic power curves, we fix $ c_1 $ and vary $ c_2 $ in an interval.
%Results for the other distributions are presented in \autoref{additionalmodels_asym}.

The asymptotic powers of the tests in the setup described above are presented in \autoref{fig:anovaasym1}.
We observe that the performance of the SS test is better than all the other tests for both the Gaussian and the non-Gaussian processes.

%In \autoref{additionalmodels_asym}, we notice that the performance of the SS test may not be always better than all the other tests, particularly for the Gaussian processes, however, it is found to be always close to the best performing test. For the non-Gaussian heavy tail processes considered, the performance of the SS test is almost always found to be significantly better than all the other tests.

\subsection{Finite sample performance in simulated data}
\label{simulation}
We compare the finite sample power of the SS test, the CFF test, the ZC test, the F-max test, the GPF test, the HR test, the CAFB test and the GS test in the setup \eqref{setup} with 3 groups.
We consider the same space $ \mathcal{H} = L_2[ a, b ] $ with $ [ a, b ] = [ 0.25, 0.75 ] $ as in \autoref{asympower}.
%Considering $ [ a, b ] $ in this way ensures that all the processes we consider as the underlying distributions of the groups are non-degenerate at every $ t \in [ a, b ] $.
For the cases of the distribution $ P_0 $, we consider the following three collections:
\begin{itemize}
\item \textit{Gaussian and $ t $ processes:} The cases of the distribution $ P_0 $ are the SBM and its associated $ t_1 $ and $ t_3 $ processes.

\item \textit{Contaminated models:} Here, the distribution $ P_0 $ is replaced with a mixture distribution $ ( 1 - p ) P_0 + p P^{(s)} $, where $ p $ is the proportion of contamination, $ s $ is a positive number and $ P^{(s)} $ is the distribution of the random element $ s \mathbf{X} $, where $ \mathbf{X} $ follows $ P_0 $. We take $ p = 0.25 $ and $ s = 5 $. The cases of the distribution $ P_0 $, which are replaced with the respective contaminated models, are the SBM and its associated $ t_1 $ and $ t_3 $ processes as above.

\item \textit{Skewed distributions:} Three skewed distributions are considered:
\begin{enumerate}
\item Geometric Brownian motion: $ P_0 $ is the distribution of the process $ \mathbf{G}_1(t) = \exp[ \mathbf{B}(t) ] $, where $ \mathbf{B}(t) $ is an SBM.

\item Squared Brownian motion: $ P_0 $ is the distribution of the process $ \mathbf{G}_2(t) = ( \mathbf{B}(t) )^2 $, where $ \mathbf{B}(t) $ is an SBM.

\item Squared $ t $ process: $ P_0 $ is the distribution of the process $ \mathbf{G}_3(t) = ( \mathbf{T}(t) )^2 $, where $ \mathbf{T}(t) $ is the $ t_{3, \text{SBM}} $ process.
\end{enumerate}
\end{itemize}
Note that the covariance operator of $ P_0 $ exists for all the cases above except the $ t_{1, \text{SBM}} $ process and the contaminated $ t_{1, \text{SBM}} $ process. Most of the other tests, which we consider here for comparison of powers, require the existence of the covariance operator for the validity of their theoretical properties. However, in this finite sample power study, we consider in addition the heavy-tailed $ t_{1, \text{SBM}} $ and the contaminated $ t_{1, \text{SBM}} $ processes to investigate the finite-sample performance of the tests under such non-Gaussian heavy-tailed distributions.

The observations from these distributions, which are random functions, are generated on an equispaced grid of length 100 on the interval $ [ a, b ] $.

The shifts in the 3 groups are taken as
$ \boldsymbol{\mu}_1 = \mathbf{0} $, $ \boldsymbol{\mu}_2 = c_1 \boldsymbol{\eta}_1 $ and $ \boldsymbol{\mu}_3 = c_2 \boldsymbol{\eta}_2 $, where $ c_1, c_2 $ are constants and $ \boldsymbol{\eta}_1, \boldsymbol{\eta}_2 \in \mathcal{H} $ are defined by $ \boldsymbol{\eta}_1( t ) = t $ and $ \boldsymbol{\eta}_2( t ) = (t - 0.25) (0.75 - t) $ for $ t \in [ 0.25, 0.75 ] $, which are the same as considered in \autoref{asympower}.
Note that when $ c_1 = c_2 = 0 $, $ \text{H}_0' $ in \eqref{h_0l} is satisfied. So, taking $ c_1 = c_2 = 0 $, we estimate the sizes of the tests at nominal level 5\%.
To generate power curves, we fix $ c_1 $ and vary $ c_2 $ in an interval, and estimate the powers of the tests corresponding to those values of $ c_1 $ and $ c_2 $. Specifically, we choose 20 equispaced values of $ c_2 $ in an interval which spans most of the power range of the SS test. The specific values of the constants vary among the distributions.

To investigate the performances of the asymptotic and the permutation implementations of the SS test, we consider two cases of group sizes. We take $ n_1 = n_2 = n_3 = 20 $ in the first case, where the asymptotic implementation of the SS test is used, and $ n_1 = n_2 = n_3 = 4 $ in the second case, where the permutation implementation is used.
We estimate the sizes and the powers of the tests under nominal level 5\% based on 1000 independent replications.

To get observations from the chosen distributions of $ \mathbf{X}_{k i} $ in \eqref{setup}, we first generate observations $ \boldsymbol{\epsilon}_{k i} $ from $ P_0 $, then add the appropriate $ \boldsymbol{\mu}_k $ to get the $ \mathbf{X}_{k i} $ observations. For each of the cases of $ P_0 $ in the power study, we generate 1000 independent sets $ \{ \boldsymbol{\epsilon}_{k i} \midil \boldsymbol{\epsilon}_{k i} \sim P_0 \text{ are independent}, i = 1, \ldots, n_k; k = 1, 2, 3 \} $, then obtain the corresponding 1000 independent sets $ \{ \mathbf{X}_{k i} = \boldsymbol{\mu}_k + \boldsymbol{\epsilon}_{k i} \midil \boldsymbol{\mu}_1 = \mathbf{0}, \boldsymbol{\mu}_2 = c_1 \boldsymbol{\eta}_1, \boldsymbol{\mu}_3 = c_2 \boldsymbol{\eta}_2,\, i = 1, \ldots, n_k; k = 1, 2, 3 \} $. For different values of $ c_2 $, the sets of observations $ \{ \boldsymbol{\epsilon}_{k i} \} $ are kept same. This eliminates the random fluctuation in the estimated power curves arising from the random sampling process, and the estimated power curves are relatively smooth. The 1000 independent sets of the $ \boldsymbol{\epsilon}_{k i} $ observations used to estimated the sizes of the tests for different underlying distributions $ P_0 $ are generated independently from those in the power study.

\subsubsection*{Results for the Gaussian and $ t $ processes:}
The estimated sizes of the tests in the the Gaussian and $ t $ processes are presented in \autoref{table1}, and the estimated power curves of the tests are plotted in \autoref{fig:sim1}. We note that only the SS test and the GS test have sizes close to the nominal level irrespective of the group sizes or underlying distributions.
\begin{table}
\caption{Estimated sizes of the SS test, CFF test, ZC test, F-max test, GPF test, F-type test, HR test, CAFB test and GS test in the Gaussian and $ t $ processes (nominal level 5\%).}
\begin{tabular} {lcccccccccc}  
\hline
$ P_0 $	& $ ( n_1, n_2, n_3 ) $	& SS & CFF & ZC & F-max & GPF & F-type & HR & CAFB & GS \\ \hline 
SBM	& (20, 20, 20) & 0.054 & 0.054 & 0.060 & 0.064 & 0.061 & 0.051 & 0.081 & 0.032 & 0.048 \\
SBM & (4, 4, 4) & 0.056 & 0.110 & 0.125 & 0.044 & 0.102 & 0.013 & 0.305 & 0.031 & 0.052 \\
$ t_{1, \text{SBM}} $	& (20, 20, 20) & 0.053 & 0.014 & 0.016 & 0.017 & 0.015 & 0.014 & 0.025 & 0.025 & 0.055 \\
$ t_{1, \text{SBM}} $	& (4, 4, 4) & 0.058 & 0.041 & 0.080 & 0.021 & 0.045 & 0.001 & 0.183 & 0.039 & 0.057 \\
$ t_{3, \text{SBM}} $	& (20, 20, 20) & 0.049 & 0.053 & 0.054 & 0.053 & 0.051 & 0.048 & 0.081 & 0.030 & 0.052 \\
$ t_{3, \text{SBM}} $	& (4, 4, 4) & 0.055 & 0.089 & 0.095 & 0.029 & 0.090 & 0.002 & 0.270 & 0.031 & 0.054 \\
\hline
\end{tabular}
\label{table1}
\end{table}
\begin{figure}[h]
	\centering
	\includegraphics[width=1\linewidth]{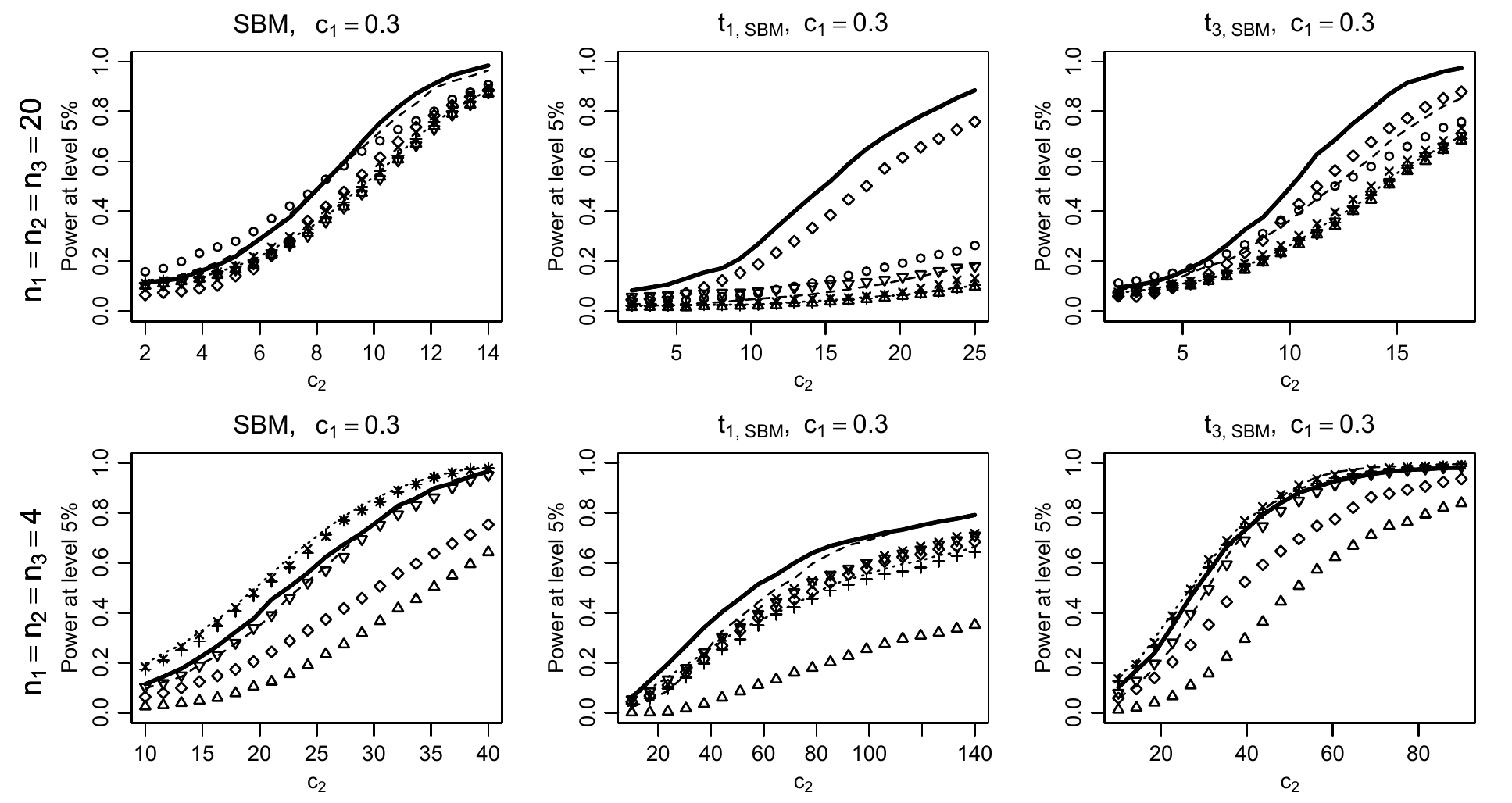}
	\caption{SS test ($ \boldsymbol{-} $), CFF test (++), ZC test ($ \cdots $), F-max test (--{ }--), GPF test ($ \times \times $), F-type test ($ \triangle $), HR test ($ \circ $), CAFB test ($ \diamond $) and GS test ($ \triangledown $) under nominal level 5\% in SBM and associated $ t $ processes: $ 1^\text{st} $ row: $ n_1 = n_2 = n_3 = 20 $, $ 2^\text{nd} $ row: $ n_1 = n_2 = n_3 = 4 $.}
	\label{fig:sim1}
\end{figure}
The F-max test, the F-type test and the CAFB test tend to have lower sizes for small group sizes or non-Gaussian heavy tail processes (i.e., $ t_{1, \text{SBM}} $ and $ t_{3, \text{SBM}} $), while the sizes of the CFF test, the ZC test and the GPF test are sometimes significantly higher than the nominal level and sometimes somewhat lower than the nominal level. The size of the HR test is always significantly higher than the nominal level for small group sizes, and so it is omitted from the power comparison for the case of $ n_1 = n_2 = n_3 = 4 $.

For $ n_1 = n_2 = n_3 = 20 $, we note in \autoref{fig:sim1} that the performance of the SS test in case of SBM is sometimes narrowly beaten by only the HR test. However, the HR test also exhibits an estimated size over 8.5\% in that case, considerably higher than the nominal level, indicating it is biased towards rejecting the null hypothesis. For the non-Gaussian heavy tail processes (i.e., $ t_{1, SBM} $ and $ t_{3, \text{SBM}} $), the performance of the SS test is significantly better than all other tests.

When the group sizes are small, i.e., $ n_1 = n_2 = n_3 = 4 $, and $ P_0 $ is the SBM, we notice that the power of the SS test is narrowly exceeded by that of the CFF test, the ZC test and the GPF test. However, in this case, the CFF test, the ZC test and the GPF test have sizes over 10\% indicating their bias towards rejecting the null hypothesis. In the case of $ P_0 $ being $ t_{1, \text{SBM}} $ here, the SS test exhibits better performance than all other tests. For $ P_0 $ being $ t_{3, \text{SBM}} $, the power curve of the SS test is very close to that of the best performing test.

\subsubsection*{Results for the contaminated models:}
Since the SS test exhibits superior performance in distributions having heavy tails, we were motivated to investigate its performance in contaminated models.
The contaminated models introduce impurities in the covariance operator of the groups but do not affect the means.
We consider the same choices of $ \boldsymbol{\eta}_1 $, $ \boldsymbol{\eta}_2 $, $ n_1 $, $ n_2 $ and $ n_3 $ as in the uncontaminated models. Under the nominal level 5\%, we estimate the sizes and the powers of the tests and present them in \autoref{table1impure} and \autoref{fig:sim2}, respectively.

\begin{table}[h]
\caption{Estimated sizes of the SS test, CFF test, ZC test, F-max test, GPF test, F-type test, HR test, CAFB test and GS test in the Gaussian and $ t $ processes under covariance impurity (nominal level 5\%).}
\begin{tabular} {lcccccccccc}  
\hline
$ P_0 $	& $ ( n_1, n_2, n_3 ) $	& SS & CFF & ZC & F-max & GPF & F-type & HR & CAFB & GS \\ \hline 
SBM with impurity	& (20, 20, 20) & 0.046 & 0.047 & 0.046 & 0.045 & 0.046 & 0.043 & 0.049 & 0.027 & 0.045 \\
SBM with impurity	& (4, 4, 4) & 0.058 & 0.034 & 0.056 & 0.012 & 0.026 & 0.000 & 0.198 & 0.032 & 0.037 \\ 
$ t_{1, \text{SBM}} $ with impurity	& (20, 20, 20) & 0.053 & 0.019 & 0.018 & 0.017 & 0.019 & 0.016 & 0.023 & 0.039 & 0.057 \\
$ t_{1, \text{SBM}} $ with impurity	& (4, 4, 4) & 0.054 & 0.024 & 0.047 & 0.008 & 0.030 & 0.000 & 0.172 & 0.030 & 0.045 \\ 
$ t_{3, \text{SBM}} $ with impurity	& (20, 20, 20) & 0.054 & 0.039 & 0.039 & 0.032 & 0.041 & 0.038 & 0.052 & 0.028 & 0.048 \\
$ t_{3, \text{SBM}} $ with impurity	& (4, 4, 4) & 0.043 & 0.031 & 0.042 & 0.014 & 0.038 & 0.000 & 0.196 & 0.027 & 0.038 \\ 
\hline
\end{tabular}
\label{table1impure}
\end{table}
\begin{figure}[h]
	\centering
	\includegraphics[width=1\linewidth]{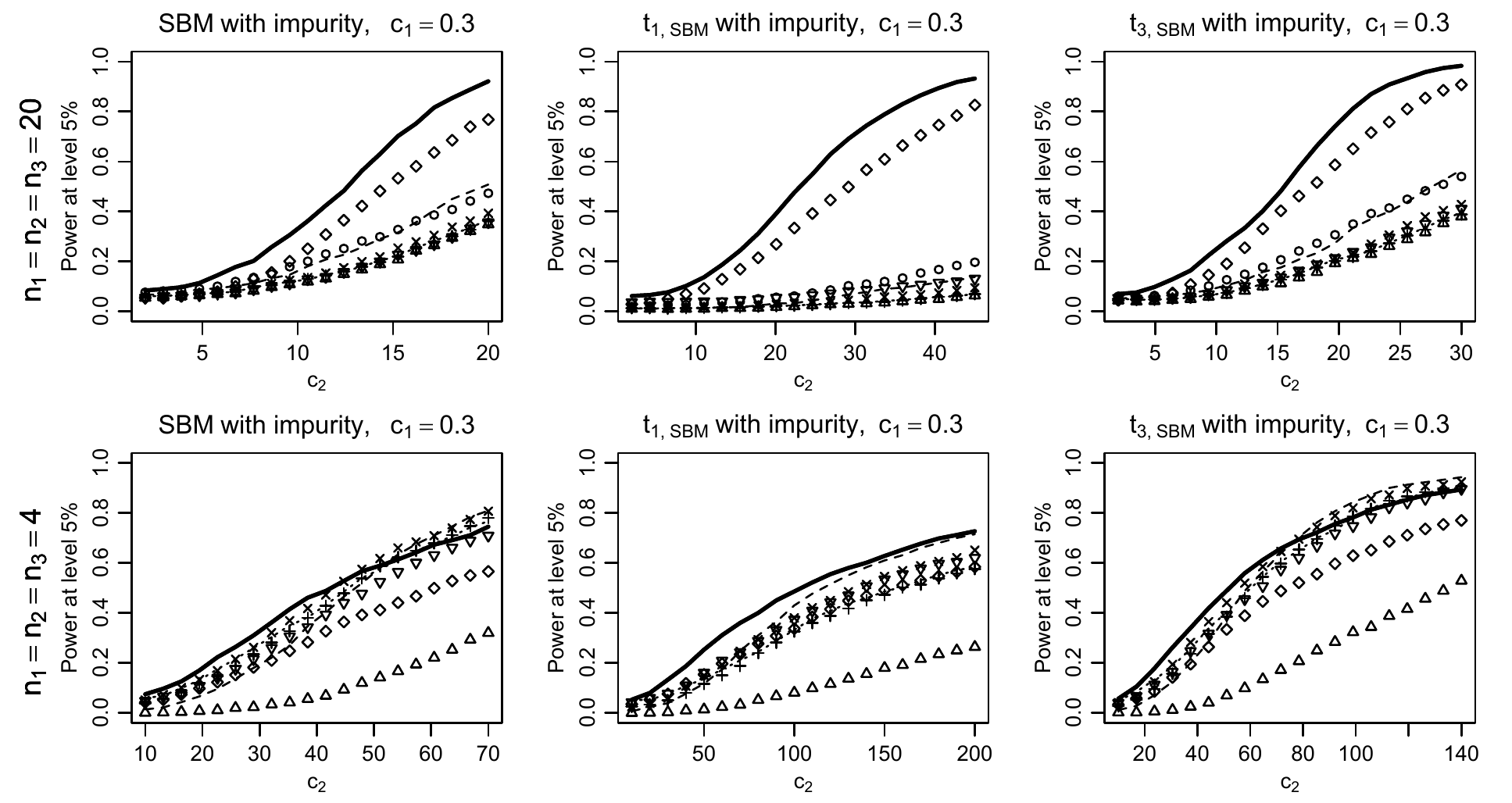}
	\caption{SS test ($ \boldsymbol{-} $), CFF test (++), ZC test ($ \cdots $), F-max test (--{ }--), GPF test ($ \times \times $), F-type test ($ \triangle $), HR test ($ \circ $), CAFB test ($ \diamond $) and GS test ($ \triangledown $) under nominal level 5\% in SBM and associated $ t $ processes with impurities: $ 1^\text{st} $ row: $ n_1 = n_2 = n_3 = 20 $, $ 2^\text{nd} $ row: $ n_1 = n_2 = n_3 = 4 $.}
	\label{fig:sim2}
\end{figure}
In \autoref{table1impure}, we note that the estimated sizes of the SS test is always close to the nominal level 5\%. The estimated sizes of the GS test is close to the nominal level most of the times. The sizes for the HR test are significantly high for small group sizes, i.e., $ n_1 = n_2 = n_3 = 4 $, and hence it is omitted in the power comparison for $ n_1 = n_2 = n_3 = 4 $ in \autoref{fig:sim2}. The sizes of the other tests are often considerably lower than the nominal level. In \autoref{fig:sim2}, we notice that the power of the SS test is nearly always higher than the powers of all other tests.
In comparison to the cases presented in \autoref{fig:sim1}, the difference between the performances of the SS test and other tests increases for SBM and $ t_{3, \text{SBM}} $ and $ n_1 = n_2 = n_3 = 20 $.

\subsubsection*{Results for the skewed distributions:}
We consider three skewed distributions as described before. The third case, which is the squared $ t $ process given by $ \mathbf{G}_3(t) = ( \mathbf{T}(t) )^2 $, $ \mathbf{T}(t) $ being the $ t_{3, \text{SBM}} $ process, can also be considered as a skewed and heavy-tailed process. However, the geometric Brownian motion also has heavy tails. We denote the case of $ P_0 $ being the geometric Brownian motion as GBM, the case of the squared Brownian motion as $ ( \text{SBM} )^2 $ and the case of the squared $ t_{3, \text{SBM}} $ process as $ ( t_{3, \text{SBM}} )^2 $. The estimated sizes under the nominal level 5\% are presented in \autoref{table1skew}. The estimated powers under the nominal level 5\% are depicted in \autoref{fig:sim3}.
\begin{table}[h]
\caption{Estimated sizes of the SS test, CFF test, ZC test, F-max test, GPF test, F-type test, HR test, CAFB test and GS test in the skewed distributions (nominal level 5\%).}
\begin{tabular} {lcccccccccc}  
\hline
$ P_0 $	& $ ( n_1, n_2, n_3 ) $	& SS & CFF & ZC & F-max & GPF & F-type & HR & CAFB & GS \\ \hline 
GBM	& (20, 20, 20) & 0.051 & 0.043 & 0.044 & 0.051 & 0.048 & 0.038 & 0.093 & 0.029 & 0.046 \\
GBM	& (4, 4, 4) & 0.047 & 0.076 & 0.092 & 0.048 & 0.082 & 0.007 & 0.281 & 0.035 & 0.054 \\ 
$ ( \text{SBM} )^2 $	& (20, 20, 20) & 0.055 & 0.035 & 0.034 & 0.041 & 0.037 & 0.031 & 0.106 & 0.034 & 0.035 \\
$ ( \text{SBM} )^2 $	& (4, 4, 4) & 0.047 & 0.061 & 0.070 & 0.049 & 0.068 & 0.002 & 0.361 & 0.029 & 0.044 \\ 
$ ( t_{3, SBM} )^2 $	& (20, 20, 20) & 0.051 & 0.020 & 0.022 & 0.031 & 0.030 & 0.022 & 0.049 & 0.030 & 0.049 \\
$ ( t_{3, SBM} )^2 $	& (4, 4, 4) & 0.044 & 0.036 & 0.041 & 0.037 & 0.041 & 0.000 & 0.285 & 0.026 & 0.043 \\ 
\hline
\end{tabular}
\label{table1skew}
\end{table}
\begin{figure}[h]
	\centering
	\includegraphics[width=1\linewidth]{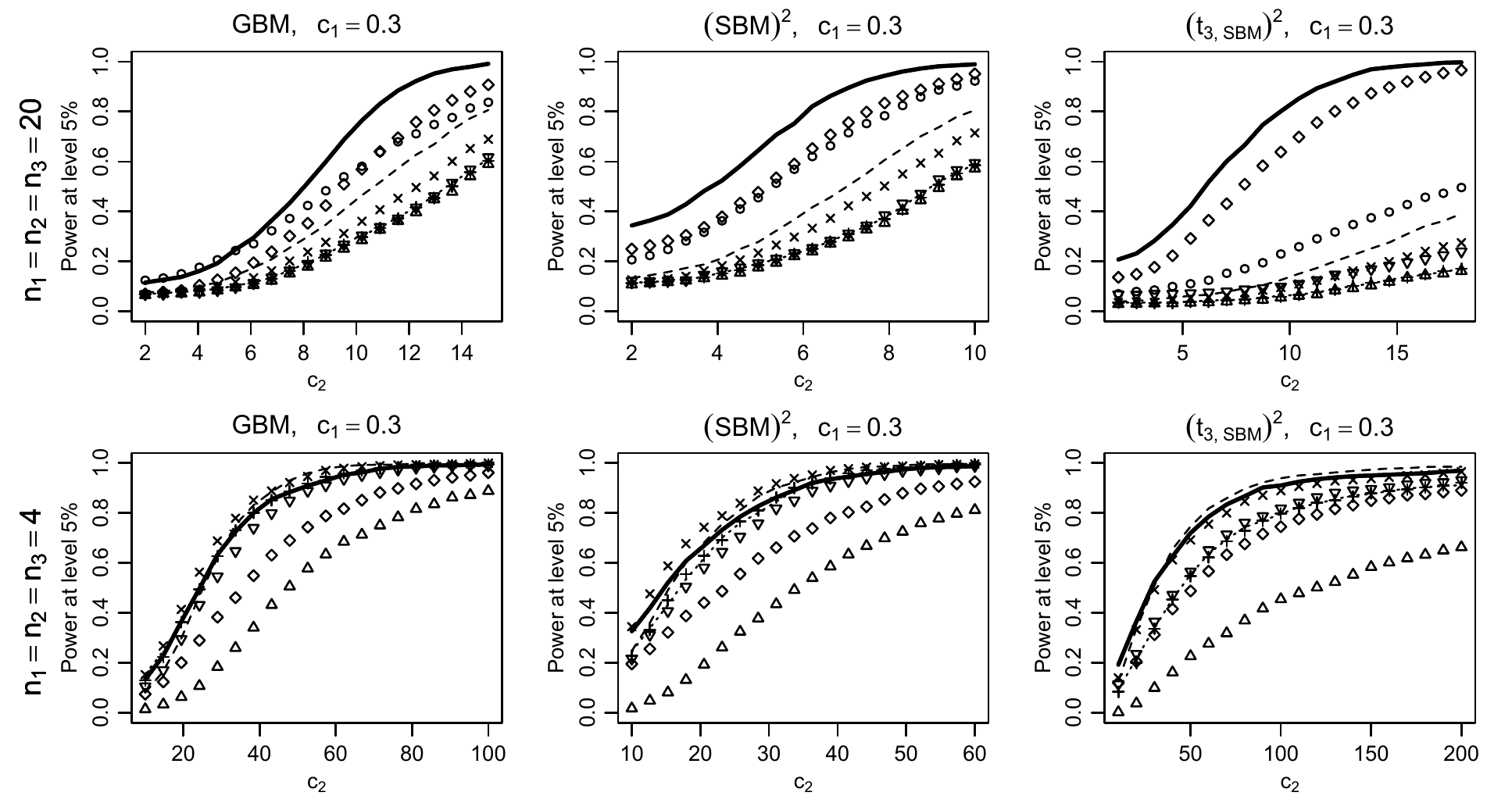}
	\caption{SS test ($ \boldsymbol{-} $), CFF test (++), ZC test ($ \cdots $), F-max test (--{ }--), GPF test ($ \times \times $), F-type test ($ \triangle $), HR test ($ \circ $), CAFB test ($ \diamond $) and GS test ($ \triangledown $) under nominal level 5\% in the skewed models: $ 1^\text{st} $ row: $ n_1 = n_2 = n_3 = 20 $, $ 2^\text{nd} $ row: $ n_1 = n_2 = n_3 = 4 $.}
	\label{fig:sim3}
\end{figure}

In \autoref{table1skew}, the estimated sizes of the SS test is again found to be always very close to the nominal level 5\%. The estimated sizes of the GS test is also close to the nominal level. The estimated sizes for the HR test are significantly high for small group sizes, i.e., $ n_1 = n_2 = n_3 = 4 $, and so we omit it in the power comparison for $ n_1 = n_2 = n_3 = 4 $ in \autoref{fig:sim3}. In \autoref{fig:sim3}, it can be seen that the power of the SS test is almost uniformly higher than the powers of all other tests for higher sample sizes, i.e., $ n_1 = n_2 = n_3 = 20 $. But for the lower sample sizes $ n_1 = n_2 = n_3 = 4 $ also, the power of the SS test is either the highest or very close to the highest among the tests.

\subsubsection*{Comparison of the asymptotic and permutation implementations:}
We compare the sizes of the asymptotic and the permutation implementations of the SS test for 3 groups with $ n_1 = n_2 = n_3 = 5, 10, 20 \text{ and } 40 $, and compare their powers with $ n_1 = n_2 = n_3 = 10 $ for the distributions considered previously in this subsection. The number of random permutations taken is 1000, and the sizes and the powers are estimated based on 1000 independent replications with nominal level 5\%.
All the results pertaining to this investigation are presented in \autoref{asym_perm}.
We notice there that the powers of the two implementations are very close to each other. However, when $ n_1 = n_2 = n_3 = 5 $, the size of the asymptotic implementation is sometimes somewhat smaller than the nominal level. The estimated sizes for the permutation implementation are always close to the nominal level.
These findings indicate that the permutation implementation has power similar to the asymptotic implementation, while the asymptotic implementation may be slightly more conservative when the group sizes are small. This points to the utility of the permutation implementation for small group sizes. We suggest using the permutation implementation when the minimum group size is less than or equal to 5 and the total sample size is less than or equal to 20.

The permutation implementation is found to be computationally expensive compared to the asymptotic implementation. We found that when $ n_1 = n_2 = n_3 = 5 $, the computation of a single p-value using the permutation implementation takes on average nearly twice the time taken by the asymptotic implementation. When $ n_1 = n_2 = n_3 = 10, 20 \text{ and } 40 $, the factors are on average around 7, 22 and 60, respectively.

The code for the SS test is written in R and the numerically heavy portions are implemented using the Rcpp package.

For univariate, multivariate and high-dimensional data, several authors investigated and demonstrated the utility of Studentizing the bootstrap or the permutation test statistic (see, e.g., \cite{janssen2003bootstrap}, \cite{delaigle2011robustness} and \cite{pauly2015asymptotic}). In \cite{delaigle2011robustness}, it was demonstrated that Studentization leads to better performance of the test particularly when the underlying distribution is non-Gaussian or heavy-tailed. However, in the setup considered in this paper, Studentization of the permutation or the bootstrap procedures would significantly increase the computational complexity of the testing method due to the involvement of the covariance operator, the calculation of which in our case is found to be quite expensive computationally. Further, the performance of our testing procedures is found to be considerably less affected by the non-Gaussian tails compared to other tests in the literature. We investigated the performance of the Studentized test statistic in all the probability models described in this subsection, and found that despite taking considerably more time to yield output, the Studentized test shows no significant gain in estimated power over the usual non-Studentized test. In fact, we found that the estimated powers and sizes of the non-Studentized test are very close to the those of the Studentized test. Based on these findings, we do not recommend using Studentization with our procedures.

\section{Analysis of real data}
\label{realdata}
We compare the performance of the tests in three real datasets: the Canadian Temperature data, the Tecator data and the Orthosis data.

The Canadian Temperature dataset, analyzed in \cite{zhang2013analysis}, contains the daily temperature records over 365 days at 35 Canadian weather stations situated in three regions: 15 in Eastern Canada, 15 in Western Canada, and 5 in Northern Canada.
One may be interested to know whether there is a statistically significant effect of location on the temperature curves over theses three regions. This leads to an ANOVA problem with three groups with a total of 35 observations, where each observation is a daily temperature curve, and the three groups are the three regions.
Since the sizes of the groups are small, we apply the permutation implementation of the SS test, and the other tests are implemented based on the advices of their authors.

The Tecator dataset, available in the `caret' package in R, contains the percentage values of moisture, fat and protein contents of 215 meat samples along with their absorbance spectra in the wavelength range 850--1050 nm. Moisture, fat and protein contents were measured by analytical chemistry, while a Tecator Infratec Food and Feed Analyzer was used to record the absorbance spectrum.
We divide the observations in three groups by the value of the protein content such that the 33\% of the observations with the lowest protein content are in the first group, the 33\% with the highest protein content are in the third group and the rest in the second group.
The group sizes are found to be 72, 69 and 74. We are interested to know whether their is a significant effect of the protein content over the absorption spectra curves, and we get an ANOVA problem with three groups.

The Orthosis dataset, analyzed in \cite{zhang2014one}, contains orthotic measurements on a volunteer under four experimental conditions. For each condition, 10 observations were recorded, and each observation depicts the moments of force over a fixed time duration at the knee of the volunteer. Each observation is recorded at a fixed grid of length 256, and we consider them as functional observations.
We are interested in whether the experimental conditions have same effect, and consequently get an ANOVA problem with 4 groups with 10 observations in each group.
On visual inspection, we found that the observations corresponding to one of the conditions are far away from the other observations, and so we discard that group from our analysis. Hence, we get an ANOVA problem with 3 groups with 10 observations in each group.

The p-values of all the tests in each of the 3 datasets are presented in \autoref{table:realdatasuppvalue}. We notice that all the tests exhibit very small p-values in every dataset.
\begin{table}[h]
\caption{P-values of the tests in the real datasets.
}
\begin{center}
\begin{tabular}{cccccccccc}
	\hline
	Data													& SS	& CFF	& ZC	& F-max		& GPF	& F-type	& HR	& CAFB	& GS	\\\hline
	Canada	& 0	& 0	& 0	& 0	& 0	& 0	& 0	& 0	& 0	\\
	Tecator													& 0		& 0.001	& 0.002	& 0.001		& 0.002	& 0.002		& 0.002	& 0		& 0.002	\\
	Orthosis												& 0		& 0		& 0		& 0			& 0		& 0			& 0		& 0		& 0		\\\hline
\end{tabular}
\end{center}
\label{table:realdatasuppvalue}
\end{table}

Since the p-values of all the tests in all the datasets are found to be statistically significant, we carry out a size and power study of the tests based on subsampling in each of the datasets.
To estimate the sizes of the tests, we divide each of the $ K $ groups randomly in $ K $ subgroups, carry out the tests on those subgroups and then take the mean of the rejection rates within the groups as an estimate of the size in this subsampling process. If a group has too few observations to draw $ K $ distinct random subgroups from it, we omit it from the size study.
To estimate the powers of the tests in a dataset, we draw the $ K $ random subgroups from the $ K $ groups separately, carry out the tests.
The nominal level is 5\% everywhere.
The sizes and the powers of the tests are estimated based on 1000 independent replications of the above procedure.
Similar procedures of size and power study in real datasets were described in \cite{chakraborty2015wilcoxon} and \cite{gretton2012kernel}.
When the number of observations in the subsamples are small, we implement the SS test based on the permutation procedure with 1000 random permutations.

In the Canadian Temperature data, since the third group has only 5 observations, it is omitted from the size study in subsampling but included in the power study. The size of each of the 3 subgroups in this dataset is fixed at 4. Due to the small size of the subgroups, the SS test is implemented based on the permutation procedures.
\begin{table}
\caption{Estimated sizes and powers of the tests in the size and power study in real datasets under nominal level 5\%. In Canadian Temperature data, $ n_1 = n_2 = n_3 = 4 $, in Tecator data, $ n_1 = n_2 = n_3 = 23 $, and in Orthosis data, $ n_1 = n_2 = 3 $ and $ n_3 = 4 $.}
\begin{center}
\begin{tabular}{ccccccccccc}
	\hline
	Data & Size/Power	& SS	& CFF	& ZC	& F-max		& GPF	& F-type	& HR	& CAFB	& GS	\\\hline
	Canada	& Size	& 0.055	& 0.069	& 0.071	& 0.038	& 0.073	& 0.009	& 0.303	& 0.032	& 0.051	\\
	Canada	& Power	& 0.997	& 0.993	& 0.993	& 0.994		& 0.983	& 0.741		& 1	& 0.921	& 0.984	\\
	Tecator	& Size	& 0.053	& 0.047	& 0.049	& 0.049		& 0.048	& 0.042		& 0.067	& 0.033	& 0.041	\\
	Tecator	& Power	& 0.501	& 0.339	& 0.347	& 0.498		& 0.326	& 0.319		& 0.359	& 0.476	& 0.322	\\
	Orthosis & Size	& 0.041	& 0.028	& 0.094	& 0.020		& 0.104	& 0			& 0.138	& 0.034	& 0.050	\\
	Orthosis & Power	& 0.535	& 0.448	& 0.613	& 0.076		& 0.678	& 0.006		& 0.573	& 0.200	& 0.539	\\\hline
\end{tabular}
\end{center}
\label{table:realdatasizepower}
\end{table}
In \autoref{table:realdatasizepower}, we present the estimated sizes and the estimated powers of the tests in the size and power study. The nominal level is 5\% everywhere.
The SS test is implemented based on the permutation procedure in those two datasets. In the Tecator dataset, the tests are implemented using their usual asymptotic procedure.

In the power study of the tests in the Canadian Temperature data, we found that the estimated powers of all the tests except the F-type test and the CAFB test are very close to 1. In the corresponding size study, we found that the estimated size of the HR test is very far from the nominal level. On the other hand, the estimated size of the F-type test is rather low compared to the nominal level.
In the size study in the Tecator data, we found that all the tests have their estimated sizes close to the nominal level. In the corresponding power study, the SS test demonstrates the highest power, though the F-max test and the CAFB test have powers close to that of the SS test.
In the power study in the Orthosis data, we see that the ZC test, the GPF test, the HR test and the GS test exhibit higher estimated powers than the SS test. However, the estimated sizes of the ZC test, the GPF test and the HR test in the corresponding size study are found to be significantly higher than the nominal level, which indicates that their powers may be overestimated. The GS test has estimated size very close to the nominal level and slightly higher power than the SS test in this dataset.

In the case of the Tecator data, the second derivatives of the curves may also be considered as observations and a similar size and power study may be carried out. However, when we carried out this study, we found that all the p-values of the tests are 0, and the all the powers to be 1.

\section{Concluding remarks}
In this paper, we have developed a test for multisample comparison using spatial signs. The sample is assumed to consist of Hilbert space valued random elements. We described three implementations for this test: one is based on the asymptotic distribution of the test statistic, one is based on a bootstrap procedure and the other is based on a permutation procedure. The asymptotic implementation is faster, and it is recommended for moderate to large sample sizes. But for small sample sizes, its size tends to be lower than the nominal level. We observed the bootstrap and the permutation procedures to behave almost identically, and so choose the permutation procedure for the implementation in suitable scenarios. The permutation procedure, though computationally costly than the asymptotic procedure, exhibits satisfactory performance under the null hypothesis irrespective of the sample size unlike the asymptotic procedure. Further, its power is very similar to the asymptotic implementation. So, the permutation procedure is recommended for small sample sizes.

In our simulation study, we found that the test exhibits substantially superior performance than other mean based ANOVA for functional data in the literature when the underlying distribution is non-Gaussian with heavy tails or skewed.
In the Gaussian models considered by us, we found that our test beats most of the other mean based tests, if not all, and the performance of our test is almost always very close to the best among the mean based tests.
Our test also exhibits better performance in contaminated models than the other tests.

The above discussion implies the advantage of using our test in a general situation, where it is unknown whether the data are truly Gaussian and without any contamination.

%%%%%%%%%%%%%%%%%%%%%%%%%%%%%%%%%%%%%%%%%%%%%%
%% Single Appendix:                         %%
%%%%%%%%%%%%%%%%%%%%%%%%%%%%%%%%%%%%%%%%%%%%%%
\begin{appendix}
%\section*{???}%% if no title is needed, leave empty \section*{}.

\section{Description of the other tests and asymptotic powers}\label{asymptoticpowerexpression}
The test statistic of the CFF test (\cite{cuevas2004anova}) is
$ CFF_n = \sum_{k < l} n_k \left\| \bar{\mathbf{X}}_{k \cdot} - \bar{\mathbf{X}}_{l \cdot} \right\|^2 $.
We have used the implementation of the CFF test suggested in \cite{cuevas2004anova} for homoscedastic samples because all the simulation models we consider are homoscedastic and the real data we consider do not exhibit heteroscedasticity, and we take 1000 bootstrap samples for its implementation.

The test statistic of the ZC test (\cite{zhang2007statistical}) is
$ ZC_n = \sum_{k=1}^{K} n_k \left\| \bar{\mathbf{X}}_{k \cdot} - \bar{\mathbf{X}}_{\cdot \cdot} \right\|^2 $,
and the test statistic of the F-type test (\cite{shen2004f,zhang2011statistical}) is
\begin{align*}
F_n = \frac{( K - 1 )^{-1} \sum_{k=1}^{K} n_k \left\| \bar{\mathbf{X}}_{k \cdot} - \bar{\mathbf{X}}_{\cdot \cdot} \right\|^2}{( n - K )^{-1} \sum_{k=1}^{K} \sum_{i = 1}^{n_k} \left\| \mathbf{X}_{k i} - \bar{\mathbf{X}}_{k \cdot} \right\|^2} .
\end{align*}
The parameter estimation for implementing the ZC test and the F-type test for moderate or large sample sizes was done following the so-called `naive method' (see subsection 2.1 in \cite{gorecki2015comparison}) instead of the `bias-reduced method', because we did not see any significant difference in performances for those two implementations, and the `naive method' involves less computation. For small samples sizes, these two tests are implemented based on 1000 bootstrap samples using the R package `fdANOVA' (see \cite{gorecki2019fdanova}).

The test statistic of the GPF test (\cite{zhang2014one}) is
\begin{align*}
FG_n = \int_\mathcal{T} \frac{( K - 1 )^{-1} \sum_{k=1}^{K} n_k \left( \bar{\mathbf{X}}_{k \cdot}( t ) - \bar{\mathbf{X}}_{\cdot \cdot}( t ) \right)^2}{( n - K )^{-1} \sum_{k=1}^{K} \sum_{i = 1}^{n_k} \left( \mathbf{X}_{k i}( t ) - \bar{\mathbf{X}}_{k \cdot}( t ) \right)^2} \mathrm{d}t ,
\end{align*}
and the test statistic of the F-max test (\cite{zhang2018new}) is
\begin{align*}
FM_n = \sup_{t \in [ a, b ]} \frac{( K - 1 )^{-1} \sum_{k=1}^{K} n_k \left( \bar{\mathbf{X}}_{k \cdot}( t ) - \bar{\mathbf{X}}_{\cdot \cdot}( t ) \right)^2}{( n - K )^{-1} \sum_{k=1}^{K} \sum_{i = 1}^{n_k} \left( \mathbf{X}_{k i}( t ) - \bar{\mathbf{X}}_{k \cdot}( t ) \right)^2} .
\end{align*}
The GPF test and the F-max test are implemented according to the advice of the respective authors, and the number of bootstrap samples used for the implementation of the F-max test is 1000.

The GS test (\cite{gorecki2015comparison}) is a permutation test of ANOVA based on the basis function representation of the functional observations, where the test statistic is similar to that of the F-type test.
The GS test is implemented based on the function in the R package `fdANOVA' using the Fourier basis and the BIC criterion and 1000 permutations.

The CAFB test (\cite{cuesta2010simple}) is based on random projections of the sample. To implement the test, we generate 30 random directions following standard Brownian motion on which the sample observations are projected. Then, we carry out the univariate ANOVA procedure described in section 3 in \cite{brunner1997box} on the projected observations as suggested in \cite{cuesta2010simple}. The test of ANOVA for the functional observations is then decided based on the p-values of the univariate ANOVA. The reader is referred to \cite{cuesta2010simple} for further details on the implementation of this test, including the p-value computation.

We need to define a few quantities to describe the test statistic of the HR test (\cite{horvath2015introduction}), which is based on $ d $ principal component scores of the sample in the following way.
Define $ \boldsymbol{\Omega}_{n k} : [ a, b ] \times [ a, b ] \to \mathbb{R} $ as
\begin{align*}
\boldsymbol{\Omega}_{n k}\left( s , t \right)
= n_k^{-1} \sum_{i = 1}^{n_k} \left( \mathbf{X}_{k i}( s ) - \bar{\mathbf{X}}_{k \cdot}( s ) \right) \left( \mathbf{X}_{k i}( t ) - \bar{\mathbf{X}}_{k \cdot}( t ) \right)
\end{align*}
for $ k = 1, \ldots, K $.
Let $ \boldsymbol{\Omega}_n = n^{-1} \sum_{k = 1}^{K} n_k \boldsymbol{\Omega}_{n k} $ and $ \boldsymbol{\phi}_{n 1}, \ldots, \boldsymbol{\phi}_{n d} $ be the eigenfunctions corresponding to the $ d $ largest eigenvalues of $ \boldsymbol{\Omega}_n $.
Define
\begin{align*}
& \bar{\boldsymbol{\xi}}_{k i} = \left( \left\langle \mathbf{X}_{k i}, \boldsymbol{\phi}_{n 1} \right\rangle, \ldots, \left\langle \mathbf{X}_{k i}, \boldsymbol{\phi}_{n d} \right\rangle \right)^t ,\quad
\bar{\boldsymbol{\xi}}_{k \cdot} = n_k^{-1} \sum_{i = 1}^{n_k} \bar{\boldsymbol{\xi}}_{k i} ,\\
& \boldsymbol{\Psi}_{n k} = n_k^{-1} \sum_{i = 1}^{n_k} \left( \bar{\boldsymbol{\xi}}_{k i} - \bar{\boldsymbol{\xi}}_{k \cdot} \right) \left( \bar{\boldsymbol{\xi}}_{k i} - \bar{\boldsymbol{\xi}}_{k \cdot} \right)^t ,\quad
\bar{\boldsymbol{\xi}}_{\cdot \cdot} = \allowbreak \left( \sum_{k = 1}^{K} n_k \boldsymbol{\Psi}_{n k}^{-1} \right)^{-1} \sum_{k = 1}^{K} n_k \boldsymbol{\Psi}_{n k}^{-1} \bar{\boldsymbol{\xi}}_{k \cdot} .
\end{align*}
The test statistic of the HR test is defined as
\begin{align*}
HR_n = \sum_{k = 1}^{K} n_k \left( \bar{\boldsymbol{\xi}}_{k \cdot} - \bar{\boldsymbol{\xi}}_{\cdot \cdot} \right)^t \boldsymbol{\Psi}_{n k}^{-1} \left( \bar{\boldsymbol{\xi}}_{k \cdot} - \bar{\boldsymbol{\xi}}_{\cdot \cdot} \right) .
\end{align*}
In \cite{horvath2015introduction}, there was no clear direction on the number of principal components to be used for implementing the HR test. So, we have taken the number so that the chosen principal components explain at least 90\% of the variance present in the sample. In other words, we choose the number $ d $ such that the sum of the $ d $ largest eigenvalues of the sample covariance operator $ \boldsymbol{\Omega}_n $ is at least 90\% of the sum of all its eigenvalues.

We derive here the asymptotic powers of the CFF test, the ZC test and the HR test under the class of shrinking alternatives described in \eqref{asympower1}, where $ \mathcal{H} = L_2[ a, b ] $.

For $ \mathbf{u} = \left( \mathbf{u}_1, \ldots, \mathbf{u}_p \right)^t $, $ \mathbf{v} = \left( \mathbf{v}_1, \ldots, \mathbf{v}_p \right)^t $, where $ \mathbf{u}_1, \ldots, \mathbf{u}_p, \mathbf{v}_1, \ldots, \mathbf{v}_p \in \mathcal{H} $, define $ \mathbf{u}^t \mathbf{v} $ as $ \mathbf{u}^t \mathbf{v} = \sum_{i=1}^{p} \langle \mathbf{u}_i, \mathbf{v}_i \rangle $. Matrix operations between elements in $ \mathcal{H} $ or scalars are defined analogously.
For the result on the HR test, we need to further define the following quantities.
Define $ \boldsymbol{\Omega}_{k}, \boldsymbol{\Omega} : [ a, b ] \times [ a, b ] \to \mathbb{R} $ as
$ \boldsymbol{\Omega}_{k}\left( s , t \right)
= E\left[ \left( \mathbf{X}_{k 1}( s ) - E\left[ \mathbf{X}_{k 1}( s ) \right] \right) \left( \mathbf{X}_{k 1}( t ) - E\left[ \mathbf{X}_{k 1}( t ) \right] \right) \right] $
and
$ \boldsymbol{\Omega} = \sum_{k = 1}^{K} \lambda_k \boldsymbol{\Omega}_{k} $
for $ k = 1, \ldots, K $.
Let $ \boldsymbol{\phi}_{1}, \ldots, \boldsymbol{\phi}_{d} $ be the eigenfunctions corresponding to the $ d $ largest eigenvalues of $ \boldsymbol{\Omega} $.
Define
$ \mathbf{d}_k = \left( \left\langle \boldsymbol{\delta}_k, \boldsymbol{\phi}_{1} \right\rangle, \ldots, \left\langle \boldsymbol{\delta}_k, \boldsymbol{\phi}_{d} \right\rangle \right)^t $,
$ \boldsymbol{\xi}_{k i} = \left( \left\langle \mathbf{X}_{k i}, \boldsymbol{\phi}_{1} \right\rangle, \ldots, \left\langle \mathbf{X}_{k i}, \boldsymbol{\phi}_{d} \right\rangle \right)^t $,
$ \boldsymbol{\Psi}_{k} = E\left[ \left( \boldsymbol{\xi}_{k 1} - E\left[ \boldsymbol{\xi}_{k 1} \right] \right) \left( \boldsymbol{\xi}_{k 1} - E\left[ \boldsymbol{\xi}_{k 1} \right] \right)^t \right] $
and
\begin{align*}
\mathbf{A}_{K d \times K d} = \mathbb{I}_{K d} -
\begin{bmatrix}
\sqrt{\lambda_1} \boldsymbol{\Psi}_{1}^{-1/2}\\
\vdots\\
\sqrt{\lambda_K} \boldsymbol{\Psi}_{K}^{-1/2}
\end{bmatrix}
\left( \sum_{k=1}^{K} \lambda_k \boldsymbol{\Psi}_{k}^{-1} \right)^{-1}
\left[ \sqrt{\lambda_1} \boldsymbol{\Psi}_{1}^{-1/2}, \ldots, \sqrt{\lambda_K} \boldsymbol{\Psi}_{K}^{-1/2} \right] .
\end{align*}
One can verify that $ \mathbf{A}_{K d \times K d} $ is an idempotent matrix of degree $ d ( K - 1 ) $. Let $ \mathbf{V}_{K d \times K d} $ be such that
\begin{align*}
\mathbf{A}_{K d \times K d} = \mathbf{V}_{K d \times K d}
\begin{bmatrix}
\mathbb{I}_{d ( K - 1 )} 			& \mathbf{0}_{d ( K - 1 ) \times d} \\
\mathbf{0}_{d \times d ( K - 1 )} 	& \mathbf{0}_{d \times d}
\end{bmatrix}
\mathbf{V}_{K d \times K d}^t .
\end{align*}
Define
$ \mathbf{M}_{K d \times 1} = \mathbf{V}_{K d}^t
\left[ \sqrt{\lambda_1} \boldsymbol{\Psi}_{1}^{-1/2} \mathbf{d}_1^t, \ldots, \sqrt{\lambda_K} \boldsymbol{\Psi}_{K}^{-1/2} \mathbf{d}_K^t \right]^t $,
and let $ \widetilde{\mathbf{M}}_{d ( K - 1 ) \times 1} $ be the truncated column vector obtained by dropping the bottom $ d $ elements from $ \mathbf{M}_{K d \times 1} $.
We have the following theorem on the asymptotic powers of the three tests.
\begin{theorem} \label{thm:5}
Consider the class of shrinking alternatives described in \eqref{asympower1}, and assume $ E[ \| \mathbf{X} \|^2 ] < \infty $, where $ \mathbf{X} $ has distribution $ P_0 $, and that $ n^{-1} n_k \to \lambda_k \in ( 0, 1 ) $ for all $ k $ as $ n \to \infty $.
Let $ \mathbf{Y}_1, \ldots, \mathbf{Y}_K $ be independent zero-mean Gaussian random elements with the covariance operator of $ P_0 $ denoted as $ \boldsymbol{\Gamma} $.
For $ \alpha \in \left( 0, 1 \right) $, let $ CFF( \alpha ) $, $ ZC( \alpha ) $ and $ HR( \alpha ) $ denote the asymptotic powers at level $ \alpha $  of the CFF test, the ZC test, and the HR test, respectively, under \eqref{asympower1} with $ n \to \infty $.
Then, we have the following.
\begin{enumerate}[label = (\alph*), ref = (\alph*)]
\item 
$ CFF( \alpha )
= P\left[ \sum_{k < l} \lambda_k \left\| \lambda_k \boldsymbol{\delta}_k + \mathbf{Y}_k - \lambda_l^{-1/2} \lambda_k^{1/2} \left( \lambda_l \boldsymbol{\delta}_l + \mathbf{Y}_l \right) \right\|^2 
\ge Q_1\left( 1 - \alpha \right) \right] $,
where $ Q_1\left( 1 - \alpha \right) $ is the $ \left( 1 - \alpha \right) $-quantile of the distribution of the random variable
$ \sum_{k < l} \lambda_k \left\| \mathbf{Y}_k - \lambda_l^{-1/2} \lambda_k^{1/2} \mathbf{Y}_l \right\|^2 $.

\item
Let $ \gamma_1, \gamma_2, \ldots $ be the sequence of eigenvalues of $ \boldsymbol{\Gamma} $, and $ \boldsymbol{\chi}_1^2, \boldsymbol{\chi}_2^2, \ldots $ be an infinite sequence of independent $ \chi_{K - 1}^2 $ random variables.
Further, let $ \mathbf{I}_K $ denote a $ K \times K $ matrix of identity operators, $ \mathbf{p}_0 = \left( \sqrt{\lambda_1}, \ldots, \sqrt{\lambda_K} \right)^t $, and $ \mathbf{Z} = \left( \sqrt{\lambda_1} \boldsymbol{\delta}_1 + \mathbf{Y}_1, \ldots, \sqrt{\lambda_K} \boldsymbol{\delta}_K + \mathbf{Y}_K \right)^t $.
Then,
$ ZC( \alpha )
= P[ \mathbf{Z}^t \left( \mathbf{I}_K - \mathbf{p}_0 \mathbf{p}_0^t \right) \mathbf{Z}
\ge Q_2\left( 1 - \alpha \right) ] $,
where $ Q_2\left( 1 - \alpha \right) $ is the $ \left( 1 - \alpha \right) $-quantile of the distribution of the random variable
$ \sum_{i=1}^{\infty} \gamma_i \boldsymbol{\chi}_i^2 $.

\item
Let $ \boldsymbol{\chi} $ be a random variable following a non-central chi-square distribution with $ \left( \left( K - 1 \right) d \right) $ degrees of freedom and noncentrality parameter $ \big\| \widetilde{\mathbf{M}}_{d ( K - 1 ) \times 1} \big\|^2 $.
Then,
$ HR( \alpha )
= P\left[ \boldsymbol{\chi}
\ge Q_3\left( 1 - \alpha \right) \right] $,
where $ Q_3\left( 1 - \alpha \right) $ is the $ \left( 1 - \alpha \right) $-quantile of the central chi-square distribution with $ \left( \left( K - 1 \right) d \right) $ degrees of freedom.
\end{enumerate}
\end{theorem}

\section{Proof of mathematical results} \label{sec:4}
\begin{proof}[Proof of \autoref{thm:1}]
Define $ \tilde{\mathbf{s}}\left( \mathbf{X} , \mathbf{Y} \right)
= \mathbf{s}\left( \mathbf{X} - \mathbf{Y} \right) - E\left[ \mathbf{s}\left( \mathbf{X} - \mathbf{Y} \right) \right] $, where $ \mathbf{X} $ and $ \mathbf{Y} $ are independent random elements.
Define
\begin{align*}
& \mathbf{S}_k
= \frac{1}{n} \sum_{l=1}^{K} n_l \left[ \frac{1}{n_k} \sum_{i_k = 1}^{n_k} E\left[ \tilde{\mathbf{s}}\left( \mathbf{X}_{k i_k} , \mathbf{X}_{l 1} \right) \mid \mathbf{X}_{k i_k} \right]
+ \frac{1}{n_l} \sum_{i_l = 1}^{n_l} E\left[ \tilde{\mathbf{s}}\left( \mathbf{X}_{k 1} , \mathbf{X}_{l i_l} \right) \mid \mathbf{X}_{l i_l} \right] \right] \\
& \text{and }
\mathbf{V}_{k l}
= \frac{1}{n_k n_l} \sum_{i_k = 1}^{n_k} \sum_{i_l = 1}^{n_l} \left[ \tilde{\mathbf{s}}\left( \mathbf{X}_{k i_k} , \mathbf{X}_{l i_l} \right) 
- E\left[ \tilde{\mathbf{s}}\left( \mathbf{X}_{k i_k} , \mathbf{X}_{l 1} \right) \mid \mathbf{X}_{k i_k} \right]
- E\left[ \tilde{\mathbf{s}}\left( \mathbf{X}_{k 1} , \mathbf{X}_{l i_l} \right) \mid \mathbf{X}_{l i_l} \right] \right] ,
\end{align*}
where $ i_k = 1, \cdots, n_k $, $ i_l = 1, \cdots, n_l $ and $ k, l = 1, \cdots, K $.
We have
$ \bar{\mathbf{R}}_k - E\left[ \bar{\mathbf{R}}_k \right]
= \mathbf{S}_k + n^{-1} \sum_{l=1}^{K} n_l \mathbf{V}_{k l} $.
%%%%%%%%%%%%%%%%
Also, define
$ \mathbf{W}_n = \left( \sqrt{n_1} \mathbf{S}_1, \ldots, \sqrt{n_K} \mathbf{S}_K \right) $,
and
$ \bar{\mathbf{s}}\left( \mathbf{X} , \mathbf{Y} \right)
= \tilde{\mathbf{s}}\left( \mathbf{X} , \mathbf{Y} \right) 
- E\left[ \tilde{\mathbf{s}}\left( \mathbf{X} , \mathbf{Y} \right) \mid \mathbf{X} \right]
- E\left[ \tilde{\mathbf{s}}\left( \mathbf{X} , \mathbf{Y} \right) \mid \mathbf{Y} \right] $.
Note that
\begin{align}
& E\left[ \bar{\mathbf{s}}\left( \mathbf{X} , \mathbf{Y} \right) \mid \mathbf{X} \right]
= E\left[ \bar{\mathbf{s}}\left( \mathbf{X} , \mathbf{Y} \right) \mid \mathbf{Y} \right]
= E\left[ \bar{\mathbf{s}}\left( \mathbf{X} , \mathbf{Y} \right) \right]
= \mathbf{0}
\label{thm1:eqcond}
\\
& \text{and } \left\| \bar{\mathbf{s}}\left( \mathbf{X} , \mathbf{Y} \right) \right\| 
\le 6 .
\label{thm1:eqbound}
\end{align}
From \eqref{thm1:eqcond}, we get that for every case other than $ ( k, i_k, l_1, i_{l_1} ) = ( k, j_k, l_2, j_{l_2} ) $ and $ ( k, i_k, l_1, i_{l_1} ) = ( l_2, j_{l_2}, k, j_k ) $,
\begin{align}
& E\left[ \left\langle \bar{\mathbf{s}}\left( \mathbf{X}_{k i_k}, \mathbf{X}_{l_1 i_{l_1}} \right), \; \bar{\mathbf{s}}\left( \mathbf{X}_{k j_k}, \mathbf{X}_{l_2 j_{l_2}} \right) \right\rangle \right] = 0 .
\label{thm1:eqinner}
\end{align}
Since $ \mathbf{V}_{k l} = 
( n_k n_l )^{-1} \sum_{i_k = 1}^{n_k} \sum_{i_l = 1}^{n_l} \bar{\mathbf{s}}\left( \mathbf{X}_{k i_k} , \mathbf{X}_{l i_l} \right) $, from \eqref{thm1:eqbound} and \eqref{thm1:eqinner} we get that for all $ k $,
\begin{align}
& E\left[ \left\| \sqrt{n_k} \frac{1}{n} \sum_{l=1}^{K} n_l \mathbf{V}_{k l} \right\|^2 \right]
\longrightarrow 0
\, \text{ as } n \to \infty ,
\nonumber\\
& \text{which implies }
\left\| \left[ \mathbf{U}_n - E\left[ \mathbf{U}_n \right] \right] - \mathbf{W}_n \right\|
= \sqrt{ \sum_{k = 1}^{K} \left\| \sqrt{n_k} \frac{1}{n} \sum_{l=1}^{K} n_l \mathbf{V}_{k l} \right\|^2 }
\stackrel{P}{\longrightarrow} 0
\label{thm1:eq4}
\end{align}
as $ n \to \infty $.
Next, we consider $ \mathbf{S}_k $.
It can be verified that
\begin{align}
\text{Cov}\left( \sqrt{n_{k_1}} \mathbf{S}_{k_1},\, \sqrt{n_{k_2}} \mathbf{S}_{k_2} \right) 
& = \frac{\sqrt{n_{k_1} n_{k_2}}}{n} \sum_{l=1}^{K} \frac{n_l}{n} \left[ \mathbf{C}( k_1, k_2, l ) - \mathbf{C}( k_1, l, k_2 ) - \mathbf{C}( l, k_2, k_1 ) \right] \nonumber\\
& \quad
+ \sum_{l_1 = 1}^{K} \sum_{l_2 = 1}^{K} \frac{n_{l_1} n_{l_2}}{n^2}
\mathbf{C}( l_1, l_2, k_1 ) \mathbb{I}( k_1 = k_2 ) .
\label{thm1:eq7}
\end{align}
Since $ n^{-1} n_k \to \lambda_k \in ( 0, 1 ) $ for all $ k $ as $ n \to \infty $, from \eqref{thm1:eq7}, we have for all $ k_1 $ and $ k_2 $,
\begin{align}
& \text{Cov}\left( \sqrt{n_{k_1}} \mathbf{S}_{k_1},\, \sqrt{n_{k_2}} \mathbf{S}_{k_2} \right)
\longrightarrow \boldsymbol{\sigma}_{k_1 k_2} \quad \text{as } n \to \infty .
\label{thm1:eq8}
\end{align}
Since $ E\left[ \mathbf{W}_n \right] = \mathbf{0} $,
from \eqref{thm1:eq8} and an application of Theorem 1.1 in \cite{kundu2000central}, we get
$ \mathbf{W}_n \stackrel{w}{\longrightarrow} G\left( \mathbf{0}, \boldsymbol{\Sigma} \right) $ as $ n \to \infty $, which, along with \eqref{thm1:eq4}, implies that
$ \left[ \mathbf{U}_n - E\left[ \mathbf{U}_n \right] \right]
\stackrel{w}{\longrightarrow} G\left( \mathbf{0}, \boldsymbol{\Sigma} \right) $ as $ n \to \infty $.
\end{proof}

\begin{proof}[Proof of \autoref{coro:1}]
When $ \text{H}_0 $ in \eqref{h_0} is true, $ \mathbf{X}_{k 1} $ is an independent copy of $ \mathbf{X}_{l 1} $ for $ k \neq l $, which, along with the fact that $ \mathbf{s}( \mathbf{0} ) = \mathbf{0} $ by definition, gives $ E[ \mathbf{s}( \mathbf{X}_{k 1} - \mathbf{X}_{l 1} ) ] = \mathbf{0} $ for all $ k, l $. So, $ E\left[ \sqrt{n_k} \bar{\mathbf{R}}_k \right]
= \sqrt{n_k} \sum_{l=1}^{K} n^{-1} n_l E\left[ \mathbf{s}\left( \mathbf{X}_{k 1} - \mathbf{X}_{l 1} \right) \right] = \mathbf{0} $ for all $ k $, and consequently, $ E\left[ \mathbf{U}_n \right] = \mathbf{0} $. Therefore from \autoref{thm:1}, we have $ \mathbf{U}_n \stackrel{w}{\longrightarrow} G\left( \mathbf{0}, \boldsymbol{\Sigma} \right) $ as $ n \to \infty $. Hence, from an application of the mapping theorem \cite[Theorem 2.7, p.~21]{billingsley2013convergence}, we get
$ \| \mathbf{U}_n \|^2 \stackrel{w}{\longrightarrow} \| \mathbf{W} \|^2 $ as $ n \to \infty $, where $ \mathbf{W} $ is a random element having distribution $ G\left( \mathbf{0}, \boldsymbol{\Sigma} \right) $.
\end{proof}

\begin{proof}[Proof of \autoref{thm:cov}]
Note that the class of compact operators on a separable Hilbert space is closed \cite[p.~164]{bhatia2009notes} and separable \cite[p.~168]{bhatia2009notes} with respect to the operator norm. So, the class of compact operators on a separable Hilbert space is a separable Banach space with respect to the operator norm, and we denote it as $ \mathcal{C} $.
%Also, if $ \mathbf{X} $ is a random element in a separable Banach space with $ E[ \| \mathbf{X} \|^2 ] < \infty $, then its covariance operator is compact (see \cite{chobanjan1973compactness}, \cite{baker1981compact}). We shall need these two facts in the subsequent proof.
We have
\begin{align}
&\mathbf{C}_n( k_1, k_2, l )
= \frac{n_l}{n_l - 1} \mathbf{C}_n^{(1)}( k_1, k_2, l ) - \mathbf{C}_n^{(2)}( k_1, l ) \otimes \mathbf{C}_n^{(2)}( k_2, l ) ,
\label{eq:cov0}
\\
& \text{where }
\mathbf{C}_n^{(1)}( j, k, l )
= \frac{1}{n_l n_{j} n_{k}} \sum_{i_l = 1}^{n_l} \sum_{i_{j} = 1}^{n_{j}} \sum_{i_{k} = 1}^{n_{k}} \mathbf{s}\left( \mathbf{X}_{j i_{j}} - \mathbf{X}_{l i_l} \right) \otimes \mathbf{s}\left( \mathbf{X}_{k i_{k}} - \mathbf{X}_{l i_l} \right)
\nonumber\\
& \text{and } \mathbf{C}_n^{(2)}( k, l )
= \frac{1}{n_k n_l} \sum_{i_k = 1}^{n_k} \sum_{i_l = 1}^{n_l} \mathbf{s}\left( \mathbf{X}_{k i_k} - \mathbf{X}_{l i_l} \right) \; \text{for } j, k, l \in \{ 1, \ldots, K \} .\nonumber
\end{align}
First, we shall consider the term $ \mathbf{C}_n^{(1)}( k_1, k_2, l ) $ and the following cases of the triplet $ ( k_1, k_2, l ) $ separately: (a) $ k_1, k_2, l $ are all distinct, (b) $ k_1 = k_2 \neq l $, (c) $ k_1 = l \neq k_2 $, (d) $ k_2 = l \neq k_1 $ and (e) $ k_1 = k_2 = l $. These five cases cover all possible cases of $ ( k_1, k_2, l ) $.

Assume that $ k_1, k_2, l $ are all distinct. Define $ \boldsymbol{\Phi}_1\left( \mathbf{x}, \mathbf{y}, \mathbf{z} \right) = \mathbf{s}\left( \mathbf{x} - \mathbf{z} \right) \otimes \mathbf{s}\left( \mathbf{y} - \mathbf{z} \right) $ for $ \mathbf{x}, \mathbf{y}, \mathbf{z} \in \mathcal{H} $. Since $ \boldsymbol{\Phi}_1\left( \mathbf{x}, \mathbf{y}, \mathbf{z} \right) $ is a finite-rank operator, it belongs to the separable Banach space $ \mathcal{C} $.
Therefore,
\begin{align*}
\mathbf{C}_n^{(1)}( k_1, k_2, l )
= \frac{1}{n_l n_{k_1} n_{k_2}} \sum_{i_l = 1}^{n_l} \sum_{i_{k_1} = 1}^{n_{k_1}} \sum_{i_{k_2} = 1}^{n_{k_2}} \boldsymbol{\Phi}_1\left( \mathbf{X}_{k_1 i_{k_1}} , \mathbf{X}_{k_2 i_{k_2}} , \mathbf{X}_{l i_l} \right)
\end{align*}
is a 3-sample UB-statistic \cite[p.~15]{borovskikh1996u}. Since $ \left\| \boldsymbol{\Phi}_1 \right\| \le 1 $ and $ n^{-1} n_k \to \lambda_k \in ( 0, 1 ) $ for all $ k $ as $ n \to \infty $, from Theorem 3.2.1 in \cite[p.~79]{borovskikh1996u}, we get
\begin{align}
\mathbf{C}_n^{(1)}( k_1, k_2, l ) 
\stackrel{a.s.}{\longrightarrow}
E\left[ \mathbf{s}\left( \mathbf{X}_{k_1 1} - \mathbf{X}_{l 1} \right) \otimes \mathbf{s}\left( \mathbf{X}_{k_2 1} - \mathbf{X}_{l 1} \right) \right]
\label{eq:cov1}
\end{align}
as $ n \to \infty $ in the operator norm.
Next, assume $ k_1 = k_2 \neq l $. Define $ \boldsymbol{\Phi}_2\left( \mathbf{x}_1, \mathbf{x}_2 ; \mathbf{y} \right) = 2^{-1} [ \mathbf{s}\left( \mathbf{x}_1 - \mathbf{y} \right) \otimes \mathbf{s}\left( \mathbf{x}_2 - \mathbf{y} \right) + \mathbf{s}\left( \mathbf{x}_2 - \mathbf{y} \right) \otimes \mathbf{s}\left( \mathbf{x}_1 - \mathbf{y} \right) ] $ and $ \boldsymbol{\Phi}_3\left( \mathbf{x}, \mathbf{y} \right) = \mathbf{s}\left( \mathbf{x} - \mathbf{y} \right) \otimes \mathbf{s}\left( \mathbf{x} - \mathbf{y} \right) $. Clearly, $ \boldsymbol{\Phi}_2\left( \mathbf{x}_1, \mathbf{x}_2 ; \mathbf{y} \right), \allowbreak \boldsymbol{\Phi}_3\left( \mathbf{x}, \mathbf{y} \right) \in \mathcal{C} $, and since $ \left\| \boldsymbol{\Phi}_2 \right\| \le 1 $ and $ \left\| \boldsymbol{\Phi}_3 \right\| \le 1 $, from Theorem 3.2.1 in \cite[p.~79]{borovskikh1996u}, we get
\begin{align}
\mathbf{C}_n^{(1)}( k_1, k_2, l )
& = \frac{n_{k_1} - 1}{n_{k_1}} \frac{1}{n_l {n_{k_1} \choose 2}} \sum_{i_l = 1}^{n_l} \sum_{1 \le i_{k_1} < j_{k_1} \le n_{k_1}} \boldsymbol{\Phi}_2\left( \mathbf{X}_{k_1 i_{k_1}} , \mathbf{X}_{k_1 j_{k_1}} ; \mathbf{X}_{l i_l} \right) \nonumber\\
& \quad
+ \frac{1}{n_{k_1}} \frac{1}{n_l n_{k_1}} \sum_{i_l = 1}^{n_l} \sum_{i_{k_1} = 1}^{n_{k_1}} \boldsymbol{\Phi}_3\left( \mathbf{X}_{k_1 i_{k_1}} , \mathbf{X}_{l i_l} \right)
\nonumber\\
& \stackrel{a.s.}{\longrightarrow}
E\left[ \mathbf{s}\left( \mathbf{X}_{k_1 1} - \mathbf{X}_{l 1} \right) \otimes \mathbf{s}\left( \mathbf{X}_{k_1 2} - \mathbf{X}_{l 1} \right) \right]
\label{eq:cov2}
\end{align}
as $ n \to \infty $ in the operator norm.
Next, assume $ k_1 = l \neq k_2 $. Define $ \boldsymbol{\Phi}_4\left( \mathbf{x}_1, \mathbf{x}_2 ; \mathbf{y} \right) = 2^{-1} [ \mathbf{s}\left( \mathbf{x}_1 - \mathbf{x}_2 \right) \otimes \mathbf{s}\left( \mathbf{y} - \mathbf{x}_2 \right) + \mathbf{s}\left( \mathbf{x}_2 - \mathbf{x}_1 \right) \otimes \mathbf{s}\left( \mathbf{y} - \mathbf{x}_1 \right) ] $. Since $ \left\| \boldsymbol{\Phi}_4 \right\| \le 1 $, we get
\begin{align}
\mathbf{C}_n^{(1)}( k_1, k_2, l ) 
& = \frac{n_{k_1} - 1}{n_{k_1}} \frac{1}{{n_{k_1} \choose 2} n_{k_2}} \sum_{1 \le i_{k_1} < j_{k_1} \le n_{k_1}} \sum_{i_{k_2} = 1}^{n_{k_2}} \boldsymbol{\Phi}_4\left( \mathbf{X}_{k_1 i_{k_1}} , \mathbf{X}_{k_1 j_{k_1}} ; \mathbf{X}_{k_2 i_{k_2}} \right)
\nonumber\\
& \stackrel{a.s.}{\longrightarrow}
E\left[ \mathbf{s}\left( \mathbf{X}_{k_1 1} - \mathbf{X}_{k_1 2} \right) \otimes \mathbf{s}\left( \mathbf{X}_{k_2 1} - \mathbf{X}_{k_1 2} \right) \right]
\label{eq:cov3}
\end{align}
as $ n \to \infty $ in the operator norm.
Next, assume $ k_2 = l \neq k_1 $. Arguing similarly as in the case $ k_1 = l \neq k_2 $, we get
\begin{align}
\mathbf{C}_n^{(1)}( k_1, k_2, l )
& \stackrel{a.s.}{\longrightarrow}
E\left[ \mathbf{s}\left( \mathbf{X}_{k_1 1} - \mathbf{X}_{k_2 2} \right) \otimes \mathbf{s}\left( \mathbf{X}_{k_2 1} - \mathbf{X}_{k_2 2} \right) \right]
\label{eq:cov4}
\end{align}
as $ n \to \infty $ in the operator norm.
Finally, assume $ k_1 = k_2 = l $. Define
\begin{align*}
\boldsymbol{\Phi}_5\left( \mathbf{x}_1, \mathbf{x}_2, \mathbf{x}_3 \right) = (3!)^{-1} \sum_{\sigma} \mathbf{s}\left( \mathbf{x}_{\sigma(1)} - \mathbf{x}_{\sigma(3)} \right) \otimes \mathbf{s}\left( \mathbf{x}_{\sigma(2)} - \mathbf{x}_{\sigma(3)} \right) ,
\end{align*}
where the sum is carried out over all 6 permutations of $ ( 1, 2, 3 ) $, and
\begin{align*}
\boldsymbol{\Phi}_6\left( \mathbf{x}_1, \mathbf{x}_2 \right) = 2^{-1} \left[ \mathbf{s}\left( \mathbf{x}_1 - \mathbf{x}_2 \right) \otimes \mathbf{s}\left( \mathbf{x}_1 - \mathbf{x}_2 \right) + \mathbf{s}\left( \mathbf{x}_2 - \mathbf{x}_1 \right) \otimes \mathbf{s}\left( \mathbf{x}_2 - \mathbf{x}_1 \right) \right] .
\end{align*}
Since $ \left\| \boldsymbol{\Phi}_5 \right\| \le 1 $ and $ \left\| \boldsymbol{\Phi}_6 \right\| \le 1 $, from Theorem 3.1.1 in \cite[p.~73]{borovskikh1996u}, we get
\begin{align}
& \mathbf{C}_n^{(1)}( k_1, k_2, l ) \nonumber\\
& = \frac{1}{n_{k_1}^3} \sum_{h_{k_1} = 1}^{n_{k_1}} \sum_{i_{k_1} = 1}^{n_{k_1}} \sum_{j_{k_1} = 1}^{n_{k_1}} \mathbf{s}\left( \mathbf{X}_{k_1 i_{k_1}} - \mathbf{X}_{k_1 h_{k_1}} \right) \otimes \mathbf{s}\left( \mathbf{X}_{k_1 j_{k_1}} - \mathbf{X}_{k_1 h_{k_1}} \right) 
\nonumber\\
& = \frac{(n_{k_1} - 1) (n_{k_1} - 2)}{n_{k_1}^2} \frac{1}{{n_{k_1} \choose 3}} \sum_{1 \le h_{k_1} < i_{k_1} < j_{k_1} \le n_{k_1}} \boldsymbol{\Phi}_5\left( \mathbf{X}_{k_1 i_{k_1}} , \mathbf{X}_{k_1 j_{k_1}} , \mathbf{X}_{k_1 h_{k_1}} \right) \nonumber\\
& \quad 
+ \frac{n_{k_1} - 1}{n_{k_1}^2} \frac{1}{{n_{k_1} \choose 2}} \sum_{1 \le i_{k_1} < j_{k_1} \le n_{k_1}} \boldsymbol{\Phi}_3\left( \mathbf{X}_{k_1 i_{k_1}} , \mathbf{X}_{k_1 j_{k_1}} \right) \nonumber\\
& \stackrel{a.s.}{\longrightarrow}
E\left[ \mathbf{s}\left( \mathbf{X}_{k_1 1} - \mathbf{X}_{k_1 3} \right) \otimes \mathbf{s}\left( \mathbf{X}_{k_1 2} - \mathbf{X}_{k_1 3} \right) \right]
\label{eq:cov5}
\end{align}
as $ n \to \infty $ in the operator norm.
Therefore, from \eqref{eq:cov1}, \eqref{eq:cov2}, \eqref{eq:cov3}, \eqref{eq:cov4} and \eqref{eq:cov5}, we get that for all $ k_1, k_2, l \in \{ 1, \ldots, K \} $,
\begin{align}
\mathbf{C}_n^{(1)}( k_1, k_2, l )
& \stackrel{a.s.}{\longrightarrow}
E\left[ \mathbf{s}\left( \mathbf{X}_{k_1 1} - \mathbf{X}_{l 1} \right) \otimes \mathbf{s}\left( \mathbf{X}_{k_2 1} - \mathbf{X}_{l 1} \right) \right] \nonumber\\
& = E\left[ E\left[ \mathbf{s}( \mathbf{X}_{k_1 1} - \mathbf{X}_{l 1} ) \midil \mathbf{X}_{l 1} \right] \otimes E\left[ \mathbf{s}( \mathbf{X}_{k_2 1} - \mathbf{X}_{l 1} ) \midil \mathbf{X}_{l 1} \right] \right]
\label{eq:cov6}
\end{align}
as $ n \to \infty $ in the operator norm, where, if $ k_1 = l $ or $ k_2 = l $, $ \mathbf{X}_{k_1 1} $ or $ \mathbf{X}_{k_2 1} $ are considered to be independent copies of $ \mathbf{X}_{l 1} $, respectively.
Next, we consider the term $ \mathbf{C}_n^{(2)}( k, l ) $.
Note that if $ k = l $, $ \mathbf{C}_n^{(2)}( k, l ) = \mathbf{0} $ and $ E\left[ \mathbf{s}( \mathbf{X}_{k_1 1} - \mathbf{X}_{l 1} ) \right] = \mathbf{0} $. When $ k \neq l $, $ \mathbf{C}_n^{(2)}( k, l ) $ is a 2-sample UB statistic, and we have
\begin{align}
\mathbf{C}_n^{(2)}( k, l ) \stackrel{a.s.}{\longrightarrow} E\left[ \mathbf{s}( \mathbf{X}_{k_1 1} - \mathbf{X}_{l 1} ) \right]
\label{eq:cov7}
\end{align}
as $ n \to \infty $ in the operator norm. Therefore, from \eqref{eq:cov0}, \eqref{eq:cov6} and \eqref{eq:cov7}, we have
\begin{align}
\mathbf{C}_n( k_1, k_2, l )
\stackrel{a.s.}{\longrightarrow}
\mathbf{C}( k_1, k_2, l ) .
\label{eq:cov8}
\end{align}
as $ n \to \infty $ in the operator norm. Since $ n^{-1} n_k \to \lambda_k $ for all $ k $ as $ n \to \infty $, from \eqref{eq:cov8} we have
$ \boldsymbol{\sigma}^{(n)}_{k_1 k_2} \stackrel{a.s.}{\longrightarrow} \boldsymbol{\sigma}_{k_1 k_2} $ as $ n \to \infty $ in the operator norm for all $ k_1 , k_2 $, and hence $ \widehat{\boldsymbol{\Sigma}}_n \stackrel{a.s.}{\longrightarrow} \boldsymbol{\Sigma} $ as $ n \to \infty $ in the operator norm.
\end{proof}

\begin{proof}[Proof of \autoref{thm:bootstrapperm}]
The proof for the test based on the asymptotic procedure follows from \autoref{coro:1}.

%%%%%%%%%%%%%%%%%%%%%%%%%%%%%
We now prove the theorem for the bootstrap implementation.
Let $ \mathcal{S}{( n_1, \ldots, n_K )} $ denote the original sample $ \{ \mathbf{X}_{k i} \midil k = 1, \ldots, K ; i = 1, \ldots, n_k \} $ and $ \mathcal{S} = \cup \{ \mathcal{S}_{( n_1, \ldots, n_K )} \midil n_1 \ge 1, \ldots, n_K \ge 1 \} $. Let $ \{ \mathbf{X}_{k i}^* \midil k = 1, \ldots, K ; i = 1, \ldots, n_k \} $ denote a bootstrap sample.
Define
$ \mathbf{R}^*( \mathbf{x} ) = n^{-1} \sum_{k=1}^{K} \sum_{i=1}^{n_k} \mathbf{s}( \mathbf{x} - \mathbf{X}_{k i}^* ) $,
$ \bar{\mathbf{R}}_k^* = n_k^{-1} \sum_{i_k = 1}^{n_k} \mathbf{R}^*\left( \mathbf{X}_{k i_k}^* \right) $,
$ \mathbf{U}_n^* = \left( \sqrt{n_1} \bar{\mathbf{R}}_1^*, \ldots, \sqrt{n_K} \bar{\mathbf{R}}_K^* \right) $
and $ SS_n^* = \sum_{k=1}^{K} n_k \left\| \bar{\mathbf{R}}_k^* \right\|^2 = \| \mathbf{U}_n^* \|^2 $.
Define $ P^*[\, \cdot \,] = P[\, \cdot \midil \mathcal{S} ] $, $ E^*[\, \cdot \,] = E[\, \cdot \midil \mathcal{S} ] $, $ \text{Cov}^*[\, \cdot, \cdot \,] = \text{Cov}[\, \cdot, \cdot \midil \mathcal{S} ] $ and $ E^*[\, \cdot \midil \mathbf{Z} ] = E[\, \cdot \midil \mathbf{Z}, \mathcal{S} ] $ for any random element $ \mathbf{Z} $.
Also, let $ `` \stackrel{P^*}{\longrightarrow} " $ imply convergence in probability given $ \mathcal{S} $, and $ `` \stackrel{w^*}{\longrightarrow} " $ imply convergence in distribution given $ \mathcal{S} $.
Define $ \tilde{\mathbf{s}}^*\left( \mathbf{X}_{k i_k}^* , \mathbf{X}_{l i_l}^* \right)
= \mathbf{s}\left( \mathbf{X}_{k i_k}^* - \mathbf{X}_{l i_l}^* \right) - E^*\left[ \mathbf{s}\left( \mathbf{X}_{k i_k}^* - \mathbf{X}_{l i_l}^* \right) \right] $.
Note that
$ E^*\left[ \tilde{\mathbf{s}}^*\left( \mathbf{X}_{k i_k}^* , \mathbf{X}_{l i_l}^* \right) \mid \mathbf{X}_{k i_k}^* \right]
\allowbreak = E^*\left[ \tilde{\mathbf{s}}^*\left( \mathbf{X}_{k i_k}^* , \mathbf{X}_{l 1}^* \right) \mid \mathbf{X}_{k i_k}^* \right] $ and
$ E\left[ \tilde{\mathbf{s}}^*\left( \mathbf{X}_{k i_k}^* , \mathbf{X}_{l i_l}^* \right) \mid \mathbf{X}_{l i_l}^* \right]
= \; \allowbreak E\left[ \tilde{\mathbf{s}}^*\left( \mathbf{X}_{k 1}^* , \mathbf{X}_{l i_l}^* \right) \mid \mathbf{X}_{l i_l}^* \right] $.
Also, define
\begin{align*}
& \mathbf{S}_k^*
= \frac{1}{n} \sum_{l=1}^{K} n_l \left[ \frac{1}{n_k} \sum_{i_k = 1}^{n_k} E^*\left[ \tilde{\mathbf{s}}^*\left( \mathbf{X}_{k i_k}^* , \mathbf{X}_{l 1}^* \right) \mid \mathbf{X}_{k i_k}^* \right]
+ \frac{1}{n_l} \sum_{i_l = 1}^{n_l} E\left[ \tilde{\mathbf{s}}^*\left( \mathbf{X}_{k 1}^* , \mathbf{X}_{l i_l}^* \right) \mid \mathbf{X}_{l i_l}^* \right] \right] \\
& \text{and }
\mathbf{V}_{k l}^*
= \frac{1}{n_k n_l} \sum_{i_k = 1}^{n_k} \sum_{i_l = 1}^{n_l} \left[ \tilde{\mathbf{s}}^*\left( \mathbf{X}_{k i_k}^* , \mathbf{X}_{l i_l}^* \right) 
- E^*\left[ \tilde{\mathbf{s}}^*\left( \mathbf{X}_{k i_k}^* , \mathbf{X}_{l 1}^* \right) \mid \mathbf{X}_{k i_k}^* \right]
- E^*\left[ \tilde{\mathbf{s}}^*\left( \mathbf{X}_{k 1}^* , \mathbf{X}_{l i_l}^* \right) \mid \mathbf{X}_{l i_l}^* \right] \right] .
\end{align*}
We have
$ \bar{\mathbf{R}}_k^* - E^*\left[ \bar{\mathbf{R}}_k^* \right]
= \mathbf{S}_k^* + n^{-1} \sum_{l=1}^{K} n_l \mathbf{V}_{k l}^* $.
Let
$ \mathbf{W}_n^* = \left( \sqrt{n_1} \mathbf{S}_1^*, \ldots, \sqrt{n_K} \mathbf{S}_K^* \right) $.
Using virtually identical arguments used while proving \eqref{thm1:eq4}, one can show that
\begin{align}
\left\| \left[ \mathbf{U}_n^* - E^*\left[ \mathbf{U}_n^* \right] \right] - \mathbf{W}_n^* \right\|
& = \sqrt{ \sum_{k = 1}^{K} \left\| \sqrt{n_k} \frac{1}{n} \sum_{l=1}^{K} n_l \mathbf{V}_{k l}^* \right\|^2 }
\stackrel{P^*}{\longrightarrow} 0
\quad\text{as } n \to \infty .
\label{thm1:booteq4}
\end{align}
Note that for any $ k, l, i_k, i_l $,
$ E^*\left[ \mathbf{s}\left( \mathbf{X}_{k i_k}^* - \mathbf{X}_{l i_l}^* \right) \right] 
= n^{-2} \sum_{q = 1}^{K} \sum_{r = 1}^{K} \sum_{i_q = 1}^{n_q} \sum_{i_r = 1}^{n_r} \mathbf{s}\left( \mathbf{X}_{q i_q} - \mathbf{X}_{r i_r} \right) 
\allowbreak= \mathbf{0} $,
which implies $ E^*\left[ \bar{\mathbf{R}}_k^* \right] = \mathbf{0} $ for all $ k $, and hence $ E^*\left[ \mathbf{U}_n^* \right] = \mathbf{0} $. Hence, from \eqref{thm1:booteq4}, we get
\begin{align}
\left\| \mathbf{U}_n^* - \mathbf{W}_n^* \right\|
& = \sqrt{ \sum_{k = 1}^{K} \left\| \sqrt{n_k} \frac{1}{n} \sum_{l=1}^{K} n_l \mathbf{V}_{k l}^* \right\|^2 }
\stackrel{P^*}{\longrightarrow} 0
\quad\text{as } n \to \infty .
\label{thm1:booteq4a}
\end{align}
Next,
%% Long form %%%%%%%%%%%%%%%%%%%%%%%%%%%%%%%%%%%%%
define
$ F_k^*\left( \mathbf{X}_{l i_l}^* \right) 
= E^*\left[ \tilde{\mathbf{s}}^*\left( \mathbf{X}_{k 1}^*, \mathbf{X}_{l i_l}^* \right) \mid \mathbf{X}_{l i_l}^* \right] $ for all $ k, l $ and $ i_l $,
\begin{align*}
\mathbf{E}_n
& = \frac{1}{n^3} \sum_{m = 1}^{K} \sum_{k = 1}^{K} \sum_{l = 1}^{K} \sum_{i_m = 1}^{n_m} \sum_{i_k = 1}^{n_k} \sum_{i_l = 1}^{n_l} \mathbf{s}\left( \mathbf{X}_{k i_k} - \mathbf{X}_{m i_m} \right) \otimes \mathbf{s}\left( \mathbf{X}_{l i_l} - \mathbf{X}_{m i_m} \right) \nonumber\\
& \quad -
\left\{ \frac{1}{n^2} \sum_{k = 1}^{K} \sum_{l = 1}^{K} \sum_{i_k = 1}^{n_k} \sum_{i_l = 1}^{n_l} \mathbf{s}\left( \mathbf{X}_{k i_k} - \mathbf{X}_{l i_l} \right) \right\} \otimes \left\{ \frac{1}{n^2} \sum_{k = 1}^{K} \sum_{l = 1}^{K} \sum_{i_k = 1}^{n_k} \sum_{i_l = 1}^{n_l} \mathbf{s}\left( \mathbf{X}_{k i_k} - \mathbf{X}_{l i_l} \right) \right\} \\
\text{and }
\mathbf{E}
& = \sum_{m = 1}^{K} \sum_{k = 1}^{K} \sum_{l = 1}^{K} \lambda_k \lambda_l \lambda_m E\left[ \mathbf{s}\left( \mathbf{X}_{k 1} - \mathbf{X}_{m 1} \right) \otimes \mathbf{s}\left( \mathbf{X}_{l 1} - \mathbf{X}_{m 1} \right) \right] \nonumber\\
& \quad -
\left\{ \sum_{k = 1}^{K} \sum_{l = 1}^{K} \lambda_k \lambda_l E\left[ \mathbf{s}\left( \mathbf{X}_{k 1} - \mathbf{X}_{l 1} \right) \right] \right\} \otimes \left\{ \sum_{k = 1}^{K} \sum_{l = 1}^{K} \lambda_k \lambda_l E\left[ \mathbf{s}\left( \mathbf{X}_{k 1} - \mathbf{X}_{l 1} \right) \right] \right\} .
\end{align*}
Note that for all $ k $ and $ l $, $ E^*\left[ F_k^*\left( \mathbf{X}_{l i_l}^* \right) \right] = \mathbf{0} $, and one can verify that for all $ k, l_1, l_2 $,
\begin{align}
& E^*\left[ F_{l_1}^*\left( \mathbf{X}_{k i_k}^* \right) \otimes F_{l_2}^*\left( \mathbf{X}_{k i_k}^* \right) \right]
= \mathbf{E}_n .
\label{thm1:booteq5}
\end{align}
Using the fact that $ \tilde{\mathbf{s}}^*\left( \mathbf{X}_{k i_k}^*, \mathbf{X}_{l i_l}^* \right) = - \tilde{\mathbf{s}}^*\left( \mathbf{X}_{l i_l}^*, \mathbf{X}_{k i_k}^* \right) $, we derive
\begin{align}
\mathbf{S}_k^*
& = - \frac{1}{n} \sum_{\substack{l=1 \\ l \neq k}}^{K} \frac{n_l}{n_k} \sum_{i_k = 1}^{n_k} F_l^*\left( \mathbf{X}_{k i_k}^* \right)
+ \frac{1}{n} \sum_{\substack{l=1 \\ l \neq k}}^{K} \sum_{i_l = 1}^{n_l} F_k^*\left( \mathbf{X}_{l i_l}^* \right) .
\label{thm1:booteq6}
\end{align}
Since $ \text{Cov}^*\left( F_{l_1}^*\left( \mathbf{X}_{k i_k}^* \right) , F_{l_2}^*\left( \mathbf{X}_{l i_l}^* \right) \right) = \mathbf{0} $ for $ ( k, i_k ) \neq ( l, i_l ) $ and for all $ l_1, l_2 $, from \eqref{thm1:booteq5} and \eqref{thm1:booteq6}, we get
\begin{align}
\text{Cov}^*\left( \sqrt{n_{k_1}} \mathbf{S}_{k_1}^* , \sqrt{n_{k_2}} \mathbf{S}_{k_2}^* \right)
= \left( \mathbb{I}\left( k_1 = k_2 \right) - \frac{\sqrt{n_{k_1} n_{k_2}}}{n} \right) \mathbf{E}_n .
\label{thm1:booteq7}
\end{align}
Using arguments similar to those in the proof of \autoref{thm:cov}, we get that $ \mathbf{E}_n \stackrel{a.s.}{\longrightarrow} \mathbf{E} $ as $ n \to \infty $ in the operator norm, and hence, from \eqref{thm1:booteq7}, we get
\begin{align}
\text{Cov}^*\left( \sqrt{n_{k_1}} \mathbf{S}_{k_1}^* , \sqrt{n_{k_2}} \mathbf{S}_{k_2}^* \right)
\stackrel{a.s.}{\longrightarrow}
\left( \mathbb{I}\left( k_1 = k_2 \right) - \sqrt{\lambda_{k_1} \lambda_{k_2}} \right) \mathbf{E}
\quad\text{as } n \to \infty .
\label{thm1:booteq8}
\end{align}
Define $ \boldsymbol{\Sigma}^* = \left( \boldsymbol{\sigma}_{k_1 k_2}^* \right)_{K \times K} $, where
$ \boldsymbol{\sigma}_{k_1 k_2}^*
= \left( \mathbb{I}\left( k_1 = k_2 \right) - \sqrt{\lambda_{k_1} \lambda_{k_2}} \right) \mathbf{E} $
for $ k_1, k_2 = 1, \ldots, K $.
Since $ E^*\left[ \mathbf{S}_k^* \right] = \mathbf{0} $ for all $ k $, we have $ E\left[ \mathbf{W}_n^* \right] = \mathbf{0} $, which, along with \eqref{thm1:booteq8} and an application of Theorem 1.1 in \cite{kundu2000central}, implies that
$ \mathbf{W}_n^* \stackrel{w^*}{\longrightarrow} G\left( \mathbf{0}, \boldsymbol{\Sigma}^* \right) $ as $ n \to \infty $. Therefore, from \eqref{thm1:booteq4a}, we get
$ \mathbf{U}_n^*
\stackrel{w^*}{\longrightarrow} G\left( \mathbf{0}, \boldsymbol{\Sigma}^* \right) $
as $ n \to \infty $, and since $ \boldsymbol{\Sigma}^* $ is independent of $ \mathcal{S} $, we have the unconditional convergence
\begin{align}
\mathbf{U}_n^*
\stackrel{w}{\longrightarrow} G\left( \mathbf{0}, \boldsymbol{\Sigma}^* \right)
\quad\text{as } n \to \infty .
\label{thm1:booteq9}
\end{align}

It can be verified that when the null hypothesis $ \text{H}_0 $ in \eqref{h_0} is true, we have $ \boldsymbol{\Sigma}^* = \boldsymbol{\Sigma} $, where $ \boldsymbol{\Sigma} $ is as defined in \autoref{thm:1}.
From \autoref{coro:1} and using arguments similar to those used in its proof, we get that under $ \text{H}_0 $ in \eqref{h_0}, the asymptotic distribution of $ SS_n^* $ is identical to that of $ SS_n $.
Hence under $ \text{H}_0 $ in \eqref{h_0}, for every $ 0 < \alpha < 1 $, the size of a level $ \alpha $ test based on the bootstrap procedure converges to $ \alpha $ as $ n \to \infty $ and $ M_b \to \infty $.
%%%%%%%%%%%%%%%%%%%%%%%%%%%%%%%%%%%%%%%%%%%%%%%

Next, we prove the theorem for the permutation method.
Let $ \boldsymbol{\sigma}_n $ denote a random permutation of $ ( 1, \ldots, n ) $, and
$ \{ \mathbf{X}_{k i}^\# \midil k = 1, \ldots, K ; i = 1, \ldots, n_k \} $ denote the permuted sample obtained by applying $ \boldsymbol{\sigma}_n $ to the original sample.
Define
$ \mathbf{R}^\#( \mathbf{x} ) = n^{-1} \sum_{k=1}^{K} \sum_{i=1}^{n_k} \mathbf{s}( \mathbf{x} - \mathbf{X}_{k i}^\# ) $ and
$ \bar{\mathbf{R}}_k^\# = n_k^{-1} \sum_{i_k = 1}^{n_k} \mathbf{R}^\#\left( \mathbf{X}_{k i_k}^\# \right) $.
So, $ SS_n^\# = \sum_{k=1}^{K} n_k \left\| \bar{\mathbf{R}}_k^\# \right\|^2 $.
Given the original sample, all possible realizations of the permuted sample $ \{ \mathbf{X}_{k i}^\# \midil k = 1, \ldots, K ; i = 1, \ldots, n_k \} $ are equally likely. Also, under the null hypothesis $ \text{H}_0 $ in \eqref{h_0}, the underlying distributions of all the classes are identical. So, the joint distribution of $ \{ \mathbf{X}_{k i}^\# \midil k = 1, \ldots, K ; i = 1, \ldots, n_k \} $ is the same as that of the original sample.
Hence, under $ \text{H}_0 $ in \eqref{h_0}, the distribution of $ SS_n^\# $ is the same as the distribution of the original test statistic $ SS_n $ for all $ n $.
Therefore, under $ \text{H}_0 $ in \eqref{h_0}, the asymptotic distribution of $ SS_n^\# $ is the same as that of $ SS_n $, which is given in \autoref{coro:1}.
So, as $ n \to \infty $ and $ M \to \infty $, the empirical distribution of $ SS_n^\# $ converges to the asymptotic null distribution of $ SS_n $.
Hence under $ \text{H}_0 $ in \eqref{h_0}, for every $ 0 < \alpha < 1 $, the size of a level $ \alpha $ test based on the permutation procedure converges to $ \alpha $ as $ n \to \infty $ and $ M \to \infty $.
\end{proof}

\begin{proof}[Proof of \autoref{thm:2}]
Since for all $ k $, $ E[ \bar{\mathbf{R}}_k ] \to \sum_{l=1}^{K} \lambda_l E[ \mathbf{s}( \mathbf{X}_{k 1} - \mathbf{X}_{l 1} ) ] $ as $ n \to \infty $, if $ \sum_{l=1}^{K} \lambda_l E[ \mathbf{s}( \mathbf{X}_{k 1} - \mathbf{X}_{l 1} ) ] \neq \mathbf{0} $ for any $ k $, we have $ \left\| E\left[ \sqrt{n_k} \bar{\mathbf{R}}_k \right] \right\| \to \infty $ as $ n \to \infty $ for that $ k $, and this implies $ \left\| E\left[ \mathbf{U}_n \right] \right\| \to \infty $ as $ n \to \infty $.
Consequently, from \autoref{thm:1}, we have $ \left\| \mathbf{U}_n \right\| \ge \left\| E\left[ \mathbf{U}_n \right] \right\| - \left\| \mathbf{U}_n - E\left[ \mathbf{U}_n \right] \right\| \stackrel{P}{\longrightarrow} \infty $ as $ n \to \infty $.
Therefore,
\begin{align}
SS_n = \left\| \mathbf{U}_n \right\|^2 \stackrel{P}{\longrightarrow} \infty
\quad\text{as } n \to \infty .
\label{thm2:eq1}
\end{align}

From \eqref{thm2:eq1} and \autoref{coro:1}, we get that for every $ 0 < \alpha < 1 $, the power of a level $ \alpha $ test based on the asymptotic procedure converges to $ 1 $ as $ n \to \infty $.

Next, from \eqref{thm1:booteq9}, we get that $ SS_n^* $ is stochastically bounded whether the null or the alternative is true. Hence, from \eqref{thm2:eq1}, we get that for every $ 0 < \alpha < 1 $, the power of a level $ \alpha $ test based on the bootstrap procedure converges to $ 1 $ as $ n \to \infty $ and $ M_b \to \infty $.

Next we prove the theorem for the permutation procedure. Recall $ \bar{\mathbf{R}}_k^\# $ defined in the proof of \autoref{thm:bootstrapperm}.
We shall show that $ n_k \| \bar{\mathbf{R}}_k^\# \|^2 $ is stochastically bounded for all $ k $, and
this would imply that $ SS_n^\# $ is stochastically bounded under the null as well as the alternative hypotheses. Then the proof would follow from \eqref{thm2:eq1}.
It can be verified that
\begin{align*}
& n_k \left\| \bar{\mathbf{R}}_k^\# \right\|^2 \\
& = n_k \mathop{\sum_{l_1 = 1}^{K}}_{l_1 \neq k} \mathop{\sum_{l_2 = 1}^{K}}_{l_2 \neq k} \frac{n_{l_1}}{n} \frac{n_{l_2}}{n} \frac{1}{n_k^2 n_{l_1} n_{l_2}} \mathop{\sum_{i_1 = 1}^{n_k} \sum_{i_2 = 1}^{n_k}}_{i_1 \neq i_2} \mathop{\sum_{j_1 = 1}^{n_{l_1}} \sum_{j_2 = 1}^{n_{l_2}}}_{( l_1, j_1 ) \neq ( l_2, j_2 )} \left\langle \mathbf{s}\left( \mathbf{X}_{k i_1}^\# - \mathbf{X}_{l_1 j_1}^\# \right), \; \mathbf{s}\left( \mathbf{X}_{k i_2}^\# - \mathbf{X}_{l_2 j_2}^\# \right) \right\rangle \\
& \quad + n_k \mathop{\sum_{l = 1}^{K}}_{l \neq k} \frac{n_{l}^2}{n^2} \frac{1}{n_k^2 n_{l}^2} \mathop{\sum_{i_1 = 1}^{n_k} \sum_{i_2 = 1}^{n_k}}_{i_1 \neq i_2} \sum_{j = 1}^{n_l} \left\langle \mathbf{s}\left( \mathbf{X}_{k i_1}^\# - \mathbf{X}_{l j}^\# \right), \; \mathbf{s}\left( \mathbf{X}_{k i_2}^\# - \mathbf{X}_{l j}^\# \right) \right\rangle \\
& \quad + n_k \mathop{\sum_{l_1 = 1}^{K}}_{l_1 \neq k} \mathop{\sum_{l_2 = 1}^{K}}_{l_2 \neq k} \frac{n_{l_1}}{n} \frac{n_{l_2}}{n} \frac{1}{n_k^2 n_{l_1} n_{l_2}} \sum_{i = 1}^{n_k} \mathop{\sum_{j_1 = 1}^{n_{l_1}} \sum_{j_2 = 1}^{n_{l_2}}}_{( l_1, j_1 ) \neq ( l_2, j_2 )} \left\langle \mathbf{s}\left( \mathbf{X}_{k i}^\# - \mathbf{X}_{l_1 j_1}^\# \right), \; \mathbf{s}\left( \mathbf{X}_{k i}^\# - \mathbf{X}_{l_2 j_2}^\# \right) \right\rangle \\
& \quad + n_k \mathop{\sum_{l = 1}^{K}}_{l \neq k} \frac{n_{l}^2}{n^2} \frac{1}{n_k^2 n_{l}^2} \sum_{i = 1}^{n_k} \sum_{j = 1}^{n_l} \left\langle \mathbf{s}\left( \mathbf{X}_{k i}^\# - \mathbf{X}_{l j}^\# \right), \; \mathbf{s}\left( \mathbf{X}_{k i}^\# - \mathbf{X}_{l j}^\# \right) \right\rangle .
\end{align*}
Let $ E^\#\left[ \cdot \right] $ denote the conditional expectation given the original sample.
We get that for $ ( k, i_1 ) $, $ ( k, i_2 ) $, $ ( l_1, j_1 ) $ and $ ( l_2, j_2 ) $ all distinct,
\begin{align*}
E^\#\left[ \left\langle \mathbf{s}\left( \mathbf{X}_{k i_1}^\# - \mathbf{X}_{l_1 j_1}^\# \right), \; \mathbf{s}\left( \mathbf{X}_{k i_2}^\# - \mathbf{X}_{l_2 j_2}^\# \right) \right\rangle \right] = \mathbf{0} ,
\end{align*}
for $ ( k, i_1 ) $, $ ( k, i_2 ) $ and $ ( l, j ) $ all distinct,
\begin{align*}
E^\#\left[ \left\langle \mathbf{s}\left( \mathbf{X}_{k i_1}^\# - \mathbf{X}_{l j}^\# \right), \; \mathbf{s}\left( \mathbf{X}_{k i_2}^\# - \mathbf{X}_{l j}^\# \right) \right\rangle \right]
& = \frac{(n - 3) !}{n !} \sum_{(i, j, k) \in \mathcal{P}_3^n} \left\langle \mathbf{s}\left( \mathbf{Y}_i - \mathbf{Y}_k \right), \; \mathbf{s}\left( \mathbf{Y}_j - \mathbf{Y}_k \right) \right\rangle ,
\end{align*}
for $ ( k, i ) $, $ ( l_1, j_1 ) $ and $ ( l_2, j_2 ) $ all distinct,
\begin{align*}
E^\#\left[ \left\langle \mathbf{s}\left( \mathbf{X}_{k i}^\# - \mathbf{X}_{l_1 j_1}^\# \right), \; \mathbf{s}\left( \mathbf{X}_{k i}^\# - \mathbf{X}_{l_2 j_2}^\# \right) \right\rangle \right]
& = \frac{(n - 3) !}{n !} \sum_{(i, j, k) \in \mathcal{P}_3^n} \left\langle \mathbf{s}\left( \mathbf{Y}_i - \mathbf{Y}_k \right), \; \mathbf{s}\left( \mathbf{Y}_j - \mathbf{Y}_k \right) \right\rangle ,
\end{align*}
and for $ ( k, i ) \neq ( l, j ) $,
\begin{align*}
E^\#\left[ \left\langle \mathbf{s}\left( \mathbf{X}_{k i}^\# - \mathbf{X}_{l j}^\# \right), \; \mathbf{s}\left( \mathbf{X}_{k i}^\# - \mathbf{X}_{l j}^\# \right) \right\rangle \right]
& = \frac{(n - 2) !}{n !} \sum_{(i, j) \in \mathcal{P}_2^n} \left\langle \mathbf{s}\left( \mathbf{Y}_i - \mathbf{Y}_j \right), \; \mathbf{s}\left( \mathbf{Y}_i - \mathbf{Y}_j \right) \right\rangle ,
\end{align*}
where $ \mathcal{P}_k^n $ is the collection of all permutations of size $ k $ drawn from $ ( 1, \ldots, n ) $, and $ \mathbf{Y}_i $, $ \mathbf{Y}_j $ and $ \mathbf{Y}_k $ are the $ i^\text{th} $, $ j^\text{th} $ and the $ k^\text{th} $ elements of the sample, respectively.
Define
\begin{align*}
& R_{n, 1}
= \frac{(n - 3) !}{n !} \sum_{(i, j, k) \in \mathcal{P}_3^n} \left\langle \mathbf{s}\left( \mathbf{Y}_i - \mathbf{Y}_k \right), \; \mathbf{s}\left( \mathbf{Y}_j - \mathbf{Y}_k \right) \right\rangle ,
\\
& \text{and }
R_{n, 2}
= \frac{(n - 2) !}{n !} \sum_{(i, j) \in \mathcal{P}_2^n} \left\langle \mathbf{s}\left( \mathbf{Y}_i - \mathbf{Y}_j \right), \; \mathbf{s}\left( \mathbf{Y}_i - \mathbf{Y}_j \right) \right\rangle .
\end{align*}
From the above discussion, we get
\begin{align}
E\left[ n_k \left\| \bar{\mathbf{R}}_k^\# \right\|^2 \right]
& = E\left[ R_{n, 1} \right] \left[ \sum_{l \neq k} \frac{(n_k + n_l - 2) n_l}{n^2} + \sum_{l_1 \neq l_2 \neq k} \frac{n_{l_1} n_{l_2}}{n^2} \right]
+ E\left[ R_{n, 2} \right] \frac{1}{n} \sum_{l \neq k} \frac{n_l}{n} ,
\label{thm2:eq2}
\\
\text{where } & R_{n, 1}
= \frac{(n - 3) !}{n !} \sum_{(i, j, k) \in \mathcal{P}_3^n} \left\langle \mathbf{s}\left( \mathbf{Y}_i - \mathbf{Y}_k \right), \; \mathbf{s}\left( \mathbf{Y}_j - \mathbf{Y}_k \right) \right\rangle 
\nonumber\\
\text{and } &
R_{n, 2}
= \frac{(n - 2) !}{n !} \sum_{(i, j) \in \mathcal{P}_2^n} \left\langle \mathbf{s}\left( \mathbf{Y}_i - \mathbf{Y}_j \right), \; \mathbf{s}\left( \mathbf{Y}_i - \mathbf{Y}_j \right) \right\rangle .\nonumber
\end{align}
Now, from an application of the Cauchy-Schwarz inequality, we get that $ | R_{n, 1} | \le 1 $ and $ | R_{n, 2} | \le 1 $. Therefore, from \eqref{thm2:eq2}, we have
$ E[ n_k \| \bar{\mathbf{R}}_k^\# \|^2 ] \le 2 K $,
and consequently $ n_k \| \bar{\mathbf{R}}_k^\# \|^2 $ is stochastically bounded for all $ k $. Hence, $ SS_n^\# $ is stochastically bounded.

Hence, from \eqref{thm2:eq1}, we get that for every $ 0 < \alpha < 1 $, the power of a level $ \alpha $ test based on the permutation procedure converges to $ 1 $ as $ n \to \infty $ and $ M \to \infty $.
\end{proof}

\begin{proof}[Proof of \autoref{thm:2linear}]
Under $ \text{H}_0' $ in \eqref{h_0l}, $ \text{H}_0 $ in \eqref{h_0} holds, and from \autoref{thm:bootstrapperm}, we get that for every $ 0 < \alpha < 1 $, the sizes of a level $ \alpha $ test based on the asymptotic and the permutation procedures converge to $ \alpha $ as $ n \to \infty $ and $ M \to \infty $.

Next, let $ \mathbf{X} $ and $ \mathbf{X}' $ be independent random elements having distribution $ P_0 $. Consider a random element $ \mathbf{Z} $ such that $ \mathbf{Z} = \left( \mathbf{X} + \boldsymbol{\mu}_k \right) - \mathbf{X}' $ with probability $ \lambda_k $ for $ k = 1, \ldots, K $.
The spatial distribution of $ \mathbf{Z} $ at a point $ \mathbf{z} $ is defined as $ \mathbf{S}\left( \mathbf{z} \right) = E\left[ \mathbf{s}\left( \mathbf{z} - \mathbf{Z} \right) \right] $ \cite[p.~1205]{chakraborty2014spatial}.
Since the support of $ P_0 $ is not contained in a straight line in $ \mathcal{H} $, it follows from the proof of Theorem 3.1 in \cite{chakraborty2014spatial} that $ \mathbf{S}\left( \mathbf{z}_1 \right) \neq \mathbf{S}\left( \mathbf{z}_2 \right) $ for $ \mathbf{z}_1 \neq \mathbf{z}_2 $.
So, $ \mathbf{S}\left( \boldsymbol{\mu}_k \right) = \mathbf{0} $ for all $ k $ implies $ \boldsymbol{\mu}_k = \boldsymbol{\mu}_l $ for all $ k \neq l $.
Note that $ \mathbf{S}\left( \boldsymbol{\mu}_k \right) = \sum_{l=1}^{K} \lambda_l E\left[ \mathbf{s}\left( \mathbf{X}_{k 1} - \mathbf{X}_{l 1} \right) \right] $, and $ E\left[ \bar{\mathbf{R}}_k \right] \to \mathbf{S}\left( \boldsymbol{\mu}_k \right) $ as $ n \to \infty $ since $ n^{-1} n_k \to \lambda_k $ as $ n \to \infty $ for all $ k $.
When $ \text{H}_0' $ is false, $ \boldsymbol{\mu}_k \neq \boldsymbol{\mu}_l $ for some $ k \neq l $, which implies $ \mathbf{S}\left( \boldsymbol{\mu}_k \right) \neq \mathbf{0} $ for some $ k $, and the proof follows from \autoref{thm:2}.
\end{proof}

We need the following lemma for the proof of \autoref{thm:3}. The arguments used in its proof are similar to those in the proof of Proposition 2.1 in \cite{cardot2013efficient}.
\begin{lemma} \label{lemma:lemma1}
Define $ \mathbf{f} : \mathcal{H} \to \mathcal{H} $ by $ \mathbf{f}( \mathbf{a} ) = E\left[ \mathbf{s}\left( \mathbf{X} - \mathbf{X}' + \mathbf{a} \right) \right] $.
Assume that the probability measure $ P_0 $ is non-atomic and not contained in any straight line in $ \mathcal{H} $,
and $ E\left[ \left\| \mathbf{X} - \mathbf{X}' \right\|^{-1} \right] < \infty $, where $ \mathbf{X} $ and $ \mathbf{X}' $ are independent random elements having identical distribution $ P_0 $.
Then, the Frech\'{e}t derivative of $ \mathbf{f}( \cdot ) $ at $ \mathbf{0} $ is $ E\left[ \mathbf{s}^{(1)}\left( \mathbf{X} - \mathbf{X}' \right) \right] $.
Further, for any $ \mathbf{h} \neq \mathbf{0} $, $ E\left[ \mathbf{s}^{(1)}\left( \mathbf{X} - \mathbf{X}' \right) \right]\left( \mathbf{h} \right) \neq \mathbf{0} $.
\end{lemma}
\begin{proof}
The Frech\'{e}t derivative of $ \mathbf{s}( \mathbf{x} ) = \| \mathbf{x} \|^{-1} \mathbf{x} $ at $ \mathbf{x} $, denoted as $ \mathbf{s}^{(1)}( \mathbf{x} ) $, exists for all $ \mathbf{x} \neq \mathbf{0} $, and it is a bounded linear operator on $ \mathcal{H} $ defined by
%% Long form %%
\begin{align*}
\mathbf{s}^{(1)}\left( \mathbf{x} \right)\left( \mathbf{h} \right)
= \frac{1}{\left\| \mathbf{x} \right\|} \left( \mathbf{h} - \left\langle \mathbf{h}, \frac{\mathbf{x}}{\left\| \mathbf{x} \right\|} \right\rangle \frac{\mathbf{x}}{\left\| \mathbf{x} \right\|} \right) ,
\text{ where } \mathbf{h} \in \mathcal{H} .
\end{align*}
%% Short form %%
%$ \mathbf{s}^{(1)}\left( \mathbf{x} \right)\left( \mathbf{h} \right)
%= \left\| \mathbf{x} \right\|^{-1} \left( \mathbf{h} - \left\langle \mathbf{h}, \left\| \mathbf{x} \right\|^{-1} \mathbf{x} \right\rangle \left\| \mathbf{x} \right\|^{-1} \mathbf{x} \right) $, where $ \mathbf{h} \in \mathcal{H} $.
Note that
$ \left\| \mathbf{s}^{(1)}\left( \mathbf{x} \right) \right\| \le \left\| \mathbf{x} \right\|^{-1} $.
Since by assumption $ E\left[ \left\| \mathbf{X} - \mathbf{X}' \right\|^{-1} \right] < \infty $,
which implies $ P\left[ \mathbf{X} - \mathbf{X}' = \mathbf{0} \right] \allowbreak = 0 $,
we get that
$ E\left[ \mathbf{s}^{(1)}\left( \mathbf{X} - \mathbf{X}' \right) \right] $ exists and is the Frech\'{e}t derivative of $ \mathbf{f}( \mathbf{a} ) = E\left[ \mathbf{s}\left( \mathbf{X} - \mathbf{X}' + \mathbf{a} \right) \right] $ at $ \mathbf{a} = \mathbf{0} $.

Next, we shall show that the operator $ E\left[ \mathbf{s}^{(1)}\left( \mathbf{X} - \mathbf{X}' \right) \right] $ is bounded below.
The arguments in this part of the proof is similar to those in the proof of Proposition 2.1 in \cite{cardot2013efficient}.
Let $ \mathbf{P}_{\mathbf{h}}\left( \cdot \right) $ be the orthogonal projection operator on the orthogonal complement of $ \mathbf{h} $.
Note that
%% Long form %%
%\begin{align*}
%\left\langle \mathbf{h}, E\left[ \mathbf{s}^{(1)}\left( \mathbf{X} - \mathbf{X}' \right) \right]\left( \mathbf{h} \right) \right\rangle
%& = E\left[ \frac{1}{\left\| \mathbf{X} - \mathbf{X}' \right\|} \left( 1 - \left\langle \mathbf{h},\, \frac{\mathbf{X} - \mathbf{X}'}{\left\| \mathbf{X} - \mathbf{X}' \right\|} \right\rangle^2 \right) \right]
%\\
%& = E\left[ \frac{1}{\left\| \mathbf{X} - \mathbf{X}' \right\|} \left\| \mathbf{P}_{\mathbf{h}}\left( \frac{\mathbf{X} - \mathbf{X}'}{\left\| \mathbf{X} - \mathbf{X}' \right\|} \right) \right\|^2 \right] .
%\end{align*}
%% Short form %%
\begin{align*}
\left\langle \mathbf{h}, E\left[ \mathbf{s}^{(1)}\left( \mathbf{X} - \mathbf{X}' \right) \right]\left( \mathbf{h} \right) \right\rangle
= E\left[ \frac{1}{\left\| \mathbf{X} - \mathbf{X}' \right\|} \left\| \mathbf{P}_{\mathbf{h}}\left( \frac{\mathbf{X} - \mathbf{X}'}{\left\| \mathbf{X} - \mathbf{X}' \right\|} \right) \right\|^2 \right] .
\end{align*}
%%%%%%%%%%%%%%%%
Consider the set of all subspaces $ \mathcal{K} \subseteq \mathcal{H} $ such that $ \text{Var}\left( \left\langle \mathbf{v},\, \left\| \mathbf{X} - \mathbf{X}' \right\|^{-1} \left( \mathbf{X} - \mathbf{X}' \right) \right\rangle \right) = 0 $ for all $ \mathbf{v} \in \mathcal{K} $.
Suppose this set is non-empty. Then, by Zorn's lemma, it has a maximal element, which we denote as $ \mathcal{K}_0 $.
The maximality of $ \mathcal{K}_0 $ implies that $ \text{Var}\left( \left\langle \mathbf{v},\, \left\| \mathbf{X} - \mathbf{X}' \right\|^{-1} \left( \mathbf{X} - \mathbf{X}' \right) \right\rangle \right) > 0 $ for all $ \mathbf{v} \in \mathcal{K}_0^\perp $, the orthogonal complement of $ \mathcal{K}_0 $.
Since $ \mathbf{X} $ and $ \mathbf{X}' $ are independent random elements having distribution $ P_0 $, which is not contained in a straight line, $ \mathcal{K}_0^\perp $ must have dimension at least 2.
Let $ \mathbf{v}_1 , \mathbf{v}_2 $ be two orthogonal unit vectors in $ \mathcal{K}_0^\perp $.
Define $ g\left( t \right) = \text{Var}\left( \left\langle \cos\left( t \right) \mathbf{v}_1 + \sin\left( t \right) \mathbf{v}_2 ,\, \left\| \mathbf{X} - \mathbf{X}' \right\|^{-1} \left( \mathbf{X} - \mathbf{X}' \right) \right\rangle \right) $ for $ 0 \le t \le 2 \pi $.
Note that $ g\left( t \right) > 0 $ for all $ t $, because $ \cos\left( t \right) \mathbf{v}_1 + \sin\left( t \right) \mathbf{v}_2 \in \mathcal{K}_0^\perp $.
Since $ g\left( \cdot \right) $ is continuous on a compact set, we get that there is $ l_0 > 0 $ such that for all unit vectors in $ \text{span}\left\{ \mathbf{v}_1, \mathbf{v}_2 \right\} $, $ \text{Var}\left( \left\langle \mathbf{v} ,\, \left\| \mathbf{X} - \mathbf{X}' \right\|^{-1} \left( \mathbf{X} - \mathbf{X}' \right) \right\rangle \right) \ge l_0 $.
The orthogonal complement of $ \mathbf{h} $ and $ \text{span}\left\{ \mathbf{v}_1, \mathbf{v}_2 \right\} $ must have nonempty intersection, which implies that there is a unit vector $ \mathbf{v}_0 \in \text{span}\left\{ \mathbf{v}_1, \mathbf{v}_2 \right\} $ such that $ \left\langle \mathbf{v}_0, \mathbf{h} \right\rangle = 0 $.
Hence, for all $ \mathbf{y} \in \mathcal{H} $, $ \left\| \mathbf{P}_{\mathbf{h}}\left( \mathbf{y} \right) \right\|^2 \ge \left\langle \mathbf{v}_0 , \mathbf{y} \right\rangle^2 $.
%% Long form %%
%, which implies
%$ \left\| \mathbf{P}_{\mathbf{h}}\left( \left\| \mathbf{X} - \mathbf{X}' \right\|^{-1} \left( \mathbf{X} - \mathbf{X}' \right) \right) \right\|^2 \ge \left\langle \mathbf{v}_0 ,\, \left\| \mathbf{X} - \mathbf{X}' \right\|^{-1} \left( \mathbf{X} - \mathbf{X}' \right) \right\rangle^2 $.
%% No short form here, check full stop/comma %%
We take $ M $ large enough such that $ \text{Var}\left[ \left\langle \mathbf{v}_0 , \left\| \mathbf{X} - \mathbf{X}' \right\|^{-1} \left( \mathbf{X} - \mathbf{X}' \right) \mathbb{I}\left( \left\| \mathbf{X} - \mathbf{X}' \right\| \le M \right) \right\rangle \right] > \left( l_0 / 2 \right) $.
Therefore,
\begin{align}
& \left\langle \mathbf{h}, E\left[ \mathbf{s}^{(1)}\left( \mathbf{X} - \mathbf{X}' \right) \right]\left( \mathbf{h} \right) \right\rangle
\nonumber\\
& \ge E\left[ \frac{1}{\left\| \mathbf{X} - \mathbf{X}' \right\|} \left\| \mathbf{P}_{\mathbf{h}}\left( \frac{\mathbf{X} - \mathbf{X}'}{\left\| \mathbf{X} - \mathbf{X}' \right\|} \right) \right\|^2 \mathbb{I}\left( \left\| \mathbf{X} - \mathbf{X}' \right\| \le M \right) \right]
\nonumber\\
& \ge \frac{1}{M} E\left[ \left\langle \mathbf{v}_0 , \left\| \mathbf{X} - \mathbf{X}' \right\|^{-1} \left( \mathbf{X} - \mathbf{X}' \right) \mathbb{I}\left( \left\| \mathbf{X} - \mathbf{X}' \right\| \le M \right) \right\rangle^2 \right]
\nonumber\\
& \ge \frac{1}{M} \text{Var}\left[ \left\langle \mathbf{v}_0 , \left\| \mathbf{X} - \mathbf{X}' \right\|^{-1} \left( \mathbf{X} - \mathbf{X}' \right) \mathbb{I}\left( \left\| \mathbf{X} - \mathbf{X}' \right\| \le M \right) \right\rangle \right]
\ge \frac{l_0}{2 M} > 0 .
\label{lemma1:eq2}
\end{align}
From \eqref{lemma1:eq2}, we get that for every $ \mathbf{h} \in \mathcal{H} $,
\begin{align}
& 0 < \frac{l_0}{2 M} \left\| \mathbf{h} \right\|^2
\le \left\langle \mathbf{h}, E\left[ \mathbf{s}^{(1)}\left( \mathbf{X} - \mathbf{X}' \right) \right]\left( \mathbf{h} \right) \right\rangle
\le \left\| E\left[ \mathbf{s}^{(1)}\left( \mathbf{X} - \mathbf{X}' \right) \right]\left( \mathbf{h} \right) \right\| \left\| \mathbf{h} \right\| ,
\nonumber\\
& \text{which implies} \quad
\left\| E\left[ \mathbf{s}^{(1)}\left( \mathbf{X} - \mathbf{X}' \right) \right]\left( \mathbf{h} \right) \right\| 
\ge \frac{l_0}{2 M} \left\| \mathbf{h} \right\| > 0 .
\label{lemma1:eq3}
\end{align}
Hence, from \eqref{lemma1:eq3}, we get that for any $ \mathbf{h} \neq \mathbf{0} $, $ E\left[ \mathbf{s}^{(1)}\left( \mathbf{X} - \mathbf{X}' \right) \right]\left( \mathbf{h} \right) \neq \mathbf{0} $.
\end{proof}

\begin{proof}[Proof of \autoref{thm:3}]
Define $ \mathbf{f} : \mathcal{H} \to \mathcal{H} $ by $ \mathbf{f}( \mathbf{a} ) = E\left[ \mathbf{s}\left( \mathbf{X} - \mathbf{X}' + \mathbf{a} \right) \right] $.
From \autoref{lemma:lemma1},
we get that the Frech\'{e}t derivative of $ \mathbf{f}( \cdot ) $ at $ \mathbf{0} $ is $ E\left[ \mathbf{s}^{(1)}\left( \mathbf{X} - \mathbf{X}' \right) \right] $.
Define $ \mathbf{X}_k = \mathbf{X}_{k 1} - \boldsymbol{\mu}_k $ for $ k = 1, \ldots, K $. Then, $ \mathbf{X}_1, \ldots, \mathbf{X}_K $ are independent and identically distributed random elements with common distribution $ P_0 $.
Since $ \mathbf{s}\left( \mathbf{x} - \mathbf{y} \right)
= - \mathbf{s}\left( \mathbf{y} - \mathbf{x} \right) $, we get
$ E\left[ \mathbf{s}\left( \mathbf{X}_k - \mathbf{X}_l \right) \right] = - E\left[ \mathbf{s}\left( \mathbf{X}_k - \mathbf{X}_l \right) \right] = \mathbf{0} $ for all $ k, l $.
Using this fact and a Taylor expansion of $ \mathbf{f}( \cdot ) $ at $ \mathbf{0} $, we have
\begin{align}
E\left[ \mathbf{s}\left( \mathbf{X}_{k 1} - \mathbf{X}_{l 1} \right) \right]
& = E\left[ \mathbf{s}^{(1)}\left( \mathbf{X} - \mathbf{X}' \right) \right]\left( \frac{1}{\sqrt{n}} \left( \boldsymbol{\delta}_k - \boldsymbol{\delta}_l \right) \right) 
+ o\left( \frac{1}{\sqrt{n}} \left\| \boldsymbol{\delta}_k - \boldsymbol{\delta}_l \right\| \right) .
\label{thm3:eq1}
\end{align}
Since $ n^{-1} n_k \to \lambda_k \in ( 0, 1 ) $ for all $ k $ as $ n \to \infty $, from \eqref{thm3:eq1} we get
\begin{align}
E\left[ \sqrt{n_k} \bar{\mathbf{R}}_k \right]
= \sqrt{n_k} \sum_{l=1}^{K} \frac{n_l}{n} E\left[ \mathbf{s}\left( \mathbf{X}_{k 1} - \mathbf{X}_{l 1} \right) \right]
\to E\left[ \mathbf{s}^{(1)}\left( \mathbf{X} - \mathbf{X}' \right) \right] \left( \sqrt{\lambda_k} \left( \boldsymbol{\delta}_k - \bar{\boldsymbol{\delta}} \right) \right)
\label{thm3:eq2}
\end{align}
as $ n \to \infty $.
From \eqref{thm3:eq2}, we have $ E\left[ \mathbf{U}_n \right] \to \mathbf{U}_0 = \left( \mathbf{u}_1, \ldots, \mathbf{u}_K \right) $ as $ n \to \infty $, where
\begin{align*}
\mathbf{u}_k = E\left[ \mathbf{s}^{(1)}\left( \mathbf{X} - \mathbf{X}' \right) \right]\left( \sqrt{\lambda_k} \left( \boldsymbol{\delta}_k - \bar{\boldsymbol{\delta}} \right) \right)
\end{align*}
for all $ k = 1, \ldots, K $.
From Slutsky's theorem and \autoref{thm:1}, we get $ \mathbf{U}_n \stackrel{w}{\longrightarrow} G\left( \mathbf{U}_0, \boldsymbol{\Sigma} \right) $ as $ n \to \infty $. Consequently, from an application of the Karhunen-Loeve expansion and the mapping theorem \cite[Theorem 2.7, p.~21]{billingsley2013convergence}, we get
$ \| \mathbf{U}_n \|^2 \stackrel{w}{\longrightarrow} \left\| \mathbf{U}_0 - \sum_{i=1}^{\infty} \langle \mathbf{U}_0, \boldsymbol{\beta}_i \rangle \right\|^2 + \sum_{i=1}^{\infty} \left( \langle \mathbf{U}_0, \boldsymbol{\beta}_i \rangle + \sqrt{\alpha_i} \mathbf{Z}_i \right)^2 $ as $ n \to \infty $.

Let $ \boldsymbol{\delta} = \left( \boldsymbol{\delta}_1, \ldots, \boldsymbol{\delta}_K \right)^t $, $ \mathbf{D} = \text{Diag}\left( \lambda_1, \ldots, \lambda_K \right) $, $ \mathbf{L} = \left( \lambda_1, \ldots, \lambda_K \right)^t $, $ \mathbf{I}_K $ be the $ K \times K $ identity matrix and $ \mathbf{1}_K $ be the column vector of length $ K $ with all elements equal to 1.
Then,
$ \mathbf{U}_0 = E\left[ \mathbf{s}^{(1)}\left( \mathbf{X} - \mathbf{X}' \right) \right]\left( \sqrt{\mathbf{D}} \left( \mathbf{I}_K - \mathbf{1}_K \mathbf{L}^t \right) \boldsymbol{\delta} \right) $.
Note that $ \left( \mathbf{I}_K - \mathbf{1}_K \mathbf{L}^t \right) $ is a $ K \times K $ idempotent matrix with trace $ ( K - 1 ) $ and its null space contains $ \mathbf{1}_K $.
So, for any $ \boldsymbol{\delta} $ such that $ \boldsymbol{\delta}_k \neq \boldsymbol{\delta}_l $ for some $ k $ and $ l $, $ \left( \mathbf{I}_K - \mathbf{1}_K \mathbf{L}^t \right) \boldsymbol{\delta} \neq \mathbf{0} $, which implies
$ \mathbf{U}_0 = E\left[ \mathbf{s}^{(1)}\left( \mathbf{X} - \mathbf{X}' \right) \right]\left( \sqrt{\mathbf{D}} \left( \mathbf{I}_K - \mathbf{1}_K \mathbf{L}^t \right) \boldsymbol{\delta} \right) \neq \mathbf{0} $
since $ \sqrt{\mathbf{D}} $ is non-singular, and $ E\left[ \mathbf{s}^{(1)}\left( \mathbf{X} - \mathbf{X}' \right) \right]\left( \mathbf{h} \right) \neq \mathbf{0} $ for any $ \mathbf{h} \neq \mathbf{0} $ from \autoref{lemma:lemma1}.
\end{proof}

\begin{proof}[Proof of \autoref{thm:5}]
\emph{Proof of (a):}
Since $ E[ \| \mathbf{X} \|^2 ] < \infty $ and $ \sqrt{n_k} n^{-\frac{1}{2}} \to \sqrt{\lambda_k} $ as $ n \to \infty $ for all $ k $, from Theorem 1.1 in \cite{kundu2000central} and Slutsky's Theorem, we get
\begin{align}
& \left( \sqrt{n_1} \left( \bar{\mathbf{X}}_{1 \cdot} - E[ \mathbf{X} ] \right), \ldots, \sqrt{n_K} \left( \bar{\mathbf{X}}_{K \cdot} - E[ \mathbf{X} ] \right) \right) \nonumber\\
& \stackrel{w}{\longrightarrow}
\left( \sqrt{\lambda_1} \boldsymbol{\delta}_1 + \mathbf{Y}_1, \ldots, \sqrt{\lambda_K} \boldsymbol{\delta}_K + \mathbf{Y}_K \right)
\label{thm5eq1}
\end{align}
as $ n \to \infty $.
Define $ \mathbf{g}\left( \cdot \right) ,\, \mathbf{g}_n\left( \cdot \right) : \mathcal{H}^K \to \mathbb{R} $ as
$ \mathbf{g}( \mathbf{x}_1, \ldots, \mathbf{x}_K ) 
= \sum_{k < l} \left\| \mathbf{x}_k - \lambda_l^{-1/2} \lambda_k^{1/2} \mathbf{x}_l \right\|^2 $ and $ \mathbf{g}_n( \mathbf{x}_1, \ldots, \mathbf{x}_K ) =
\sum_{k < l} \left\| \mathbf{x}_k - n_l^{-1/2} n_k^{1/2} \mathbf{x}_l \right\|^2 $ .
Since $ \mathbf{g}\left( \cdot \right) $ is a continuous function and for every $ \mathbf{x}_1, \ldots, \mathbf{x}_K $, $ \left| \mathbf{g}_n( \mathbf{x}_1, \ldots, \mathbf{x}_K ) - \mathbf{g}( \mathbf{x}_1, \ldots, \mathbf{x}_K ) \right| \longrightarrow 0 $ as $ n \to \infty $, from \eqref{thm5eq1}, Theorem 2.7 in \cite[p.~21]{billingsley2013convergence} and an application of Slutsky's Theorem, we get
\begin{align*}
CFF_n
& = \mathbf{g}_n\left( \sqrt{n_1} \left( \bar{\mathbf{X}}_{1 \cdot} - E[ \mathbf{X} ] \right), \ldots, \sqrt{n_K} \left( \bar{\mathbf{X}}_{K \cdot} - E[ \mathbf{X} ] \right) \right)
\\
& \stackrel{w}{\longrightarrow}
\mathbf{g}\left( \sqrt{\lambda_1} \boldsymbol{\delta}_1 + \mathbf{Y}_1, \ldots, \sqrt{\lambda_K} \boldsymbol{\delta}_K + \mathbf{Y}_K \right) \\
& \quad\quad= \sum_{k < l} \left\| \sqrt{\lambda_k} \boldsymbol{\delta}_k + \mathbf{Y}_k - \lambda_l^{-1/2} \lambda_k^{1/2} \left( \sqrt{\lambda_l} \boldsymbol{\delta}_l + \mathbf{Y}_l \right) \right\|^2
\end{align*}
as $ n \to \infty $.
Since under $ \text{H}_0 $ described in \eqref{h_0} in the main paper, $ CFF_n \stackrel{w}{\longrightarrow} \sum_{k < l} \left\| \mathbf{Y}_k - \lambda_l^{-1/2} \lambda_k^{1/2} \mathbf{Y}_l \right\|^2 $ as $ n \to \infty $, the proof of part \emph{(a)} of the theorem follows.

\emph{Proof of (b):}
From the arguments in \cite[p.~151]{zhang2013analysis}, we get that under the null hypothesis $ \text{H}_0 $ in \eqref{h_0} in the main paper,
\begin{align}
ZC_n \stackrel{w}{\longrightarrow} \sum_{i=1}^{\infty} \gamma_i \boldsymbol{\chi}_i^2
\quad\text{as } n \to \infty .
\label{thm6eq1}
\end{align}

Let $ \mathbf{Z}_n = \left( \sqrt{n_1} \left( \bar{\mathbf{X}}_{1 \cdot} - E[ \mathbf{X} ] \right), \ldots, \sqrt{n_K} \left( \bar{\mathbf{X}}_{K \cdot} - E[ \mathbf{X} ] \right) \right)^t $.
Since $ E[ \| \mathbf{X} \|^2 ] < \infty $ and $ \sqrt{n_k} n^{-\frac{1}{2}} \to \sqrt{\lambda_k} $ for all $ k $ as $ n \to \infty $, from Theorem 1.1 in \cite{kundu2000central} and Slutsky's Theorem, we get that under the shrinking alternatives described in \eqref{asympower1} in the main paper,
$ \mathbf{Z}_n \stackrel{w}{\longrightarrow} \mathbf{Z} $ as $ n \to \infty $.
Let $ \mathbf{p}_n = \left( \sqrt{n_1}, \ldots, \sqrt{n_K} \right)^t $.
We have
$ \left( \mathbf{I}_K - \left( \mathbf{p}_n \mathbf{p}_n^t \right) / n \right) \to \left( \mathbf{I}_K - \mathbf{p}_0 \mathbf{p}_0^t \right) $ as $ n \to \infty $.
Since
$ ZC_n = \sum_{k=1}^{K} n_k \left\| \bar{\mathbf{X}}_{k \cdot} - \bar{\mathbf{X}}_{\cdot \cdot} \right\|^2
= \mathbf{Z}_n^t \left( \mathbf{I}_K - \left( \mathbf{p}_n \mathbf{p}_n^t \right) / n \right) \mathbf{Z}_n $, from Theorem 2.7 in \cite[p.~21]{billingsley2013convergence} and an application of Slutsky's theorem, we get that under the shrinking alternatives described in \eqref{asympower1},
\begin{align}
ZC_n \stackrel{w}{\longrightarrow}
\mathbf{Z}^t \left( \mathbf{I}_K - \mathbf{p}_0 \mathbf{p}_0^t \right) \mathbf{Z}
\quad\text{as } n \to \infty .
\label{thm6eq2}
\end{align}
The proof follows from \eqref{thm6eq1} and \eqref{thm6eq2}.

\emph{Proof of (c):}
Under the condition of the theorem, from \cite{kundu2000central} we get that there exist independent Gaussian processes $ \mathbf{W}_k $ with mean $ \mathbf{0} $ and covariance operator $ \boldsymbol{\Omega}_{k} $, $ k = 1, \ldots, K $, such that
\begin{align*}
\max_{1 \le k \le K} \left\| n_k^{-1/2} \sum_{i = 1}^{n_k} \left[ \mathbf{X}_{k i} - \boldsymbol{\mu}_k \right] - \mathbf{W}_k \right\| \stackrel{P}{\longrightarrow} 0
\end{align*}
as $ n \to \infty $.
Let $ \boldsymbol{\phi}_{1}, \ldots, \boldsymbol{\phi}_{d} $ be the eigenfunctions corresponding to the $ d $ largest eigenvalues of $ \boldsymbol{\Omega} $.
%%%%%%%%%%%%%%%%%%%%%%%%%%%%%%%%%%%%%
Define $ \mathbf{Z}_k = \mathbf{W}_k - \sqrt{\lambda_k} \boldsymbol{\delta}_k $ for $ k = 1, \ldots, K $.
Define
\begin{align*}
& \mathbf{z}_k = \left( \langle \mathbf{Z}_k, \boldsymbol{\phi}_1 \rangle, \ldots, \langle \mathbf{Z}_k, \boldsymbol{\phi}_K \rangle \right)^t ,
\quad
\mathbf{z}_{\cdot\cdot} = \left( \sum_{l=1}^{K} \lambda_l \boldsymbol{\Psi}_{l}^{-1} \right)^{-1} \sum_{l=1}^{K} \sqrt{\lambda_l} \boldsymbol{\Psi}_{l}^{-1} \mathbf{z}_l \\
& \text{and}\quad
\mathbf{y}_k = \boldsymbol{\Psi}_{k}^{-1/2} \mathbf{z}_k .
\end{align*}
Note that $ \mathbf{z}_k $'s are independent $ d $-variate normal random vectors with dispersion matrix $ \boldsymbol{\Psi}_{k} $ and mean vector $ \sqrt{\lambda_k} \mathbf{d}_k $. Consequently, $ \mathbf{y}_k $'s are independent $ d $-variate normal random vectors with the identity matrix $ \mathbf{I}_{d} $ as dispersion and mean vector $ \sqrt{\lambda_k} \boldsymbol{\Psi}_{k}^{-1/2} \mathbf{d}_k $.
Using arguments similar to those used in the proof of Theorem 2.1 in \cite{horvath2015introduction}, one can show that
\begin{align}
HR_n
= \sum_{k = 1}^{K} \left( \mathbf{z}_k - \sqrt{\lambda_k} \mathbf{z}_{\cdot\cdot} \right)^t \boldsymbol{\Psi}_{k}^{-1} \left( \mathbf{z}_k - \sqrt{\lambda_k} \mathbf{z}_{\cdot\cdot} \right) + o_P\left( 1 \right)
\quad\text{as } n \to \infty .
\label{eq:hr1}
\end{align}
Now,
\begin{align}
& \sum_{k = 1}^{K} \left( \mathbf{z}_k - \sqrt{\lambda_k} \mathbf{z}_{\cdot\cdot} \right)^t \boldsymbol{\Psi}_{k}^{-1} \left( \mathbf{z}_k - \sqrt{\lambda_k} \mathbf{z}_{\cdot\cdot} \right) 
\nonumber\\
& = \sum_{k = 1}^{K} \mathbf{y}_k^t \mathbf{y}_k - \left( \sum_{k=1}^{K} \sqrt{\lambda_k} \boldsymbol{\Psi}_{k}^{-1/2} \mathbf{y}_k \right)^t \left( \sum_{k=1}^{K} \lambda_k \boldsymbol{\Psi}_{k}^{-1} \right)^{-1} \left( \sum_{k=1}^{K} \sqrt{\lambda_k} \boldsymbol{\Psi}_{k}^{-1/2} \mathbf{y}_k \right) 
\nonumber\\
& = \left[ \mathbf{y}_1^t, \ldots, \mathbf{y}_K^t \right] \left[ \mathbf{y}_1^t, \ldots, \mathbf{y}_K^t \right]^t \nonumber\\
& \quad
- \left[ \mathbf{y}_1^t, \ldots, \mathbf{y}_K^t \right]
\begin{bmatrix}
\sqrt{\lambda_1} \boldsymbol{\Psi}_{1}^{-1/2}\\
\vdots\\
\sqrt{\lambda_K} \boldsymbol{\Psi}_{K}^{-1/2}
\end{bmatrix}
\left( \sum_{k=1}^{K} \lambda_k \boldsymbol{\Psi}_{k}^{-1} \right)^{-1}
\left[ \sqrt{\lambda_1} \boldsymbol{\Psi}_{1}^{-1/2}, \ldots, \sqrt{\lambda_K} \boldsymbol{\Psi}_{K}^{-1/2} \right]
\begin{bmatrix}
\mathbf{y}_1\\
\vdots\\
\mathbf{y}_K
\end{bmatrix}
\nonumber\\
& = \left[ \mathbf{y}_1^t, \ldots, \mathbf{y}_K^t \right] \mathbf{A}_{K d \times K d} \left[ \mathbf{y}_1^t, \ldots, \mathbf{y}_K^t \right]^t .
\label{eq:hr2}
\end{align}
Since $ \mathbf{A}_{K d \times K d} $ is an idempotent matrix of rank $ d ( K - 1 ) $, from \eqref{eq:hr2}, we get that the random variable $ \sum_{k = 1}^{K} \left( \mathbf{z}_k - \sqrt{\lambda_k} \mathbf{z}_{\cdot\cdot} \right)^t \boldsymbol{\Psi}_{k}^{-1} \left( \mathbf{z}_k - \sqrt{\lambda_k} \mathbf{z}_{\cdot\cdot} \right) $ follows a noncentral chi square distribution with $ d (K - 1) $ degrees of freedom and noncentrality parameter $ \big\| \widetilde{\mathbf{M}}_{d ( K - 1 ) \times 1} \big\|^2 $. The proof is complete from \eqref{eq:hr1}.
\end{proof}

Recall that in \autoref{sec:2}, we have estimated the covariance operator $ \boldsymbol{\Sigma} $ in \autoref{thm:1} and \autoref{coro:1} by $ \widehat{\boldsymbol{\Sigma}}_n $.
To generate observations from $ G\left( \mathbf{0}, \widehat{\boldsymbol{\Sigma}}_n \right) $, the estimate $ \widehat{\boldsymbol{\Sigma}}_n $ needs to be a non-negative definite operator. The non-negative definiteness of $ \widehat{\boldsymbol{\Sigma}}_n $ is established in the following theorem under the condition that each of the $ K $ groups has at least two distinct observations. Clearly, this condition is satisfied almost surely when the underlying distributions of the groups are non-atomic.
\begin{theorem} \label{thm:covnnd}
The operator $ \widehat{\boldsymbol{\Sigma}}_n $ is non-negative definite whenever each of the $ K $ groups has at least two distinct observations.
\end{theorem}
\begin{proof}
The non-negative definiteness of $ \widehat{\boldsymbol{\Sigma}}_n $ is established subject to the condition that each of the $ K $ groups, $ \{ \mathbf{X}_{k i} \midil i = 1, \ldots, n_k \} $, where $ k = 1, \ldots, K $, has at least two distinct observations. In that case, for every group of observations $ \{ \mathbf{X}_{k i} | i = 1, \ldots, n_k \} $, we can find a real-valued function $ F_k( \cdot ) $, such that
\begin{align}\label{eq:nnd1}
\frac{1}{n_k} \sum_{i=1}^{n_k} \left\{ F_k\left( \mathbf{X}_{k i} \right) - \frac{1}{n_k} \sum_{j=1}^{n_k} F_k\left( \mathbf{X}_{k j} \right) \right\}^2 = 1 .
\end{align}
For example, suppose $ \mathbf{X}_{k i_1} $ and $ \mathbf{X}_{k i_2} $ be distinct observations in the $ k $th group. Then the sample variance of the projections on $ ( \mathbf{X}_{k i_1} - \mathbf{X}_{k i_2} ) $ is positive, i.e.,
\begin{align*}
\frac{1}{n_k} \sum_{i=1}^{n_k} \left\{ \left\langle \mathbf{X}_{k i_1} - \mathbf{X}_{k i_2}, \mathbf{X}_{k i} \right\rangle - \frac{1}{n_k} \sum_{j=1}^{n_k} \left\langle \mathbf{X}_{k i_1} - \mathbf{X}_{k i_2}, \mathbf{X}_{k j} \right\rangle \right\}^2 > 0 ,
\end{align*}
and one can define
\begin{align*}
& F_k( \mathbf{x} ) \\
&= \left[ \frac{1}{n_k} \sum_{i=1}^{n_k} \left\{ \left\langle \mathbf{X}_{k i_1} - \mathbf{X}_{k i_2}, \mathbf{X}_{k i} \right\rangle - \frac{1}{n_k} \sum_{j=1}^{n_k} \left\langle \mathbf{X}_{k i_1} - \mathbf{X}_{k i_2}, \mathbf{X}_{k j} \right\rangle \right\}^2 \right]^{-\frac{1}{2}} \left\langle \mathbf{X}_{k i_1} - \mathbf{X}_{k i_2}, \mathbf{x} \right\rangle .
\end{align*}
We use these functions and the fact that every covariance operator is non-negative definite so establish the non-negative definiteness of $ \widehat{\boldsymbol{\Sigma}}_n $.

Define
\begin{align*}
& \mathbf{D}_n^{(1)}( i, j, k ) \\
& =
\frac{1}{n_k} \sum_{l_k = 1}^{n_k} \left[ \left( \frac{1}{n_{i}} \sum_{l_{i} = 1}^{n_{i}} \sqrt{\frac{n_k}{n_k - 1}} \mathbf{s}\left( \mathbf{X}_{i l_{i}} - \mathbf{X}_{k l_k} \right) \right) \otimes \left( \frac{1}{n_{j}} \sum_{l_{j} = 1}^{n_{j}} \sqrt{\frac{n_k}{n_k - 1}} \mathbf{s}\left( \mathbf{X}_{j l_{j}} - \mathbf{X}_{k l_k} \right) \right) \right] \nonumber\\
& \quad
- \left( \frac{1}{n_{i} n_k} \sum_{l_{i} = 1}^{n_{i}} \sum_{l_k = 1}^{n_k} \sqrt{\frac{n_k}{n_k - 1}} \mathbf{s}\left( \mathbf{X}_{i l_{i}} - \mathbf{X}_{k l_k} \right) \right) \otimes \left( \frac{1}{n_{j} n_k} \sum_{l_{j} = 1}^{n_{j}} \sum_{l_k = 1}^{n_k} \sqrt{\frac{n_k}{n_k - 1}} \mathbf{s}\left( \mathbf{X}_{j l_{j}} - \mathbf{X}_{k l_k} \right) \right) \\
\intertext{and}
& \mathbf{D}_n^{(2)}( i, j, k ) \\
& =
\left( \frac{1}{n_{i} n_k} \sum_{l_{i} = 1}^{n_{i}} \sum_{l_k = 1}^{n_k} \frac{\mathbf{s}\left( \mathbf{X}_{i l_{i}} - \mathbf{X}_{k l_k} \right)}{\sqrt{n_k - 1}} \right) \otimes \left( \frac{1}{n_{j} n_k} \sum_{l_{j} = 1}^{n_{j}} \sum_{l_k = 1}^{n_k} \frac{\mathbf{s}\left( \mathbf{X}_{j l_{j}} - \mathbf{X}_{k l_k} \right)}{\sqrt{n_k - 1}} \right) ,
\end{align*}
where $ i, j, k \in \{ 1, \ldots, K \} $. It can be verified that $ \mathbf{C}_n( i, j, k ) = \mathbf{D}_n^{(1)}( i, j, k ) + \mathbf{D}_n^{(2)}( i, j, k ) $ for all $ i, j, k \in \{ 1, \ldots, K \} $. Define the two operators $ \widehat{\boldsymbol{\Sigma}}_n^{(1)} $ and $ \widehat{\boldsymbol{\Sigma}}_n^{(2)} $ as
\begin{align*}
\widehat{\boldsymbol{\Sigma}}_n^{(1)} = \left( \boldsymbol{\sigma}^{(1, n)}_{k_1 k_2} \right)_{K \times K}
\quad\text{and}\quad
\widehat{\boldsymbol{\Sigma}}_n^{(2)} = \left( \boldsymbol{\sigma}^{(2, n)}_{k_1 k_2} \right)_{K \times K} ,
\end{align*}
where
\begin{align*}
\boldsymbol{\sigma}^{(1, n)}_{k_1 k_2}
& = \frac{\sqrt{n_{k_1} n_{k_2}}}{n} \sum_{l=1}^{K} \frac{n_l}{n}
\left[ \mathbf{D}_n^{(1)}( k_1, k_2, l ) - \mathbf{D}_n^{(1)}( l, k_2, k_1 ) - \mathbf{D}_n^{(1)}( k_1, l, k_2 ) \right] \\
& \quad
+ \sum_{l_1 = 1}^{K} \sum_{l_2 = 1}^{K} \frac{n_{l_1} n_{l_2}}{n^2} \mathbf{D}_n^{(1)}( l_1, l_2, k_1 ) \mathbb{I}( k_1 = k_2 ) \\
\intertext{and}
\boldsymbol{\sigma}^{(2, n)}_{k_1 k_2}
& = \frac{\sqrt{n_{k_1} n_{k_2}}}{n} \sum_{l=1}^{K} \frac{n_l}{n}
\left[ \mathbf{D}_n^{(2)}( k_1, k_2, l ) - \mathbf{D}_n^{(2)}( l, k_2, k_1 ) - \mathbf{D}_n^{(2)}( k_1, l, k_2 ) \right] \\
& \quad
+ \sum_{l_1 = 1}^{K} \sum_{l_2 = 1}^{K} \frac{n_{l_1} n_{l_2}}{n^2} \mathbf{D}_n^{(2)}( l_1, l_2, k_1 ) \mathbb{I}( k_1 = k_2 ) .
\end{align*}
Then, we have $ \widehat{\boldsymbol{\Sigma}}_n = \widehat{\boldsymbol{\Sigma}}_n^{(1)} + \widehat{\boldsymbol{\Sigma}}_n^{(2)} $. We shall show that both $ \widehat{\boldsymbol{\Sigma}}_n^{(1)} $ and $ \widehat{\boldsymbol{\Sigma}}_n^{(2)} $ are non-negative definite, which would imply that $ \widehat{\boldsymbol{\Sigma}}_n $ is also non-negative definite.

For $ k = 1, \ldots, K $, define the $ K $ independent random elements $ \tilde{\mathbf{X}}_k $ by
\begin{align*}
\tilde{\mathbf{X}}_k = \mathbf{X}_{k i} \;\text{with probability}\; \frac{1}{n_k}, \; i = 1, \ldots, n_k .
\end{align*}
Define
\begin{align*}
\tilde{\mathbf{S}}_k^{(1)}
& =
\sum_{l=1}^{K} \frac{n_l}{n} \left[ \sqrt{\frac{1}{n_k - 1}} \frac{1}{n_l} \sum_{i_l = 1}^{n_l} \mathbf{s}\left( \tilde{\mathbf{X}}_{k} - \mathbf{X}_{l i_l} \right)
+ \sqrt{\frac{1}{n_l - 1}} \frac{1}{n_k} \sum_{i_k = 1}^{n_k} \mathbf{s}\left( \mathbf{X}_{k i_k} - \tilde{\mathbf{X}}_{l} \right) \right] \\
\intertext{and}
\tilde{\mathbf{S}}_k^{(2)}
& =
\sum_{l=1}^{K} \frac{n_l}{n} \left\{ \frac{1}{n_k n_l} \sum_{i_k = 1}^{n_k} \sum_{i_l = 1}^{n_l} \mathbf{s}\left( \mathbf{X}_{k i_k} - \mathbf{X}_{l i_l} \right) \right\}
\left[ \frac{F_l\left( \tilde{\mathbf{X}}_l \right)}{\sqrt{n_l \left( n_l - 1 \right)}}
+ \frac{F_k\left( \tilde{\mathbf{X}}_k \right)}{\sqrt{n_k \left( n_k - 1 \right)}} \right] .
\end{align*}
Also define
\begin{align*}
\tilde{\mathbf{W}}_n^{(1)} = \left( \sqrt{n_1} \tilde{\mathbf{S}}_1^{(1)}, \ldots, \sqrt{n_K} \tilde{\mathbf{S}}_K^{(1)} \right)
\quad\text{and}\quad
\tilde{\mathbf{W}}_n^{(2)} = \left( \sqrt{n_1} \tilde{\mathbf{S}}_1^{(2)}, \ldots, \sqrt{n_K} \tilde{\mathbf{S}}_K^{(2)} \right) .
\end{align*}
Let $ \widetilde{\text{Cov}}( \cdot, \cdot ) $ denote the covariance between functions of $ \tilde{\mathbf{X}}_k $ with respect to the distributions of the $ K $ random elements $ \tilde{\mathbf{X}}_k $, $ k = 1, \ldots, K $, conditional on all the observations $ \mathbf{X}_{k i} $. Then, using \eqref{eq:nnd1}, it can be verified that
\begin{align*}
& \widetilde{\text{Cov}}\left( \sqrt{n_{k_1}} \tilde{\mathbf{S}}_{k_1}^{(1)}, \sqrt{n_{k_2}} \tilde{\mathbf{S}}_{k_2}^{(1)} \right)
= \boldsymbol{\sigma}^{(1, n)}_{k_1 k_2}
\quad\text{and}\quad
\widetilde{\text{Cov}}\left( \sqrt{n_{k_1}} \tilde{\mathbf{S}}_{k_1}^{(2)}, \sqrt{n_{k_2}} \tilde{\mathbf{S}}_{k_2}^{(2)} \right)
= \boldsymbol{\sigma}^{(2, n)}_{k_1 k_2}
\end{align*}
for all $ k_1, k_2 $. This implies that
\begin{align*}
\widetilde{\text{Cov}}\left( \tilde{\mathbf{W}}_n^{(1)}, \tilde{\mathbf{W}}_n^{(1)} \right)
= \widehat{\boldsymbol{\Sigma}}_n^{(1)}
\quad\text{and}\quad
\widetilde{\text{Cov}}\left( \tilde{\mathbf{W}}_n^{(2)}, \tilde{\mathbf{W}}_n^{(2)} \right)
= \widehat{\boldsymbol{\Sigma}}_n^{(2)} ,
\end{align*}
which in turn imply that $ \widehat{\boldsymbol{\Sigma}}_n^{(1)} $ and $ \widehat{\boldsymbol{\Sigma}}_n^{(2)} $ are covariance operators, and hence non-negative definite. Therefore, $ \widehat{\boldsymbol{\Sigma}}_n $ is also non-negative definite whenever each of the groups has at least two distinct observations.
\end{proof}

\section{Comparison of asymptotic and permutation implementations of the SS test} \label{asym_perm}
Here, we present the level and the power comparisons of the asymptotic and the permutation implementations of the SS test.
\begin{table}[h]
\begin{center}
\caption{Estimated sizes of the asymptotic and the permutation implementations of the SS test under nominal level 5\%; $ n_1 = n_2 = n_3 = 5, 10, 20 \text{ and } 40 $.}
\label{table:asymboot1}
\begin{tabular} {l|cc|cc}  
\hline 
& \multicolumn{2}{c|}{$ n_1 = n_2 = n_3 = 5 $} & \multicolumn{2}{c}{$ n_1 = n_2 = n_3 = 10 $} \\\hline
$ P_0 $ & Asymptotic & Permutation & Asymptotic & Permutation \\
\hline 
SBM & 0.050 & 0.055 & 0.054 & 0.051 \\
$ t_{1, \text{SBM}} $ & 0.047 & 0.047 & 0.052 & 0.045 \\
$ t_{3, \text{SBM}} $ & 0.058 & 0.057 & 0.050 & 0.050 \\
SBM with impurity & 0.046 & 0.047 & 0.053 & 0.047 \\
$ t_{1, \text{SBM}} $ with impurity & 0.046 & 0.047 & 0.048 & 0.047 \\
$ t_{3, \text{SBM}} $ with impurity & 0.058 & 0.057 & 0.055 & 0.050 \\
GBM & 0.037 & 0.041 & 0.055 & 0.049 \\
$ ( \text{SBM} )^2 $ & 0.036 & 0.047 & 0.053 & 0.057 \\
$ ( t_{3, \text{SBM}} )^2 $ & 0.034 & 0.053 & 0.042 & 0.047 \\
\hline
\hline 
& \multicolumn{2}{c|}{$ n_1 = n_2 = n_3 = 20 $} & \multicolumn{2}{c}{$ n_1 = n_2 = n_3 = 40 $} \\\hline
$ P_0 $ & Asymptotic & Permutation & Asymptotic & Permutation \\
\hline 
SBM & 0.057 & 0.043 & 0.051 & 0.053 \\
$ t_{1, \text{SBM}} $ & 0.055 & 0.053 & 0.057 & 0.052 \\
$ t_{3, \text{SBM}} $ & 0.046 & 0.043 & 0.054 & 0.046 \\
SBM with impurity & 0.055 & 0.047 & 0.049 & 0.044 \\
$ t_{1, \text{SBM}} $ with impurity & 0.055 & 0.042 & 0.059 & 0.047 \\
$ t_{3, \text{SBM}} $ with impurity & 0.054 & 0.050 & 0.045 & 0.041 \\
GBM & 0.052 & 0.044 & 0.050 & 0.046 \\
$ ( \text{SBM} )^2 $ & 0.056 & 0.051 & 0.052 & 0.047 \\
$ ( t_{3, \text{SBM}} )^2 $ & 0.047 & 0.049 & 0.053 & 0.049 \\
\hline 
\end{tabular}
\end{center}
\end{table}
We present the estimated sizes and the power curves of the two implementations at nominal level 5\% in the models considered in \autoref{simulation}. For the size study, we consider 3 groups and the underlying distribution $ P_0 $ of the groups are as in \autoref{simulation} with 4 cases of the group sizes: $ n_1 = n_2 = n_3 = 5, 10, 20 \text{ and } 40 $. For the power study, we consider the sizes of the 3 groups as 10, 10 and 10, and the underlying distribution $ P_0 $ of the groups are as in \autoref{simulation}. The shifts of the 3 groups,
$ \boldsymbol{\mu}_1 $, $ \boldsymbol{\mu}_2 $ and $ \boldsymbol{\mu}_3 $, are kept the same as in \autoref{simulation}. The data generation process is also the same.
The number of random permutations taken is 1000, and the sizes and the powers of the test procedures are estimated using 1000 independent replications.
\begin{figure}
	\centering
	\includegraphics[width=1\linewidth]{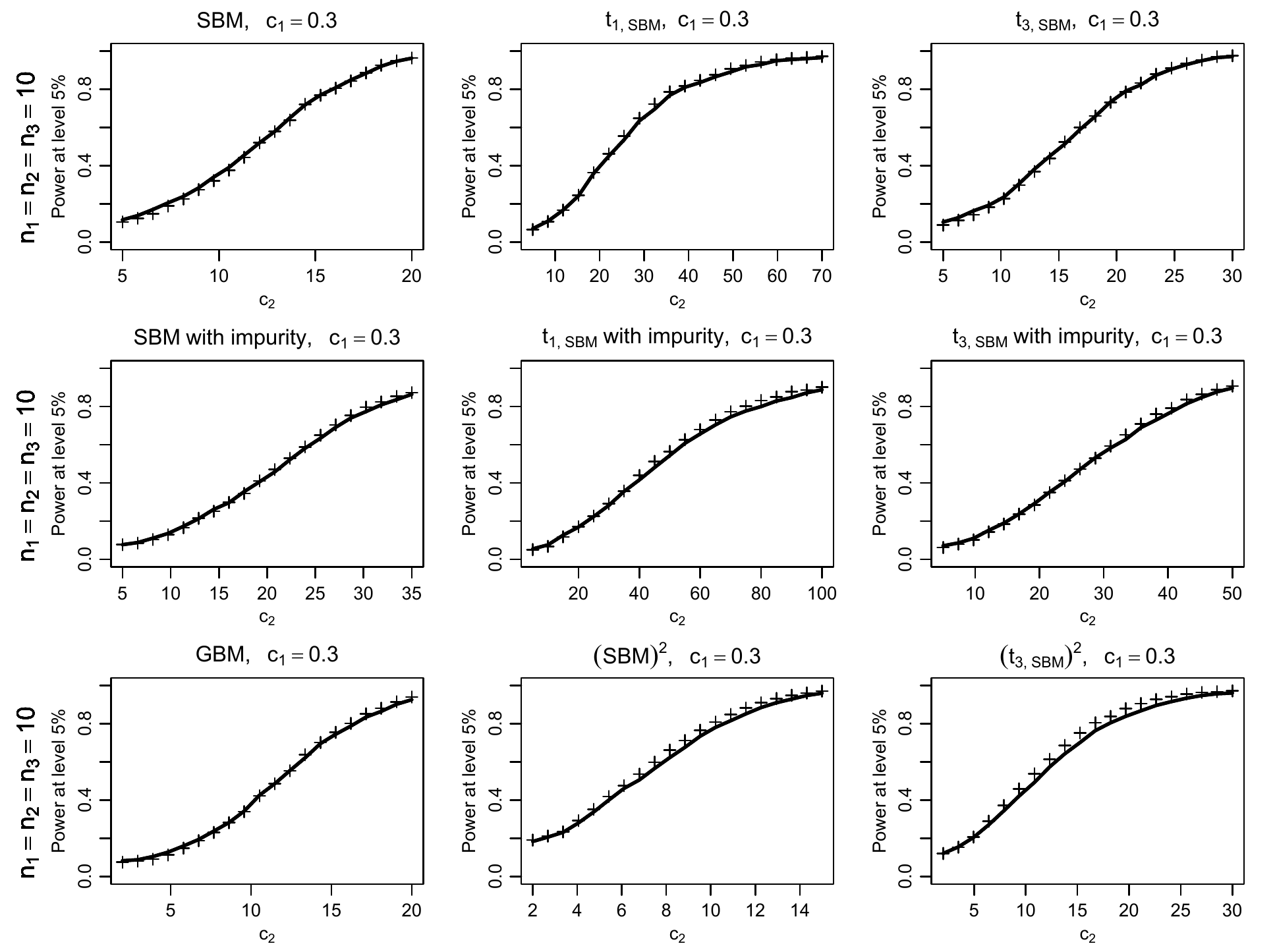}
	\caption{Estimated power curves of the asymptotic implementation ($ \boldsymbol{-} $) and the permutation implementation (++) of the SS test under nominal level 5\%.}
	\label{fig:anovaasympermcomp}
\end{figure}

The estimated sizes of the two implementations are presented in \autoref{table:asymboot1}.
The estimated power curves are presented in \autoref{fig:anovaasympermcomp}.
We note that the powers of the asymptotic and the permutation implementations are virtually identical. But the size of the asymptotic implementation is for a few cases of the underlying distribution somewhat lower than the nominal level of 5\% when the group sizes are small, i.e., 5. The estimated sizes of the permutation implementation are satisfactory for both small and larger sample sizes.

\section*{R Codes}
The R functions to compute the SS test are available in \url{https://github.com/joydeepchowdhury/spatialsign_functional_anova}.

\end{appendix}

% \begin{acks}[Acknowledgments]
% \end{acks}

%% if your bibliography is in bibtex format, uncomment commands:
\bibliographystyle{imsart-number} % Style BST file (imsart-number.bst or imsart-nameyear.bst)
\bibliography{bibliography_database}       % Bibliography file (usually '*.bib')

%% or include bibliography directly:
% \begin{thebibliography}{}
% \bibitem{b1}
% \end{thebibliography}

\end{document}